\documentclass[10pt,journal,compsoc]{IEEEtran}

\usepackage[utf8]{inputenc} % allow utf-8 input
\usepackage[T1]{fontenc}    % use 8-bit T1 fonts
\usepackage{hyperref}       % hyperlinks
\usepackage{url}            % simple URL typesetting
\usepackage{booktabs}       % professional-quality tables
\usepackage{amsfonts}       % blackboard math symbols
\usepackage{nicefrac}       % compact symbols for 1/2, etc.
\usepackage{microtype}      % microtypography
\usepackage{xcolor}         % colors
\usepackage{amsthm}
\usepackage{amsmath}
\usepackage{amssymb}
\usepackage{float}
\usepackage{graphics}
\usepackage{graphicx}
\usepackage{caption}
\usepackage{tablefootnote}
\usepackage{algorithm}
\usepackage{algpseudocode}
\usepackage{bm}
\usepackage{color}
\usepackage{multirow}
\usepackage{makecell}
\usepackage{balance}
\usepackage{diagbox}
\usepackage{stfloats}
\usepackage{bbding}
\usepackage{subfigure}
\usepackage{enumitem}
\usepackage{apptools}

\allowdisplaybreaks

\newtheorem{assumption}{Assumption}
\newtheorem{thm}{Theorem}
\newtheorem{lem}{Lemma}
\newtheorem{defi}{Definition}
\newtheorem{prop}{Proposition}
\newtheorem{rmk}{Remark}
\AtAppendix{\counterwithin{thm}{section}}
\AtAppendix{\counterwithin{lem}{section}}
\AtAppendix{\counterwithin{defi}{section}}
\AtAppendix{\counterwithin{prop}{section}}
\AtAppendix{\counterwithin{rmk}{section}}
\AtAppendix{\counterwithin{assumption}{section}}

\newcommand{\ka}{\kappa}
\newcommand{\calD}{\mathcal{D}}
\newcommand{\calO}{\tilde{\mathcal{O}}}

\newcommand{\calI}{\mathcal{I}}

\newcommand{\calN}{\mathcal{N}}

\newcommand{\bw}{\mathbf{w}}
\newcommand{\bg}{\mathbf{g}}

\newcommand{\bbf}{\mathbf{f}}
\newcommand{\EE}{\mathbb{E}}
\newcommand{\RR}{\mathbb{R}}
\newcommand{\bx}{\mathbf{x}}

\newcommand{\cc}{\CheckmarkBold}
\newcommand{\xx}{\XSolidBold}
\newcommand{\vv}{\vartriangle}

\newcommand{\RomanNumeralCaps}[1]
    {\MakeUppercase{\romannumeral #1}}
\algnewcommand\algorithmicforpara{\textbf{for}}
\algnewcommand\algorithmicdoinparallel{\textbf{do in parallel}}
\algdef{S}[FOR]{ForParallel}[1]{\algorithmicforpara\ #1\ \algorithmicdoinparallel}
\DeclareMathOperator*{\argmin}{arg\,min}

\setlist[itemize]{leftmargin=*}

\title{FedCut: A Spectral Analysis Framework for Reliable Detection of Byzantine Colluders}

\author{Hanlin~Gu,
        Lixin~Fan,~\IEEEmembership{Member,~IEEE,}
        XingXing~Tang,
        and~Qiang~Yang,~\IEEEmembership{Fellow,~IEEE}
\IEEEcompsocitemizethanks{
\IEEEcompsocthanksitem Hanlin Gu and Lixin Fan are with WeBank AI Lab, WeBank, China. E-mail: \{allengu, lixinfan\}@webank.com, \{ghltsl123, Lixin.Fan01\}@gmail.com. 
\IEEEcompsocthanksitem Xingxing Tang is with the Department of Computer Science and Engineering, Hong
Kong University of Science and Technology, Hong Kon. E-mail: xtangav@connect.ust.hk 
\IEEEcompsocthanksitem  Qiang Yang is with the Department of Computer Science and Engineering, Hong
Kong University of Science and Technology, Hong Kong and  WeBank AI Lab, WeBank, China. E-mail: qyang@cse.ust.hk.}
\thanks {Corresponding author: Lixin Fan.}}

\begin{document}

\IEEEtitleabstractindextext{
\begin{abstract}

This paper proposes a general \textit{spectral analysis} framework that thwarts a security risk in federated Learning caused by \textit{groups of malicious Byzantine attackers} or \textit{colluders}, who conspire to upload vicious model updates to severely debase global model performances. The proposed framework delineates the strong consistency and temporal coherence between Byzantine colluders’ model updates from a spectral analysis lens, and, formulates the detection of Byzantine misbehaviours as a \textit{community detection} problem in weighted graphs. The modified \text{normalized graph cut} is then utilized to discern attackers from benign participants. Moreover, the \textit{Spectral heuristics} is adopted to make the detection robust against various attacks. The proposed Byzantine colluder resilient method, i.e., FedCut, is guaranteed to converge with bounded errors. Extensive experimental results under a variety of settings justify the superiority of FedCut, which demonstrates extremely robust model performance (MP) under various attacks. It was shown that FedCut's averaged MP is 2.1\% to 16.5\% better than that of the state of the art Byzantine-resilient methods. In terms of the worst-case model performance (MP), FedCut is 17.6\% to 69.5\% better than these methods.
\end{abstract}

}

\maketitle

\IEEEdisplaynontitleabstractindextext

\IEEEpeerreviewmaketitle

\IEEEraisesectionheading{\section{Introduction}}

%\LF
\IEEEPARstart{F}{ederated} learning (FL) \cite{mcmahan2017communication,yang2019federated} is a suite of privacy-preserving machine learning techniques that allow multiple parties to {collaboratively train a global model}, yet, {without gathering or exchanging \textit{privacy} for the sake of compliance with private data protection regulation rules such as GDPR\footnote{GDPR is applicable as of May 25th, 2018 in all European member states to harmonize data privacy laws across Europe. https://gdpr.eu/}}. It is the \textit{upgraded global model performance} that motivates multiple parties to join FL, however, the risk of potential malicious attacks aiming to degrade FL model performances cannot be discounted.
%defeats the purpose FL in the first place.  
It was shown that even a \textit{single} attacker (aka a Byzantine worker) may prevent the convergence of a naive FL aggregation rule by submitting vicious model updates to outweigh benign workers (see Lemma 1 of \cite{blanchard2017machine}). A great deal of research effort was then devoted to developing numerous Byzantine-resilient methods which can effectively detect and attenuate such misbehaviours \cite{blanchard2017machine, yin2018byzantine, pillutla2019robust, chen2017distributed, xie2018generalized, yang2019byzantine, guerraoui2018hidden}. Nevertheless, recent research pointed out that \textit{a group of attackers} or \textit{colluders} may conspire to cause more damages than these Byzantine-resilient methods can deal with (see \cite{baruch2019little,xie2020fall,fang2020local}). %for reports and Sect. \ref{sec:failure} for our investigations). 
It is the misbehavior of such Byzantine \textit{colluders} that motivate our research to analyze their influences on %benign model updates
global model performances from a spectral analysis framework, and based on our findings, an effective defense algorithm to do away with Byzantine colluders is proposed.

% The main challenge brought by Byzantine colluders lies in the fact that they may conspire to \textit{misbehavior consistently} and introduce \textit{statistical bias} which often break down \textit{Robust Statistic} (RS) based resilient methods (e.g. \cite{baruch2019little,fang2020local,xie2020fall}) and, consequently, the global model performances are degraded significantly (see failure cases in Sect. \ref{sec:failure} and empirical results in Sect. \ref{sec:experiment}). This challenging situation is mitigated by \textit{clustering} based resilient methods (\cite{shen2016auror, sattler2020byzantine, ghosh2019robust} which group model updates from all clients into \textit{two clusters} according to both the inter-cluster and intra-cluster similarities. Consistent misbehaviours are therefore more easily distinguished from benign model updates. However, there are two open issues with clustering based methods. First, the adopted naive clustering, Kmeans, Kmedian \cite{jain1988algorithms,lloyd1982least} are not efficient and easily trapped into local minima. Second, they invariably assumed that all malicious model updates form one group while benign ones form the other. The detection of Byzantine attackers thus can be evaded by colluders who purposely form multiple groups to violate such an assumption (see Sect. \ref{sec:failure} for illustrative examples). 

There are two main challenges brought by Byzantine colluders. First, colluders may conspire to \textit{misbehave consistently} and introduce \textit{statistical bias} to break down \textit{Robust Statistics} (RS) based resilient methods (e.g., \cite{baruch2019little,fang2020local,xie2020fall}). Consequently, the global model performances may deteriorate significantly. 
% empirically as demonstrated in Sect. \ref{sec:failure} and Sect. \ref{sec:experiment} of this article. 
%For example, colluders may mimic the behaviours of one benign client $i^*$ towards over-emphasizing client $i^*$ and thus under-representing other clients \cite{fang2020local,karimireddy2020byzantine} to introduce the statistical bias. 
Second, Byzantine colluders may conspire to violate the assumption that all malicious model updates form one group while benign ones form the other, which is invariably assumed by most clustering-based Byzantine-resilient methods \cite{shen2016auror, sattler2020byzantine, ghosh2019robust}. By submitting multiple groups of such detrimental yet disguised model updates, colluders therefore can evade clustering-based methods and degrade the global model performance significantly. 
% (see Sect. \ref{sec:failure} for detailed analysis).

% clustering-based methods \cite{shen2016auror, sattler2020byzantine, ghosh2019robust} group model updates from all clients into \textit{two clusters} according to both the inter-cluster and intra-cluster similarities. Consistent misbehaviours by colluders are therefore more easily distinguished from benign model updates. However, they invariably assumed that all malicious model updates form one group while benign ones form the other. The detection of Byzantine attackers thus can be evaded by colluders who purposely form multiple groups to violate such an assumption (see Sect. \ref{sec:failure} for illustrative examples). 

\begin{figure*}[htbp]
    \centering
    	\subfigure{
		\begin{minipage}[b]{0.48\textwidth}
			\includegraphics[width=0.9\textwidth]{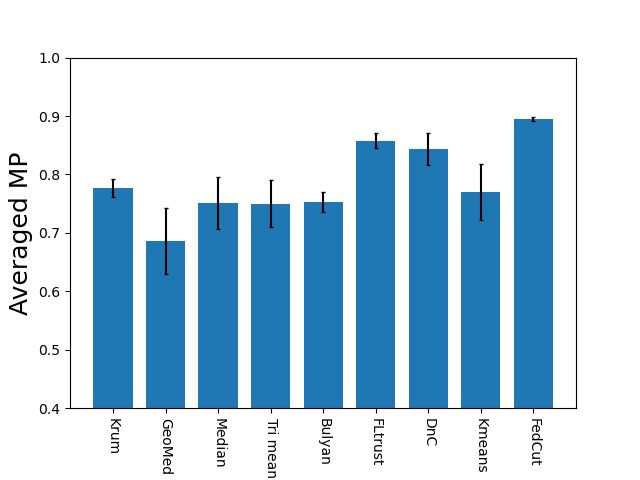}
		\end{minipage}
	}
    	\subfigure{
    		\begin{minipage}[b]{0.48\textwidth}
  		 	\includegraphics[width=0.9\textwidth]{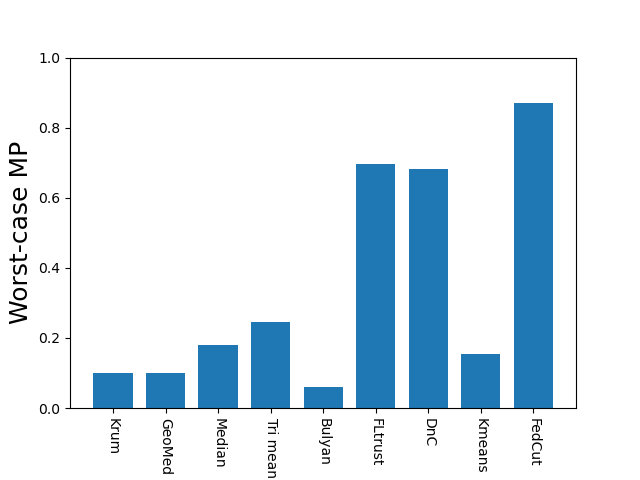}
    		\end{minipage}
    	}
    \caption{Averaged (left) and worst-case (right) model performances under all attacks for different Byzantine-resilient methods (Krum \cite{blanchard2017machine}, GeoMedian \cite{chen2017distributed}, Median, Trimmed mean \cite{yin2018byzantine}, Bulyan \cite{guerraoui2018hidden}, FLtrust \cite{cao2020fltrust}, DnC \cite{shejwalkar2021manipulating}, Kmeans \cite{shen2016auror} and our proposed FedCut) with IID setting and 30 Byzantine clients for classification of Fashion MNIST (FedCut -- the proposed method, see more details in Sect. \ref{subsec:experiment-setup}).}
    \label{fig:overview}
    \vspace{-4pt}
\end{figure*}

In order to address the challenges brought by colluders, we propose in this article a spectral analysis framework, called FedCut, which admits effective detection of Byzantine colluders in a variety of settings and provide the spectral analysis of different types of Byzantine behaviours especially for Byzantine colluders (see Sect. \ref{sec:failure}).
% and provides excellent Byzantine-resilient capabilities measured by Byzantine Tolerant Rate (BTR) (see Def. \ref{def:Byzantine Tolerance}). 
% The proposed FedSpectral framework treats all model updates as nodes and similarities between respective pairs of model updates edge weights, and thus, discerns benign model updates from one or more groups of Byzantine colluders from a graph spectral analysis perspective (see Sect. \ref{subsec:spectral property} for details). 
The essential ingredients of the proposed framework are as follows. First, we build the Spatial-Temporal graph with all clients' model updates over multiple learning iterations as nodes, and similarities between respective pairs of model updates as edge weights (see Sect. \ref{subsec:spa-temp}). Second, an extension of the \textit{normalized cut} (Ncut) \cite{shi2000normalized,ng2001spectral} provides the optimal c-partition Ncut (see Sect. \ref{subsec:fedcut}), which allows to detect colluders with consistent behaviour efficiently. Third, \textit{spectral heuristics} \cite{zelnik2004self} are used to determine the type of Byzantine attackers, the unknown number of colluder groups and the appropriate scaling factor $\sigma$ of Gaussian kernels used for measuring similarities between model updates (see Sect. \ref{subsec: spectal-heu}). By leveraging the aforementioned techniques together within the unified spectral analysis framework, we thus propose in Sect. \ref{sec:framework} the FedCut method which demonstrates superior robustness in the presence of different types of colluder attacks. Its convergence is theoretically proved in Sect. \ref{sec:converg} and it compares favorably to the state of arts Byzantine-resilient methods in thorough empirical evaluations of model performances in Sect. \ref{sec:experiment}.
% In the face of a variety of colluder attacks, this spectral analysis framework demonstrates superior robustness in terms of both the theoretical guarantee (Sect. \ref{sec:converg}) and empirical results (Sect. \ref{sec:experiment}). %Finally, our empirical finding shows that benign model updates form a "small-world" sub-network while malicious updates with collusion form a multi-party strongly connected network (see Appendix \ref{sec:appendixc} and \cite{watts1998collective} for the notion of \textit{small-world network}). 
Our contributions are three folds.

\begin{itemize}
     \item First, we provide the spectral analysis of Byzantine attacks, especially, those launched by colluders. Specifically, we formulate existing Byzantine attacks as four types and gain a deeper understanding of colluders' attacks in the lens of spectral. Moreover, we delineate root causes of failure cases of existing Byzantine-resilient methods in the face of colluder attacks (see Sect. \ref{subsec:fail-case}).  
     \item Second, we propose the spectral analysis framework, called FedCut, which distinguishes benign clients from multiple groups of Byzantine colluders. Specifically, the \textit{normalized cut}, \textit{temporal consistency} and the \textit{spectral heuristics} are adopted to address challenges brought by colluders. Moreover, we provide the theoretical analysis of convergence of the proposed algorithm FedCut.
    \item Finally, we propose to thoroughly investigate both \textit{averaged} and \textit{worst-case} model performances of different Byzantine-resilient methods with extensive experiments under a variety of settings, including different types of models, datasets, extents of Non-IID, the fraction of attackers, combinations of colluders groups. It was then demonstrated that the proposed FedCut consistently outperforms existing methods under all different settings (see Fig. \ref{fig:overview} and Tab. \ref{tab:ref}). 
\end{itemize}

\section{Related work}\label{sec:related}

Related work are abundant in the literature and we briefly review them below. \\
\textbf{Byzantine attacks} can be broadly categorized as Non-Collusion and Collusion attacks. The former attacks were proposed to degrade global model performances by uploading Gaussian noise or flipping gradients \cite{blanchard2017machine, li2019rsa}. The latter type launched consistent attacks to induce misclassifications of Byzantine attackers \cite{baruch2019little, fang2020local,xie2020fall}. \\
% \textbf{Byzantine Attack} has been designed for distributed training. Several attacks were proposed such as uploading same value or Gaussian noise to destroy the model performance \cite{blanchard2017machine, li2019rsa}, where adversaries are independent of each other. Moreover, some strong attacks are designed \cite{baruch2019little, fang2020local,xie2020fall}, in which adversaries collude with each other and take actions consistently to induce the server to misclassify benign clients.\\
\textbf{Robust statistics based aggregation} approaches treated malicious model updates as \textit{outliers}, which are far away from benign clients, and filter out outliers via robust statistics accordingly. For example, the coordinate-wise Median and some variants of Median, such as geometric Median, were proposed to remove outliers \cite{yin2018byzantine,pillutla2019robust,xie2018generalized}. Moreover, Blanchard et al. assumed few outliers were far away from benign updates. Therefore, they used the Krum function to select model updates that are very close to at least half the other updates \cite{blanchard2017machine} as benign updates. 
However, the aforementioned methods are vulnerable to attacks with collusion which may conspire to induce \textit{biased} estimations \cite{baruch2019little, fang2020local,xie2020fall}. 
A different line of approaches \cite{shejwalkar2021manipulating,diakonikolas2019robust,data2021byzantine} used the concentration filter to remove the outliers which are far away from concentration (such as the mean) \cite{shejwalkar2021manipulating,diakonikolas2019robust,data2021byzantine}. For instance, Shejwalkar et al. applied SVD to the covariance matrix of model updates and filtered out outliers that largely deviate from the empirical mean of model updates towards the principle direction \cite{shejwalkar2021manipulating}. However, the effectiveness of the concentration filter may deteriorate in the presence of collusion attacks (see Sect. \ref{subsec:fail-case}), and moreover, the time complexity of this approach is high $\tilde{O}(d^3)$ when the dimension of updates $d$ is large.\\
\textbf{Clustering based robust aggregation} grouped all clients into two clusters according to pairwise similarities between clients' updates and regarded the cluster with a smaller number as the Byzantine cluster. For instance, some methods applied \textit{Kmeans} into clustering benign and byzantine clients \cite{shen2016auror, bulusu2021byzantine} and Sattler et al. separated benign and byzantine clients by minimizing \textit{cosine similarity} of updates among two groups \cite{sattler2020byzantine, sattler2020clustered}. Moreover, Ghosh et al. made use of iterative Kmeans clustering to remove the byzantine attackers \cite{ghosh2019robust}.
However, the adopted naive clustering such as Kmeans,Kmedian \cite{jain1988algorithms,lloyd1982least} have two shortcomings: 1) being \textit{inefficient} and easily trapped into \textit{local minima}; %optimization 
2) breaking down when the number of colluder groups is unknown to the the clustering method (see details in Sect. \ref{sec:failure}).\\
\textbf{Server based robust aggregation}
assumed \textit{the server had an extra training dataset} which was used to evaluate uploaded model updates. Either abnormal updates with low scores were filtered out \cite{cao2020fltrust,xie2018zeno,regattibygars,prakash2020mitigating} or minority Byzantine updates were filtered out through majority votes \cite{chen2018draco,rajput2019detox,gupta2021byzantine}. 
However, this approach is not applicable to the case when the server-side data is not available, or it might break down if the distribution of server's data deviates far from that of training data of clients. \\
\textbf{Historical Information based byzantine robust methods} made use of historical information (such as distributed Momentums \cite{el2021distributed}) to help correct the statistical bias brought by colluders during the training, and thus lead to the convergence of optimization of federated learning \cite{allen2020byzantine, kairouz2021advances, alistarh2018byzantine, chenbyzantine, karimireddy2021learning, farhadkhani2022byzantine}.\\
\textbf{Other byzantine robust methods} %Other defending methods  
used signs of the gradients \cite{bernstein2018signsgd, sohn2020election}, optimization strategy \cite{li2019rsa} or sampling methods \cite{karimireddy2020byzantine} to achieve robust aggregation. Recently, some work studied Byzantine-robust algorithms in the decentralized setting without server \cite{guerraoui2018hidden, peng2020byzantine,he2022byzantine} and asynchronous setting with heterogeneous communication delays \cite{damaskinos2018asynchronous, xie2020zeno++, yang2021basgd}.  \\
\textbf{Community detection in graph.} In the proposed framework, the detection of Byzantine colluders is treated as the detection of multiple subgraphs or \textit{communities} in a large weighted graph (see Sect. \ref{subsec:spa-temp}). Existing approaches \cite{newman2004detecting} could be applied to detect the byzantine colluders. One important technique is to detect specific features of graph such as clustering coefficients \cite{watts1998collective}. Another important technique is to leverage the spectral property of graph based on the \textit{adjacency matrix} or \textit{normalized Laplacian matrix} %, correlation matrix 
and so on \cite{shen2010spectral, sarkar2018spectral,brouwer2011spectra}. 
Moreover, the number of communities can be determined based on the eigengap of these matrix \cite{zelnik2004self} (see Sect. \ref{subsec: spectal-heu}). 

\begin{table*}[htbp]
\vspace{-3pt}
\setlength{\tabcolsep}{2.0mm}
\renewcommand\arraystretch{1.5}
%\multirow{2}{*}{\begin{tabular}[c]{@{}c@{}}Protection \\ Mechanism \end{tabular}}
\centering
\caption{Overview of the performance for various Byzantine-resilient methods (Robust statistics based: Krum \cite{blanchard2017machine}, Median, Trimmed Mean \cite{yin2018byzantine} and DnC \cite{shejwalkar2021manipulating}; Clustering based: Kmeans \cite{shen2016auror}; Server based: FLtrust \cite{cao2020fltrust}; Spectral based: FedCut, see more details in Sect. \ref{sec:related}) under different Byzantine attacks (Non-Collusion attack: \cite{blanchard2017machine}, label flipping \cite{data2021byzantine} and sign flipping \cite{data2021byzantine}, Collusion-diff attack: same value attack \cite{li2019rsa}, Fang-v1 (design for trimmed mean) \cite{fang2020local} and our designed Multi-collusion attack, Collusion-mimic attack: Mimic \cite{karimireddy2020byzantine} and Lie \cite{baruch2019little}, see more details in Sect. \ref{subsec:spectral property}). \cc, $\vv$ and \xx denotes the drop of model performance less than $3\%$, from $3\%$  to $10\%$ and above $10\%$ on Fashion MNIST respectively (see more comparisons in Sect. \ref{sec:experiment})}
\begin{tabular}{|c|c||c|c|c|c||c||c||c|}
\hline
&
&\multicolumn{4}{|c||}{Robust statistics based}&Clustering based &Server based & Spectral based\\ 
%\hline
\cline{3-9}

&&Krum&Median&Tri mean&DnC& Kmeans&FLtrust&FedCut(Ours)\\ \hline
\multirow{3}{*}{Non-Collusion}&Gaussian\cite{blanchard2017machine} &$\vv$&\cc&\cc&\xx&\xx&\cc& \cc\\ \cline{2-9}

&label flipping\cite{data2021byzantine} &$\vv$&\xx&\xx&\cc&\cc&\cc&\cc \\ \cline{2-9}

&sign flipping\cite{data2021byzantine} &$\vv$&$\vv$&$\vv$&\xx&\xx&$\vv$&\cc \\ \hline

\multirow{3}{*}{Collusion-diff} &same value\cite{li2019rsa}  &$\vv$&\xx&\xx&\cc&\cc&\cc&\cc \\ \cline{2-9}

&Fang-v1\cite{fang2020local} &$\vv$&\xx&\xx&\cc&\cc&$\vv$& \cc\\ \cline{2-9}

&Multi-collusion&$\vv$&\xx&\xx&$\vv$&\xx&\cc&\cc \\ \hline

\multirow{2}{*}{Collusion-mimic}& Mimic\cite{karimireddy2020byzantine} &\xx&$\vv$&$\vv$&$\vv$&\cc&\cc&\cc \\ \cline{2-9}

&Lie \cite{baruch2019little}&$\vv$&\xx&\xx&\xx&\cc&\xx&\cc \\ \hline

\end{tabular}
\label{tab:ref}
\vspace{-3pt}
\end{table*}
\section{Preliminary}
\subsection{Federated Learning}
We consider a \textit{horizontal federated learning} \cite{yang2019federated,mcmahan2017communication} setting consisting of one server and $K$ clients. We assume $K$ clients\footnote{In this article we use terms "client", "node", "participant" and "party" interchangeably. } have their local dataset $\calD_i = \{(\bx_{i,j},y_{i,j} \}_{j=1}^{n_i}, i=1\cdots K$, where $\bx_{i,j}$ is the input data, $y_{i,j}$ is the label and $n_i$ is the total number of data points for $i_{th}$ client. The training in federated learning is divided into three steps which iteratively run until the learning converges:
\begin{itemize}
    \item The $i_{th}$ client takes empirical risk minimization as:
    \begin{equation}
        \min_{\bw_i} F_{i}(\bw_i, \calD_i) =\min_{\bw_i} \frac{1}{n_i} \sum_{j=1}^{n_i} \ell(\bw_i, \bx_{i,j}, y_{i,j}),
    \end{equation}
    where $\bw_i \in \RR^d$ is the $i_{th}$ client's local model weight and $\ell(\cdot)$ is a loss function that measures the accuracy of the prediction made by the model on each data point.
    \item Each client sends respective local model updates $\nabla F_i$ to the server and the server updates the \textit{global model} $\bw$ as $\bw = \bw - \eta \frac{1}{K}\sum_{i=1}^K\nabla F_i$, where $\eta$ is learning rate.
    \item The server distributes the updated global model $\bw$ to all clients.
\end{itemize}
% In each step, $i_{th}$ client minimizes the local risk function as $\min_{\bw} F_{i}(\bw) =\min_{\bw} \frac{1}{|\calD_i|} \sum_{j=1}^{|\calD_i|} f(\bw, \bx_{i,j}, \by_{i,j})$, where $\bx_{i,j} \in \calD_i, j \in [|\calD_i|]$. Then clients send the updates $\bg_i = \nabla F_i(\bw)$ to the server and the server updates the global model $\bw$ as $\bw = \frac{1}{K}\sum_{i=1}^K\bg_i + \bw$. Finally, the server distribute the updated global model to all all clients.

\subsection{Byzantine Attack in Federated Learning}
We assume a malicious threat mode where an unknown number of participants out of K clients are Byzantine, i.e., they may upload arbitrarily corrupt updates $\bg_b$ to degrade the global model performance (MP). Under this assumption, behaviours of Byzantine clients and the rest of benign clients can be summarized as follows:
\begin{equation} \label{eq:upload-gradients}
\bg_i = \left\{
\begin{aligned}
\nabla F_i \qquad &\text{Benign clients}\\
\bg_b  \qquad & \text{Byzantine clients}\\
\end{aligned}
\right.
\end{equation}
Note that under the assumed threat mode, each adversarial node has access to updates of all clients during the training procedure. They are aware of the adoption of Federated Learning Byzantine-resilient methods \cite{blanchard2017machine,baruch2019little, xie2020fall}, and conspire to upload specially designed model updates that may overwhelm existing defending methods. Byzantine clients who behave consistently as such are referred to as \textit{colluders} throughout this article.  Moreover, we assume the server to be honest and try to defend the Byzantine attacks.

% \subsection{Weighted Undirected Graph in Federated Learning}
% We regard model updates contributed by $K$ clients as an undirected graph $G = (V,E)$, where $V={v_1,\cdots, v_K}$ represent $K$ model updates, $E$ is a set of weighted edge representing similarities between uploaded model updates corresponding to clients in $V$.
% We assume that the graph $G = (V,E)$ is weighted, that is each edge between two nodes $v_i$ and $v_j$ carries a non-negative weight, e.g., $A_{ij} = \text{exp}(-||\bg_i-\bg_j||^2/2\sigma^2) \geq0$, where $\bg_i$ is uploaded gradient for $i_{th}$ client and $\sigma$ is the scaling factor. Denote $G_R = (V_R, E_R)$ and $G_B =(V_B, E_B)$ are \textit{two subgraphs} of $G$ representing benign and Byzantine clients. 

%The gist of the proposed method is how to discern colluders from benign clients, by scrutinizing similarities between their respective model updates (see Fig. \ref{fig:setup} for an overview).
\section{Failure Cases Caused by Byzantine Colluders} \label{sec:failure}
%(move to experiment section)

% The challenge brought by colluders is detrimental to Byzantine-resilient algorithms in different ways. We illustrate in this section a toy example that showcases how colluders can conspire to attack robust statistic, SVD and clustering based methods and cause serious damage to their defending capabilities, which is defined as the Byzantine Tolerant Rate as follows: 
The challenge brought by colluders is detrimental to Byzantine-resilient algorithms in different ways. We first analyze below behaviours of representative Byzantine attackers from a graph theoretic perspective. Then we demonstrate a toy example to showcase that existing Byzantine-resilient methods are vulnerable to collusion attacks. 

% how colluders can conspire to attack robust statistic, SVD and clustering based methods and cause serious damage to their defending capabilities.

\begin{figure*}[htbp]
	\centering
	\vspace{-5pt}
			\subfigure[]{
		\begin{minipage}[b]{0.22\textwidth}
			\includegraphics[width=1\textwidth]{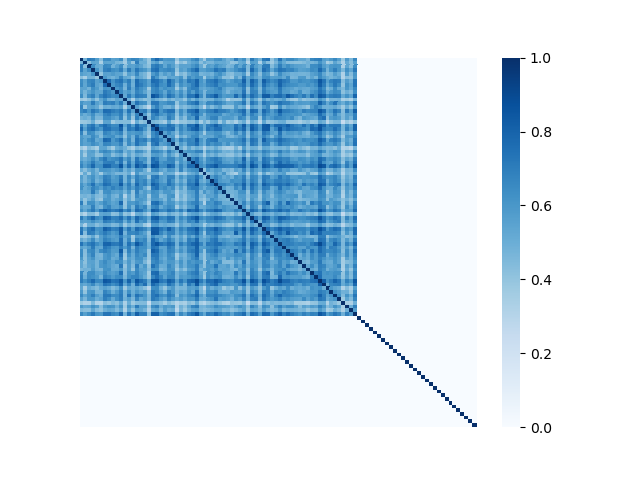}
		\end{minipage}
	}
    	\subfigure[]{
    		\begin{minipage}[b]{0.22\textwidth}
  		 	\includegraphics[width=1\textwidth]{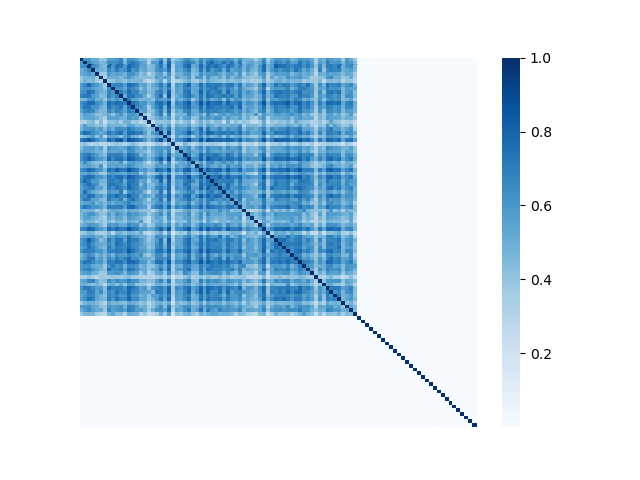}
    		\end{minipage}
    	}
    		\subfigure[]{
    		\begin{minipage}[b]{0.22\textwidth}
  		 	\includegraphics[width=1\textwidth]{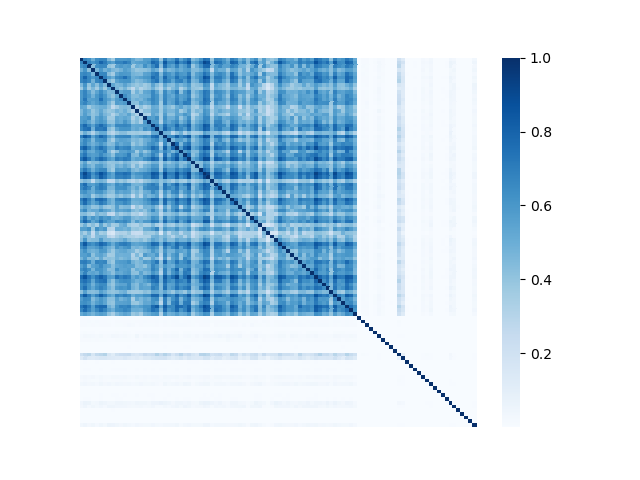}
    		\end{minipage}
    	}
    	\subfigure[]{
		\begin{minipage}[b]{0.22\textwidth}
			\includegraphics[width=1\textwidth]{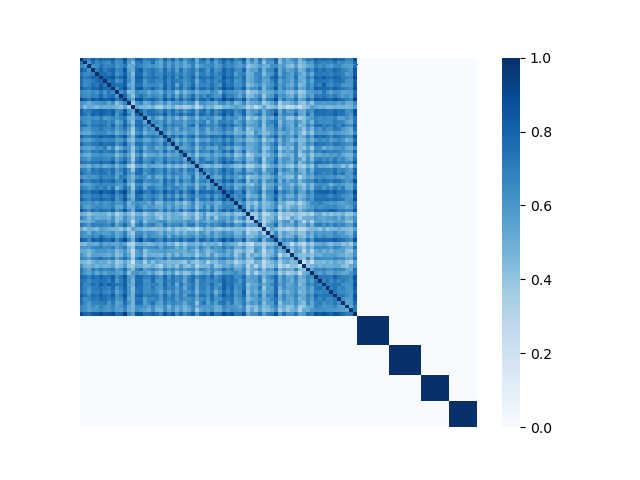}
		\end{minipage}
	}
	
	\subfigure[]{
		\begin{minipage}[b]{0.22\textwidth}
			\includegraphics[width=1\textwidth]{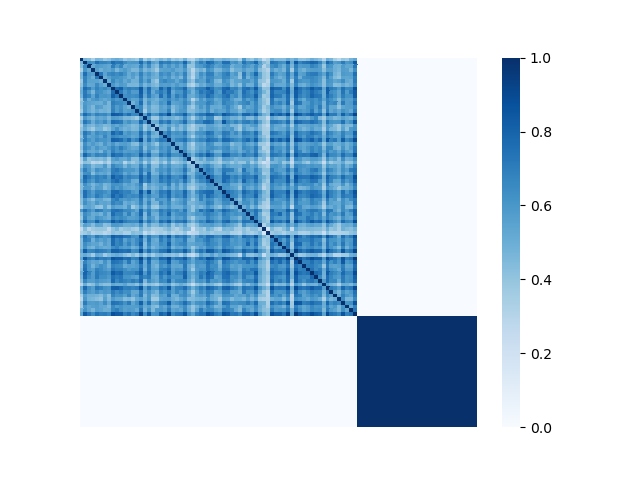}
		\end{minipage}
	}
    		\subfigure[]{
    		\begin{minipage}[b]{0.22\textwidth}
  		 	\includegraphics[width=1\textwidth]{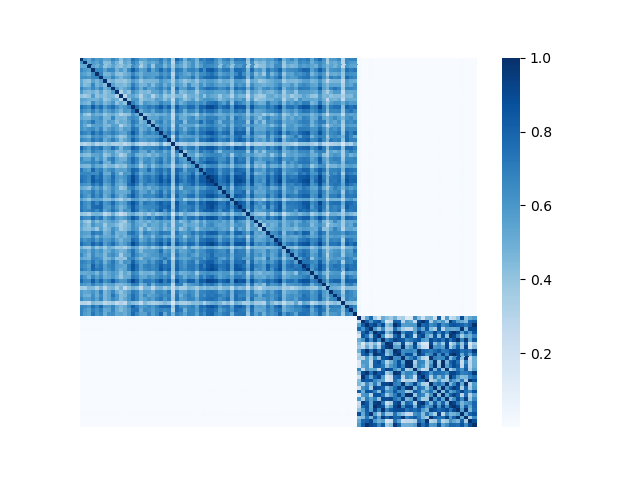}
    		\end{minipage}
    	}
    	\subfigure[]{
		\begin{minipage}[b]{0.22\textwidth}
			\includegraphics[width=1\textwidth]{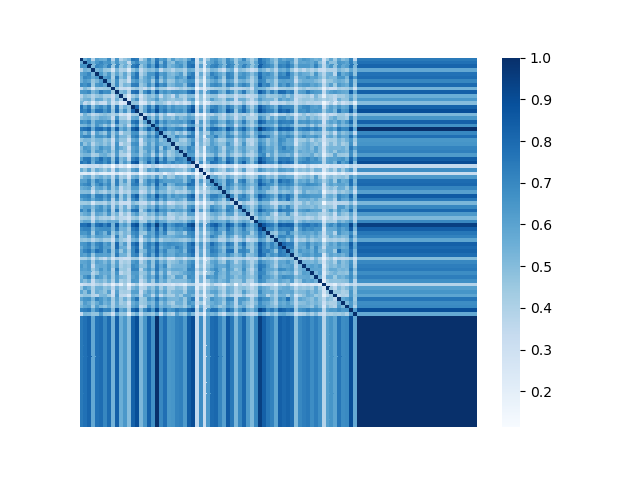}
		\end{minipage}
	}
	    \subfigure[]{
    		\begin{minipage}[b]{0.22\textwidth}
  		 	\includegraphics[width=1\textwidth]{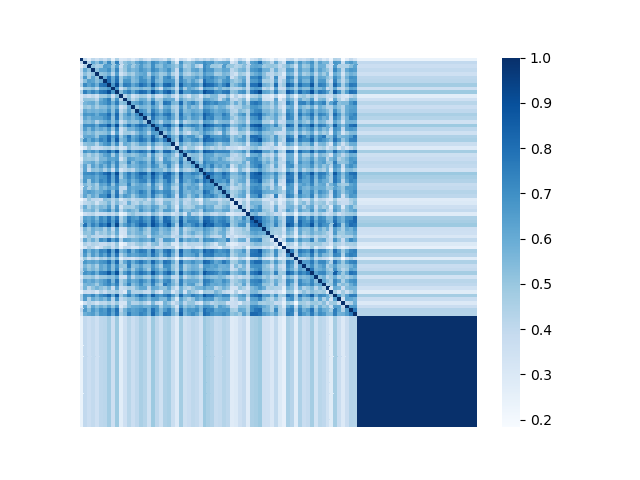}
    		\end{minipage}
    	}
	\caption{Heatmap of adjacency matrices for 8 types of attack (in Sect. \ref{subsec:experiment-setup}) of 100 client including 30 attackers at $1000_{th}$ iteration under IID setting. In each subfigure, benign clients form a single coherent cluster residing in the upper-left block of the adjacent matrix, while attackers reside in the bottom-right parts of the adjacent matrix. The scaling factor $\sigma^2=10$ and the task is logistic regression for MNIST dataset. From left to right, top to bottom, attack methods are Gaussian attack \cite{blanchard2017machine}, label flipping \cite{data2021byzantine}, sign flipping \cite{data2021byzantine}, our designed multi-collusion attack, same value attack \cite{li2019rsa}, Fang-v1 (design for trimmed mean) \cite{fang2020local}, Mimic attack, and Lie \cite{baruch2019little} respectively.}
	\vspace{-5pt}
	\label{fig:aff-matrix-iid}

\end{figure*}

\subsection{Weighted Undirected Graph in Federated Learning}
We regard model updates contributed by $K$ clients as an undirected graph $G = (V,E)$, where $V={v_1,\cdots, v_K}$ represent $K$ model updates, $E$ is a set of weighted edge representing similarities between uploaded model updates corresponding to clients in $V$.
We assume that the graph $G = (V,E)$ is weighted, and each edge between two nodes $v_i$ and $v_j$ carries a non-negative weight, e.g., $A_{ij} = \text{exp}(-||\bg_i-\bg_j||^2/2\sigma^2) \geq0$, where $\bg_i$ is uploaded gradient for $i_{th}$ client and $\sigma$ is the Gaussian scaling factor. Let $G_R = (V_R, E_R)$ and $G_B =(V_B, E_B)$ respectively denote \textit{two subgraphs} of $G$ representing benign and Byzantine clients. 

Moreover, Byzantine problem \cite{lamport2019byzantine} could be regarded as finding a optimal graph-cut for $G$ to distinguish the Byzantine and benign model updates. Since model updates from colluders form specific patterns (see Fig. \ref{fig:spectral-analysis} for examples), the aforementioned graph-cut can be generalized to the so called \emph{community-detection problem} \cite{newman2004detecting} in which multiple subsets of closely connected nodes are to be separated from each other.

\subsection{Spectral Graph Analysis for Byzantine Attackers}\label{subsec:spectral property}

We illustrate the spectral analysis of representative Byzantine attacks, especially, those launched by colluders. For example, Fig. \ref{fig:aff-matrix-iid} shows adjacency matrices with elements representing pairwise similarities between \textit{70 benign clients} under IID setting and \textit{30 attackers} (the darker the element the higher the pairwise similarity is). It is clear in each subfigure that benign clients form a single coherent cluster residing in the upper-left block of the adjacent matrix\footnote{Note that the order of the nodes illustrated in Fig. \ref{fig:aff-matrix-iid} is irrelevant and we separate benign nodes from Byzantine nodes only for better visual illustration.}, while attackers reside in the bottom-right parts of the adjacent matrix. We observe the following characteristics pertaining to benign as well as Byzantine model updates. 

First, benign model updates form a single group (upper-left block of Fig. \ref{fig:aff-matrix-iid}), which is formally illustrated by Assumption \ref{assum: Bounded gradient dissimilarity}, i.e., they all lie in the group (circle) with center $\nabla F$ and radius $\ka$.
% especially under IID setting. We provide the Assumption \ref{assum: Bounded gradient dissimilarity} to describe the property of benign updates, i.e., they all lie in the group (circle) with center $\bar{\bg}$ and radius $\ka$. Moreover, Assumption \ref{assum: Bounded gradient dissimilarity} is widely used to bound heterogeneity in datasets among clients \cite{li2019communication, yu2019linear}. 
Specifically, when the local datasets are homogeneous (IID), $\ka$ is small, indicating the strong similarity among benign model updates. When the local data of clients becomes heterogeneous, the radius $\ka$ becomes larger. Note that Assumption \ref{assum: Bounded gradient dissimilarity} has been widely used to bound differences between model updates of benign clients, e.g., in \cite{li2019communication, yu2019linear}. 
%Also note that strong connections between benign model updates can be understood from a \emph{small world} point of view \cite{watts1998collective} (see detailed analysis in Appendix C). 

\begin{assumption} \label{assum: Bounded gradient dissimilarity}
Assume the difference of local gradients $\nabla F_i$ and the mean of benign model update ${\nabla F} = \frac{1}{|V_R|}\sum_{i \in V_R} \nabla F_i$ is bounded ($V_R$ is the set of benign clients), i.e., there exists a finite $\ka$, such that
\begin{equation*}
    || \nabla F_i-{\nabla F}|| \leq \ka\footnote{\text{$||\cdot||$ in the paper represents the $\ell_2$ norm}}.
\end{equation*}
\end{assumption}

\begin{figure*}
\centering
% \begin{minipage}{.3\textwidth}
%     \subfigure[]{
%       \label{a} 
%       \includegraphics[width=.95\textwidth]{imgs/Spectral_Analysis/spectral_analysis_eigenvalue.png}}\\
%     \subfigure[]{
%       \label{b} 
%       \includegraphics[width=.95\textwidth]{imgs/Spectral_Analysis/spectral_analysis.png}}
% \end{minipage}
% \begin{minipage}{.65\textwidth}
%     \subfigure[]{
%       \label{c} 
%       \includegraphics[width=0.95\textwidth]{imgs/Spectral_Analysis/attack_type.PNG}}
% \end{minipage}
	\centering
			\subfigure[]{
		\begin{minipage}[b]{0.22\textwidth}
			\includegraphics[width=1\textwidth]{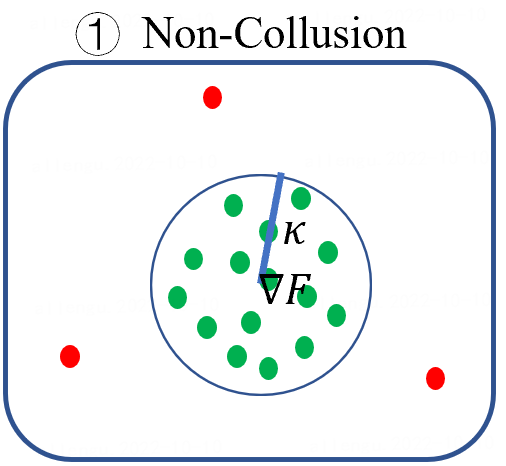}
		\end{minipage}
	}
    	\subfigure[]{
    		\begin{minipage}[b]{0.22\textwidth}
  		 	\includegraphics[width=1\textwidth]{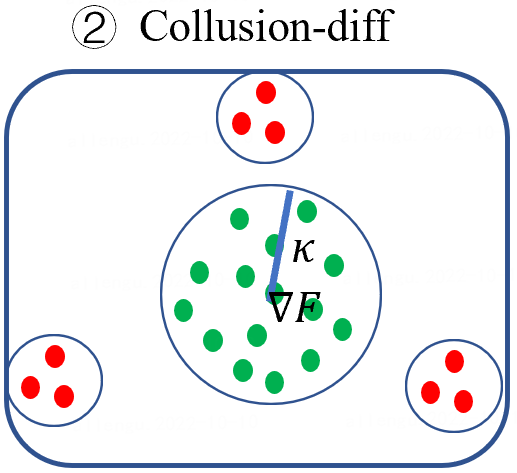}
    		\end{minipage}
    	}
    	\subfigure[]{
    		\begin{minipage}[b]{0.22\textwidth}
  		 	\includegraphics[width=1\textwidth]{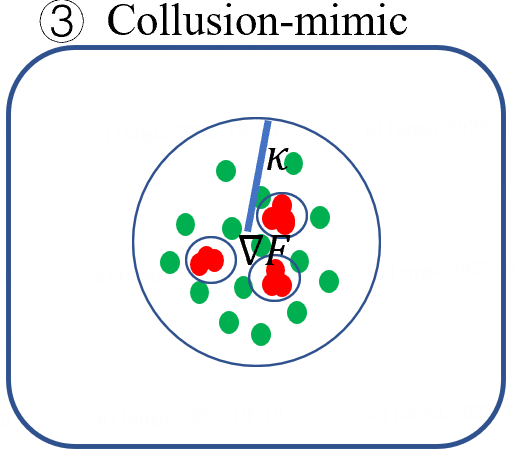}
    		\end{minipage}
    	}    	\subfigure[]{
    		\begin{minipage}[b]{0.22\textwidth}
  		 	\includegraphics[width=1\textwidth]{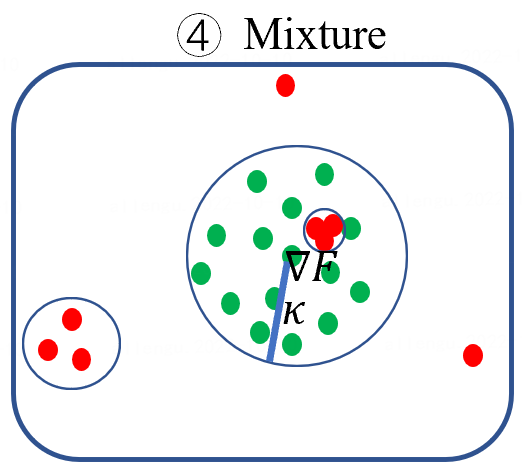}
    		\end{minipage}
    	}
	
	\centering
			\subfigure[]{
		\begin{minipage}[b]{0.22\textwidth}
			\includegraphics[width=1\textwidth]{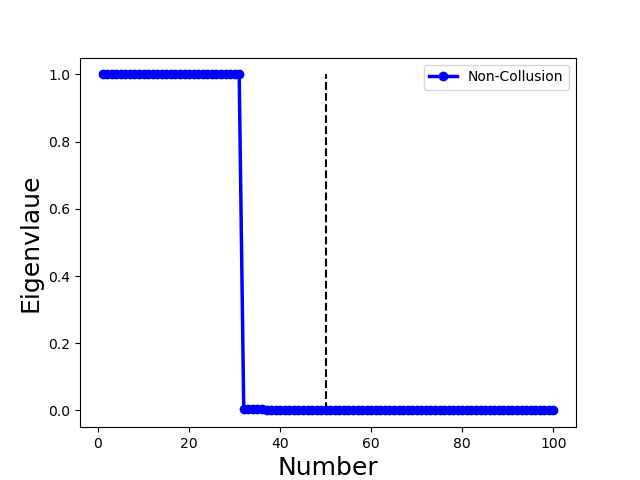}
		\end{minipage}
	}
    	\subfigure[]{
    		\begin{minipage}[b]{0.22\textwidth}
  		 	\includegraphics[width=1\textwidth]{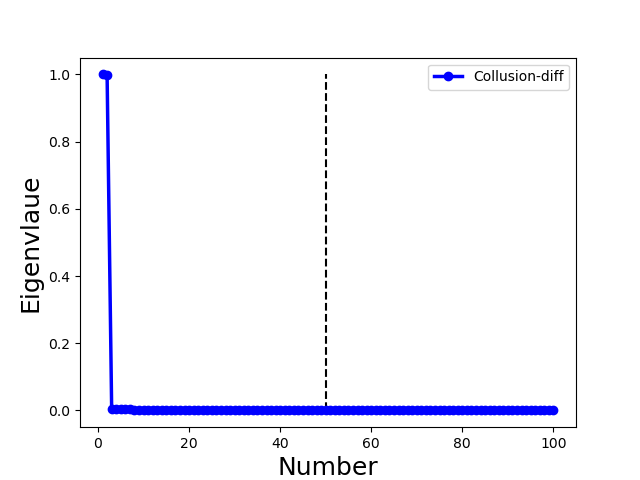}
    		\end{minipage}
    	}
    	\subfigure[]{
    		\begin{minipage}[b]{0.22\textwidth}
  		 	\includegraphics[width=1\textwidth]{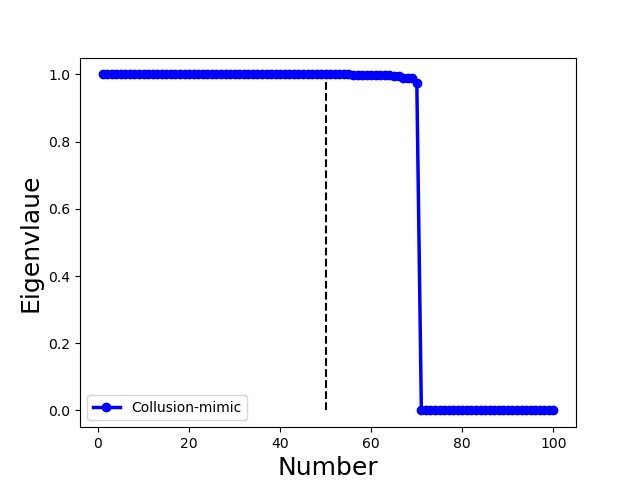}
    		\end{minipage}
    	}    	\subfigure[]{
    		\begin{minipage}[b]{0.22\textwidth}
  		 	\includegraphics[width=1\textwidth]{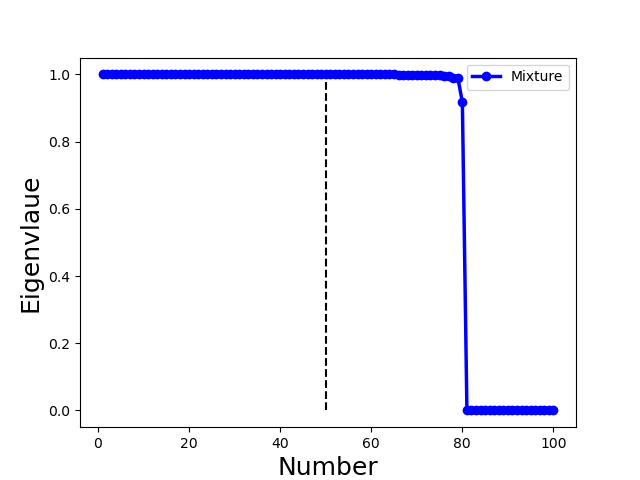}
    		\end{minipage}
    	}

 	\centering
			\subfigure[]{
		\begin{minipage}[b]{0.22\textwidth}
			\includegraphics[width=1\textwidth]{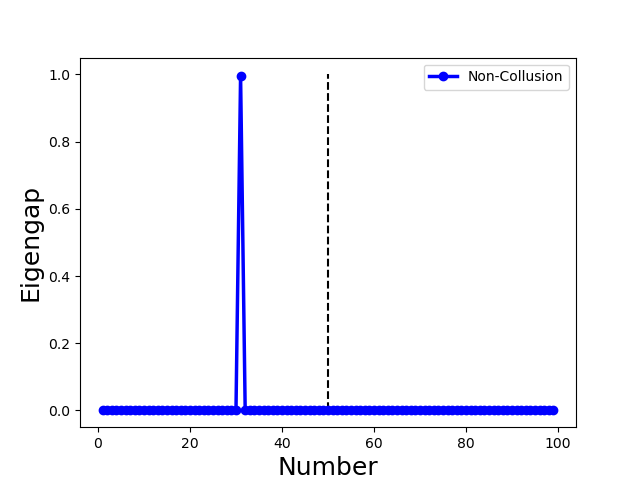}
		\end{minipage}
	}
    	\subfigure[]{
    		\begin{minipage}[b]{0.22\textwidth}
  		 	\includegraphics[width=1\textwidth]{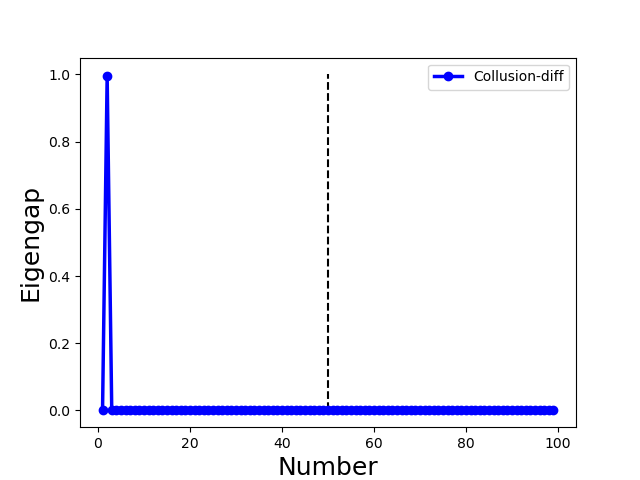}
    		\end{minipage}
    	}
    	\subfigure[]{
    		\begin{minipage}[b]{0.22\textwidth}
  		 	\includegraphics[width=1\textwidth]{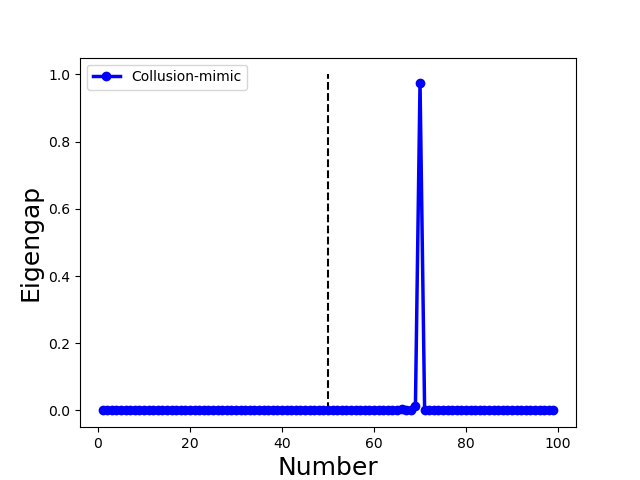}
    		\end{minipage}
    	}    	\subfigure[]{
    		\begin{minipage}[b]{0.22\textwidth}
  		 	\includegraphics[width=1\textwidth]{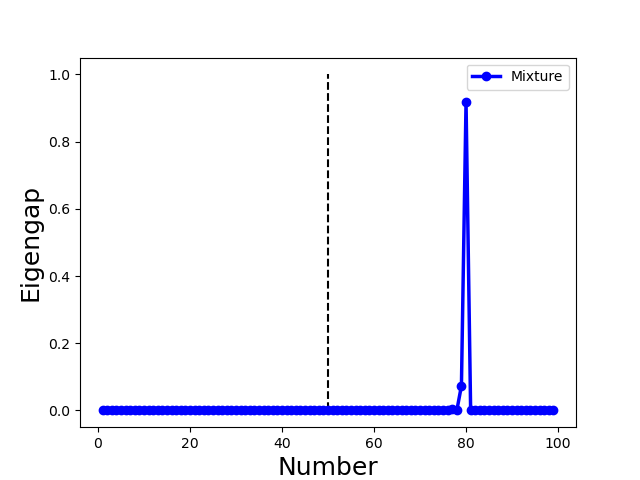}
    		\end{minipage}
    	}

\caption{Spectral analysis of four types of Byzantine attackers (Each column represents one type of attack), including Non-Collusion, Collusion-diff, Collusion-mimic and Mixture attacks. We take one example of 100 clients consisting of 70 benign clients and 30 attackers, view the model updates as nodes, and compute the pair-wise similarities as edges. The first row provides an overview of four attacks ((a), (b), (c) and (d) represents Non-Collusion, Collusion-diff, Collusion-mimic and Mixture attacks respectively where green and red points represent benign and Byzantine updates respectively), while the second and third rows show the eigenvalue and eigengap of the normalized adjacency matrix (see Def. \ref{def:eigengap}).}
\label{fig:spectral-analysis}
\vspace{-5pt}
\end{figure*}

Second, Byzantine model updates can be categorized into four types (Fig. \ref{fig:spectral-analysis}):
% \begin{itemize}
%     \item 1) the uploaded model updates among benign clients are similar under IID setting so that benign clients can be regarded as one group. Moreover, when the local data of client becomes heterogeneous, the connection of benign clients becomes weak (see Appendix C);
%     \item 2) the existing eight types of Byzantine attacks can be divided into the following four types (Fig. \ref{fig:spectral-analysis}(c)).
% \end{itemize} 

\begin{itemize}
    \item Non-Collusion: $\|\bg_b - \nabla F \| > \ka $ and malicious updates ($\bg_b$) are far away from each other (e.g., Gaussian attack \cite{blanchard2017machine}, label flipping \cite{data2021byzantine} and sign flipping \cite{data2021byzantine} attacks in Fig. \ref{fig:aff-matrix-iid}(a), (b) and (c)). 
    \item Collusion-diff: $\|\bg_b - \nabla F \| > \ka $ and malicious updates ($\bg_b$) form one or multiple clusters (small intra-cluster distance) (e.g., our designed collusion attack (see Sect. \ref{subsec:experiment-setup}), same value attack \cite{li2019rsa}, Fang-v1 \cite{fang2020local} attacks in Fig. \ref{fig:aff-matrix-iid}(d), (e) and (f)). 
    \item Collusion-mimic: $\|\bg_b - \nabla F \| < \ka $ and $\bg_b$ of different attackers are almost identical ($\|\bg_b^i - \bg_b^j \| \ll \ka$). It represents adversaries with strong connections form one or multiple clusters (small intra-cluster distance), and their behaviours are very similar to the few selected benign clients but are different from the rest of benign clients (e.g., Mimic attack \cite{karimireddy2020byzantine} and Lie \cite{baruch2019little} in Fig. \ref{fig:aff-matrix-iid}(g) and (h)).
    \item Mixture: adversaries may combine Non-collusion, Collusion-diff, and Collusion-mimic arbitrarily to obtain a mixture attack.
\end{itemize}

In order to address the complex types of both benign and Byzantine model updates, we adopt he eigengap technique \cite{shen2010spectral, zelnik2004self} to reliably detect benign and Byzantine clusters (or communities). We provide the following proposition to elucidate characteristics of different types of Byzantine attacks in the lens of spectral analysis. The proof of Proposition \ref{prop:attack-type} is deferred to Appendix C.

% Specifically, the position of the largest eigengap (Def. \ref{def:eigengap}) helps to determine the number of intrinsic communities under different types of attacks.
% Moreover, spectral property of graph is one important technique in community detection, Shen et al. \cite{shen2010spectral} pointed out that $i$ can be viewed as the appropriate candidate for the number of intrinsic communities if the $i_{th}$ eigengap (as Def. \ref{def:eigengap}) is the largest. 

\begin{prop} \label{prop:attack-type}
Suppose $K$ clients consist of $m$ benign clients and $q$ attackers ($q<m-1$). If Assumption \ref{assum: Bounded gradient dissimilarity} holds for Non-collusion and Collusion-diff attacks, then only the first $c$ eigenvalues are close to 1 and
\begin{itemize}
    \item  $c=1+q<\frac{K}{2}$ for Non-collusion attacks provided that $||\bg_b-\nabla F || \gg \ka$ and malicious updates ($\bg_b$) are far away from each other;
    \item $c=1+B<\frac{K}{2}$ for Collusion-diff attacks provided that malicious updates form $B$ groups and $||\bg_b-\nabla F || \gg \ka$;
    \item $c >\frac{K}{2}$ for Collusion-mimic attacks that $||\bg_b-\nabla F || < \ka$ and malicious updates are almost identical.
\end{itemize}
\end{prop}

\begin{rmk}

Proposition \ref{prop:attack-type} illustrates that the eigenvalue 1 has multiplicity $c$, and then there is a large gap to the $(c+1)_{th}$ eigenvalue $\lambda_{c+1} <1$. Therefore, we can use the position of the pronounced eigengap to elucidate the number of community clusters. Specifically, we use the largest eigengap to determine the number of the community clusters in our proposed method. Moreover, Proposition \ref{prop:attack-type} demonstrates that the position of the pronounced eigengap for Non-collusion and Collusion-diff attacks is less than $\frac{K}{2}$ while larger than $\frac{K}{2}$ for Collusion-mimic attack. This property is used in the proposed method to distinguish Collusion-mimic attacks from other types of attacks (see Sect. \ref{subsubsec:mimic}).  
\end{rmk}

We illustrate with a concrete example the different spectra of four types of attacks in Fig. \ref{fig:spectral-analysis}, in which each column represents one type of attack. We observe the following properties (all of the examples include 70 benign clients and 30 attackers):
\begin{itemize}
    \item for Non-Collusion attack, we use the Gaussian attack \cite{li2019rsa} to upload largely different random updates following $\calN(0,200)$. Fig. \ref{fig:spectral-analysis}(e) and (i) show the eigenvalue of adjacency matrix $L$ drops abruptly between index 30 and 31, i.e., the largest eigengap lies in index 31. It indicates the number of different communities is 31 with 70 benign clients forming a close group and 30 individual attackers forming other groups;
    \item  for Collusion-diff attack, we use the same-value attack \cite{li2019rsa} to upload updates consisting of all one elements, Fig. \ref{fig:spectral-analysis}(f) and (j) illustrate the largest eigengap lies in 2 elucidating that benign clients form one group and colluders with strong connections form the other group;
    \item  for Collusion-mimic attack, we use the mimic attack \cite{karimireddy2020byzantine} to upload updates that mimic one benign update, the largest eigengap is 70 because Byzantine colluders and mimiced benign clients form a group while the rest of 69 benign clients form 69 groups separately as Fig. \ref{fig:spectral-analysis}(g) and (k) illustrate. It is because the connection among mimic colluders is much stronger than benign clients.
    \item for Mixture attack, we combine the Gaussian attack (5 attackers), same-value attack (5 attackers), and mimic attack (20 attackers). Fig. \ref{fig:spectral-analysis}(h) and (l) reveal that the largest eigengap is 80, in which mimic colluders and one mimiced benign client form a group. In contrast, 69 benign clients and 10 attackers form 79 groups separately. This example showcase that the mimic attack dominates the spectral characteristics of the Mixture attack. Therefore, for the proposed method, the Collusion-mimic type attack is first detected and removed, followed by the detection of other two types of attack (see Algo. \ref{algo:para_deter} in Sect. \ref{subsubsec:mimic} and line 3-4 of Algo. \ref{algo:fedcut} in Sect. \ref{subsec:fedcut}).
\end{itemize}

\subsection{Failure Case Analysis} \label{subsec:fail-case}
We evaluate two types of representative Byzantine-resilient methods (Robust statistics and the clustering based aggregation methods) and the proposed FedCut method (see Sect. \ref{subsec:fedcut}) under four typical Byzantine Attacks mentioned in Sect. \ref{subsec:spectral property}, in terms of their Byzantine Tolerance Rates (BTR) defined below. Specifically, we assume 10 benign model updates following 1D Gaussian distribution $\calN(0.1,0.1)$ and see four types Byzantine attacks (\textbf{S1}-\textbf{S4}) as illustrated in Tab. \ref{tab:four_sce}. 

\begin{table}[htbp]
    \centering
\setlength{\tabcolsep}{2.0mm}
\renewcommand\arraystretch{1.0}
\vspace{-2pt}
\caption{\small four scenarios (\textbf{S1}-\textbf{S4}) of Byzantine Attacks. \textbf{S1} represents the \textbf{Non-Collusion} type; \textbf{S2-s} and \textbf{S2-m} represent \textbf{Collusion-diff} and the difference is \textbf{S2-s} only has a single cluster of colluders while \textbf{S2-m} consists multiple clusters of colluders; \textbf{S3} denotes the \textbf{Collusion-mimic} scenario; \textbf{S4} denotes a mixture scenario comprising the \textbf{Non-Collusion}, \textbf{Collusion-diff} and \textbf{Collusion-mimic} type.} 
\begin{tabular}{|l|l|l|l|}
\hline
Scenario & Attack Type & Attack Distribution & Number \\ \hline
\textbf{S1} &Non-Collusion& $\calN(0.1,1)$ & 8  \\ \hline 
\textbf{S2-s} &Collusion-diff&  $\calN(-2,0.01)$ & 8 \\ \hline 
 \multirow{2}{*}{\textbf{S2-m} }& \multirow{2}{*}{Collusion-diff }&$\calN(-2,0.01)$ & 4 \\
& &$\calN(4,0.01)$ &  4 \\ \hline

\textbf{S3} &Collusion-mimic&$\calN(\mu,0.01)$\tablefootnote{$\mu$ is the minimum elements of 10 benign model updates} & 8 \\ \hline 

 \multirow{3}{*}{\textbf{S4} }& \multirow{3}{*}{Mixture }&$\calN(-2,0.01)$ & 3 \\
& &$\calN(\mu,0.01)$ &  3 \\ 
& &$\calN(0.1, 1)$ &  1 \\ \hline
\end{tabular}
    \label{tab:four_sce}
    \vspace{-4pt}
\end{table}

\begin{table*}[htbp]
    \centering
\setlength{\tabcolsep}{2.0mm}
\renewcommand\arraystretch{1.5}
\vspace{-2pt}
\caption{four scenarios (\textbf{S1}-\textbf{S4}) of Byzantine Attacks are repeated 1000 runs to evaluate 5 representative Byzantine-resilient methods and the proposed FedCut method, in terms of their Byzantine Tolerance Rates (BTR) in Def. \ref{def:Byzantine Tolerance}} 
\begin{tabular}{|l|l|l|l|l|l|l|}
\hline
\textbf{BTR} & Krum \cite{blanchard2017machine} & Median \cite{yin2018byzantine}& Trimmed Mean \cite{yin2018byzantine} & DnC \cite{shejwalkar2021manipulating} & Kmeans  \cite{shen2016auror} &FedCut (ours) \\ \hline
\textbf{S1}  & 96.0\%   & 96.2\%                                                    & 92.6\%                                                         & 95.3\%                                                 & 70.5\%                                                     & \textbf{96.2\%}                                                  \\ \hline
\textbf{S2-s}  & 34.5\%                                                 & 29.5\%                                                    & 27.5\%                                                         & 89.6\%                                                 & 99.0\%                                                    & \textbf{99.0\% }                                                 \\ \hline
\textbf{S2-m}  & 86.1\%                                                 & 89.3\%                                                    & 88.2\%                                                         & 99.4\%                                                 & 3.5\%                                                    & \textbf{99.6\% }                                                  \\ \hline
\textbf{S3}  & 29.6\%                                                 & 36.9\%                                                    & 37.7\%                                                         & 52.9\%                                                  & 33.7\%                                                     & \textbf{98.7\%}                                                   \\ \hline
\textbf{S4}  & 59.5\%                                                 & 61.5\%                                                    & 67.6\%                                                         & 53.1\%                                                  & 85.6\%                                                     & \textbf{95.9\% }                                                  \\ \hline
\end{tabular}
    \label{tab:toy_example}
    \vspace{-4pt}
\end{table*}

Next, we provide a toy example of four types of Byzantine attacks mentioned above to illustrate failure cases of existing Byzantine-resilient methods.

In order to evaluate different Byzantine-resilient methods under the four scenarios, we define the Byzantine Tolerant Rate representing the fraction of Byzantine Tolerant cases over repeated runs against attacks as follows:

\begin{defi} \label{def:Byzantine Tolerance}
(Byzantine Tolerant Rate)
Suppose the server receives $(K - q)$ correct gradients $\mathcal{V} = \{v_1,\cdots, v_{K-q}\}$ and $q$ Byzantine gradients $\mathcal{U} = \{ u_1, \cdots, u_q \}$. A Byzantine robust method $\mathcal{A}$ is said to be \textbf{Byzantine Tolerant} to certain attacks \cite{xie2020fall} if
\begin{equation}
<\EE[\mathcal{V}], \EE[\mathcal{A}([\mathcal{V} \cup \mathcal{U}])] > \geq  0 
\end{equation}
Moreover, the \textbf{Byzantine Tolerant Rate (BTR)} for $\mathcal{A}$ is the fraction of Byzantine Tolerant cases over repeated runs against attacks. 
% (illustrated in the toy example below). 
%A high BTR indicates good defending capability.
\end{defi}

Tab. \ref{tab:toy_example} summarized toy example results for different Byzantine-resilient methods under four types of Byzantine attacks mentioned above. We can draw following conclusions:
\begin{itemize}
    \item First, for Non-Collusion attack (\textbf{S1}), all Byzantine-resilient methods except Kmeans \cite{shen2016auror} perform well (i.e., the BTR is higher than 90\%), which indicates that Non-Collusion attack is easy to be defended.
    \item Second, for Collusion-diff attack (\textbf{S2}), Robust Statistics based methods such as Krum \cite{blanchard2017machine}, Median, Trimmed Mean \cite{yin2018byzantine} and DnC \cite{shejwalkar2021manipulating} are vulnerable to \textbf{S2-s}. This failure is mainly ascribed to the wrong estimation of sample \textit{mean} or \textit{median} misled by biased model updates from colluders. Moreover, the clustering based method i.e., Kmeans \cite{shen2016auror} fails in the \textbf{S2-m}, with BTR as low as 3.5\%. This is because the clustering based method relies on the assumption that only one group of colluders exists, but two or more groups of colluders in \textbf{S2-m} are misclassified by naive clustering-based methods with wrong assumptions.
    \item Third, for Collusion-mimic attack (\textbf{S3}), both Robust Statistics based methods and clustering based method fail, with BTR lower than 52.9\%. The main reason is that colluders would introduce statistical bias for benign updates and similar behaviours of colluders are hard to detect.
    \item Finally, the proposed FedCut method is able to defend against all attacks with high BTR (more than 95\%) by using a Spatial-Temporal framework and \textit{spectral heuristics} illustrated in Sect. \ref{sec:framework}.
\end{itemize}  

It is worth mentioning that the toy example used in this Section only showcases some simplified failure cases that might defeat existing Byzantine-resilient methods. Instead, the attacking methods evaluated in Sect. \ref{sec:experiment} is more complex and detrimental, and we refer readers to thorough experimental results in Sect. \ref{sec:experiment} and Appendix B.

\section{FedCut: Spectral Analysis against Byzantine Colluders} \label{sec:framework}

This section illustrates the proposed spectral analysis framework (see Algo. \ref{algo:byz-robust}) in which distinguishing benign clients from \textit{one or more groups of Byzantine colluders} is formulated as a \textit{community detection problem} in Spatial-Temporal graphs \cite{shen2010spectral}, in which nodes represent all model updates and weighted edges represent similarities between respective pairs of model update over all training iterations (see Sect. \ref{subsec:spa-temp}). The \textit{normalized graph cut} with temporal consistency \cite{shi2000normalized,ng2001spectral} (called FedCut) is adopted to ensure that a global optimal clustering is found (see Sect. \ref{subsec:fedcut}). Moreover, the \textit{spectral heuristics} \cite{zelnik2004self} is then used to determine the Gaussian scaling factor, the number of colluder groups and the attack type (see Sect. \ref{subsec: spectal-heu}). The gist of the proposed method is how to discern colluders from benign clients, by scrutinizing similarities between their respective model updates (see Fig. \ref{fig:setup} for an overview).

\begin{algorithm}[htbp]
\caption{\textbf{FedCut Framework}}
	\begin{algorithmic}[1]
		\renewcommand{\algorithmicrequire}{\textbf{Input:}}
		\renewcommand{\algorithmicensure}{\textbf{Output:}}
		\Require $K$ clients with local training datasets $\calD_i, i=1,2,\cdots,|\calD_i|$; number of global iterations $T$; %number of clients $\tau$  sampled in each iteration; 
		learning rate $\eta$ and batch size $b$.
		\Ensure Global model $\bw$. 
		\State $\bm{w} \leftarrow $ random initialization, $\tilde{L}^0 =\mathbf{0} \in \RR^{K \times K}$.
		\For{$t=1,2,\cdots,T$}
		    \State \textbf{Step \RomanNumeralCaps{1}}: The server sends the global model $\bm{w}$ to all
		    \Statex \qquad \qquad clients $i=v_1, v_2, \cdots, v_K$.
	        \State \textbf{Step \RomanNumeralCaps{2}}: Training local models and server model.
		    \ForParallel{$i=v_1, v_2, \cdots, v_K$} 
		        \State $\bg_{i}^t = ModelUpdate(\bw,\calD_{i},b,\eta)$.
		        \State Send $\bg_i$ to the server.
		    \EndFor
		    \State \textbf{Step \RomanNumeralCaps{3}}: Updating the global model via FedCut
		    \State $\calI_R, \tilde{L}^t = \textbf{\text{C-Ncut}}(\bg_{i}^t,\tilde{L}^{t-1})$
		  %  , where $\calI$ is the set of benign clients selected by FedCut
		    \State $\bg^t$ = \textbf{Aggregate}($\bg^t_{\calI_R}$), where $\bg^t_{\calI_R} = \{\bg^t_i|v_i \in \calI_R \}$
            \State $\bw \leftarrow \bw - \eta \bg^t$.
		\EndFor\\
		\Return $\bw$.
	\end{algorithmic} 
	 \label{algo:byz-robust}
\end{algorithm}

\subsection{A Spatial-Temporal Graph in Federated Learning} \label{subsec:spa-temp}

In order to use all information during the training, we define a Spatial-Temporal graph by considering behaviours of clients among all training iterations:

\begin{defi} \label{def:Spatial-Temporal graph}
(Spatial-Temporal Graph) Define a Spatial-Temporal graph $G = (V,E)$ as a sequence of snapshots $<G^1, \cdots, G^T>$, where $G^t=(V, E^t)$ is an undirected graph at iteration $t$. $V$ denotes a fixed set of $K$ vertexes representing model updates belonging to $K$ clients. $E^t$ is a set of weighted edge representing similarities between model updates corresponding to clients in $V$ at iteration $t$, where related adjacency matrix of the graph is $A^t$. 
\end{defi}

Given a set of model updates from $K$ clients over $T$ iterations during federated learning, one is able to construct such a Spatial-Temporal graph by assigning edge weights as the measured pair-wise similarity between model updates. 

The construction of the Spatial-Temporal graph with historical information is motivated by a common shortcoming of Byzantine-resilient methods, e.g., reported in \cite{karimireddy2021learning}, which showed that methods without using historical information during the training might lead to deteriorated global model performances in the presence of colluders. Our ablation study in Appendix B also confirmed the importance of using the Spatial-Temporal graph.

\subsection{FedCut} \label{subsec:fedcut}
\begin{figure*}[hbtp]
	\centering
         {\includegraphics[width = 0.75\linewidth]{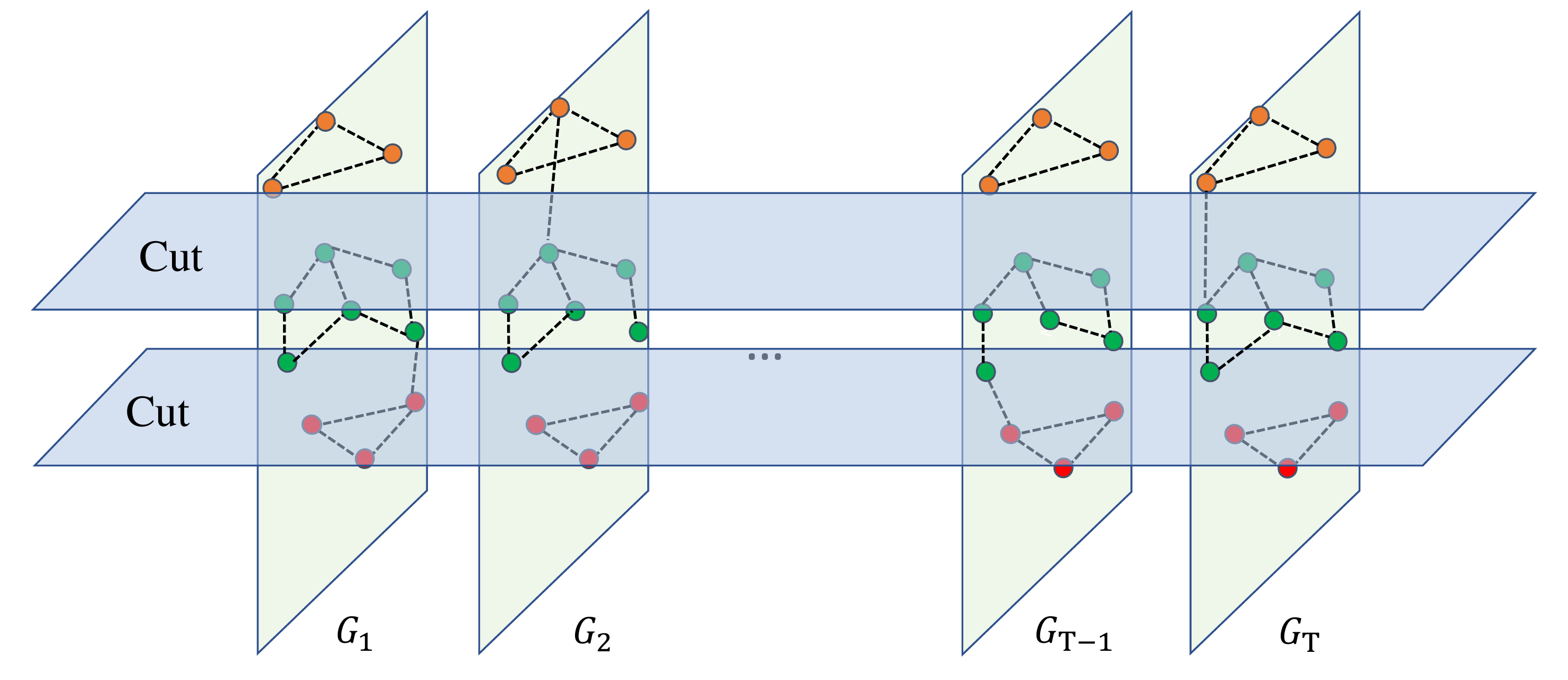}}
         \caption{FedCut: provide a temporal spatial cut to separate benign clients (green point) two-parties colluders (orange and red point) among all iterations.}

	\label{fig:setup}
	\vspace{-6pt}

\end{figure*}
%It is worth mentioning that the proposed detection method is robust under some special cases, a) there is a no Byzantine client at all i.e. the group number of Byzantine colluders is zero, b) the number of colluder group is one, i.e. the smallest colluder group; c) the number of colluder groups is more than one (? in our experiments). 
As mentioned in Sect. \ref{subsec:spa-temp}, we have defined a Spatial-Temporal graph $G=<G^1,\cdots, G^T>$ according to the model updates from $K$ clients over $T$ iterations during federated learning. Byzantine problem is then regarded as to find an \textit{optimal cut} for $G$ such that the inter-cluster similarities (between benign clients and colluders) are as small as possible and intra-cluster similarities are as large as possible \footnote{We view the cluster with the largest numbers as benign clients which constitute the largest community among all clients.}. One optimal cut as such is known as the \textit{Normalized cut} (\textit{Ncut}) and has been initially applied to image segmentation \cite{shi2000normalized, ng2001spectral}. We adopt the Ncut approach to identify those colluders which have large intra-cluster similarities.
%\HL{Compared to the Mincut \cite{von2007tutorial}, Ncut additionally maximizes the intra-class distance, which helps to detect the collusion (see results in Sect. \ref{})}. 
Moreover, we extend the notion of \textit{Ncut} to cater for the proposed Spatial-Temporal Graph over all iterations to make the detection of Byzantine colluders more consistent over multiple iterations. The \textit{c-partition Ncut for Spatial-Temporal Graph} is thus defined as follows:
% First, the ncut is averaged over all iterations to make the detection of Byzantine colluders more consistent over multiple iterations. Second, since there may exist multiple groups of colluders against benign clients, one has to resort to the k-partition instead of bi-partition of vertexes. To this end, we propose the optimization objective as follows. 
%\HL
{
\begin{defi} (c-partition Ncut for Spatial-Temporal Graph) \label{def:partition for Spatial-Temporal graph}
Let $G = (V,E)$ be a Spatial-Temporal graph as Def. \ref{def:Spatial-Temporal graph}. Denote the c-partition for graph $G$ as $V = B_1 \cup \cdots \cup B_c$ and $B_i \cap B_j = \varnothing$ for any $i,j$, c-partition Ncut for Spatial-Temporal Graph aims to optimize:
\begin{equation} \label{eq:loss-Spatial-Temporal}
    \min_{( B_1 \cup \cdots \cup B_c)=V}\quad \sum_{t=1}^T \sum_{i=1}^c\frac{W^t(B_i, \overline{B_i})}{Vol^t(B_i)}, 
\end{equation}
where $\overline{B_i}$ is the complement of $B_i$, $W^t(B_i,B_j): = \frac{1}{2}\sum_{i\in B_i, j\in B_j}A^t_{ij}$, and $A_{ij}^t$ is edge weight of $B_i$ and $B_j$, and $Vol^t(B_i): = \sum_{i \in B}\sum_{j\in V} A_{ij}^t$.
\end{defi}}
The proposed FedCut (Algo. \ref{algo:fedcut}) is used to get rid of malicious updates before model updates are aggregated. FedCut computes adjacent matrix according to the uploaded gradients over all training iterations and then takes the \textit{normalized cut} to distinguish benign clients and colluders. Specifically, three important parameters including the appropriate Gaussian kernel $\sigma_*$, the cluster number $c$ and the set of mimic colluders $\calI$ (who behave similarly to benign clients) need to be firstly determined (see Algo. \ref{algo:para_deter}). Then we calculate the adjacency matrix $A^t$ at $t_{th}$ iteration to remove the mimic colluders, i.e., set the connection between the mimic colluders and benign clients to be zero (line 2-7 in Algo. \ref{algo:fedcut}). We further compute normalized adjacency matrix $\tilde{L}^t$ by induction over $t$ iterations (line 8-9 in Algo. \ref{algo:fedcut}). Finally, we implement NCut into normalized adjacency matrix $\tilde{L}^t$ to choose the clusters of benign clients (line 10-12 in Algo. \ref{algo:fedcut}). 

% com FedCut takes as inputs a) model updates at $t_{th}$ iteration, and b) adjacent matrix normalized up to $(t-1)_{th}$ iteration. It then outputs the detected Byzantine colluders and benign clients, as well as the updated adjacent matrix for $t_{th}$ iteration. Specifically, we firstly compute the adjacency matrix among updates for different clients (line 1 in Algo. \ref{algo:fedcut}), then we calculate the normalized adjacency matrix $L^t$ according to updates of $(t-1)_{th}$ and $t_{th}$ iterations (line 2 in Algo. \ref{algo:fedcut}). It is noted that the operation in Line 2 is equivalent with calculating averaged normalized adjacency matrix over iterations $1, \cdots, t$ (see proof in Appendix \ref{sec:appE}). Finally, we implement NCut to normalized adjacency matrix (line 3-5 in Algo. \ref{algo:fedcut}). 
It is worth noting that FedCut runs over multiple learning iterations and, as shown by Proposition \ref{prop:prop1}, it is guaranteed to provide the optimal \textit{c-partition Ncut} defined by Eq. (\ref{eq:loss-Spatial-Temporal}) (See Proof in Appendix C).

\begin{prop} \label{prop:prop1}
%Solving FedCut (Algo. \ref{algo:fedcut}) $\iff$ doing \textit{C-partition Ncut for Spatial-Temporal Graph}
FedCut (Algo. \ref{algo:fedcut}) solves the \textit{c-partition Ncut} for Spatial-Temporal Graph i.e., it obtains a optimal solution by: 
\begin{equation*}
    \argmin_{( B_1 \cup \cdots \cup B_c)=V}\quad \sum_{t=1}^T \sum_{i=1}^c\frac{W^t(B_i, \overline{B_i})}{Vol^t(B_i)}
\end{equation*}.

\end{prop}

\begin{algorithm}[htbp]
\caption{c-NCut: c-partition Ncut for the Spatial-Temporal Graph}
	\begin{algorithmic}[1]
		\renewcommand{\algorithmicrequire}{\textbf{Input:}}
		\renewcommand{\algorithmicensure}{\textbf{Output:}}
\Require  The uploaded gradients of $K$ clients at $t_{th}$ iteration: $\{ \bg_{i}^t \}_{i=1}^K$, normalized adjacency matrix averaged over the first $t-1$ iterations: $\tilde{L}^{t-1}$;
\Ensure The set of benign clients: $\calI_R$; normalized adjacency matrix averaged over $t$ iterations: $\tilde{L}^t$. 
       \State  $c, \sigma_*, \calI=$ PDSH($\{ \bg_{i}^t \}_{i=1}^K$); 
        % \State Choose $\sigma_0 \ll \sigma_*$;
        \State Compute the adjacency matrix $A^t$ as:
       	\If{Client $i$ $\in \calI$}
		\State $A^t_{ij} =0$ ($i \neq j$)
		\Else
	   \State $A^t_{ij} =\text{exp}(-||\bg_i^t - \bg^t_j||^2/2\sigma_*^2)$
	   \EndIf
    %   \State Compute the adjacency matrix $A^t \in \mathbb{E}^{K\times K}$ with elements $A^t_{ij} = \text{exp}(-||\bg_i^t - \bg^t_j||^2/2\sigma^2)$; %$\bg_{i}^r$.
       \State $D = \text{diag}(\text{Sum}(A^t))$, $L^t = D^{-1/2}A^tD^{-1/2}$;
       \State $\tilde{L}^t = \frac{t-1}{t}\tilde{L}^{t-1} + \frac{1}{t}L^t$;
		%\State Use Kmeans for normalized top $c$ eigenvectors and assign $K$ clients into $c$ clusters.
		\State Eigen decomposition: $\tilde{L}^t= Q\Lambda Q^{-1}$;
		\State Apply \textit{Kmeans} into top $c$ normalized eigenvectors ($Q$) to assign $K$ clients into $c$ clusters.
    \State $\calI_R$ %$\calI^r$
    is the cluster with the largest number of clients; \\
    \Return $\calI_R$ and $\tilde{L}^t$. 
	\end{algorithmic} 
	 \label{algo:fedcut}
\end{algorithm}

\subsection{Spectral Heuristics} \label{subsec: spectal-heu}
% As mentioned in Sect. \ref{sec:failure}, clustering based byzantine resilient methods fail due to the \textit{unknown cluster numbers}. Consistent behaviours of colluders especially for the mimic colluders introduce \textit{statistical bias} to interfere the existing byzantine resilient methods. 
The determination of three parameters (i.e., the scaling parameter of Gaussian kernels $\sigma_*$, the number of clusters $c$ and mimic colluders $\calI$) are critical in thwarting Byzantine Collusion attacks as illustrated in Sect. \ref{subsec:fail-case}. We adopt the \textit{Spectral Heuristics} about the following definition of eigengap to determine these three important parameters. 

\begin{defi} \label{def:eigengap}
The $i_{th}$ \textit{eigengap} is defined as $|\lambda_i -\lambda_{i+1}|$ for which eigenvalues of the the normalized adjacency matrix $L = D^{-1/2}AD^{-1/2}$ are ranked in descending order, where $D = \text{diag}(\text{Sum}(A))$. Let $\delta$ be the largest eigengap among $|\lambda_i -\lambda_{i+1}|$.
\end{defi}

% In essence, these optimal parameters are selected such that \textit{the maximum eigengap are the largest} \cite{shen2010spectral,zelnik2004self}.
% \begin{defi} \label{def:eigengap}
% The $i_{th}$ \textit{eigengap} is defined as $|\lambda_i -\lambda_{i+1}|$ for which eigenvalues of the the normalized adjacency matrix $L = D^{-1/2}AD^{-1/2}$ are ranked in ascending order, where $D = \text{diag}(\text{Sum}(A))$.
% \end{defi}
% Then the eigengap is firstly used to select the optimal values of two important algorithm hyper-parameters i.e. 1) the scaling parameter of Gaussian kernels $\sigma$; 2) the number of colluder groups $c$. In essence, these optimal parameters are selected such that \textit{the maximum eigengap are the largest} \cite{shen2010spectral,zelnik2004self}. Moreover, mimic colluders are detected with the abnormal cluster numbers (see Algo. \ref{algo:para_deter}). 

\begin{algorithm}[htbp]
\caption{Parameter Determination via Spectral Heuristics (PDSH)}
	\begin{algorithmic}[1]
		\renewcommand{\algorithmicrequire}{\textbf{Input:}}
		\renewcommand{\algorithmicensure}{\textbf{Output:}}
        \Require The uploaded gradients of $K$ clients at $t_{th}$ iteration: $\{ \bg_{i}^t \}_{i=1}^K$, a preselected set of $\sigma$: $\{\sigma_1,\cdots \sigma_n\}$, where $\sigma_0 \ll \sigma_1$
		\Ensure The cluster number $c$, the appropriate Gaussian kernel $\sigma_*$, the set of mimic colluders $\calI$;
		\State Initialize $\calI = \emptyset$
		\While{ $\sigma \in \{\sigma_1,\cdots \sigma_n\}$}
       \State Compute the adjacency matrix $A^t$, where $A^t_{ij} =$ \\ \qquad $\text{exp}(-||\bg_i^t - \bg^t_j||^2/2\sigma^2)$.
       \State $D = \text{diag}(\text{Sum}(A^t))$, $L^t = D^{-1/2}A^tD^{-1/2}$. 
		\State Eigen decomposition: $L^t= P\Lambda P^{-1}$;
		\State Compute eigengap $\Delta_i = \mathbf{\lambda}_i - \mathbf{\lambda}_{i+1}$, where $\mathbf{\lambda}_i$ is the \\ \qquad $i_{th}$ eigenvalue;
		\State  j = $\arg\max_i\Delta_i$, $\Delta_\sigma = \Delta_j, c_\sigma = j$;
		\EndWhile 
		\State $\sigma_* = \arg \max_{\sigma} \Delta_\sigma$, $c= c_{\sigma^*}$;
		\If {$c> \frac{K}{2}$}
		\State Apply Kmeans into top $c$ normalized eigenvectors \\
		\qquad \quad to assign $K$ clients into $c$ clusters;
         \State $\calI$ are the clusters with number larger than 1.
        \EndIf
        \State $\sigma_* = \arg \max_{\sigma} \{\Delta_\sigma|c_{\sigma}< \frac{K}{2}\}$, $c= c_{\sigma^*}$;
\\
		\Return $c$, $\sigma^*$ and the set of mimic colluders $\calI$.
	\end{algorithmic} 
	 \label{algo:para_deter}
\end{algorithm}

%See Fig. \ref{fig:eigenp-heuristics} for example eigengaps and the use of eigengaps to select scaling parameters/factors. We refer readers to Appendix \ref{sec:appD} for detailed analysis of spectral heuristics. 
\begin{figure*}[htbp]
	\centering
			\subfigure[]{
		\begin{minipage}[b]{0.22\textwidth}
			\includegraphics[width=1\textwidth]{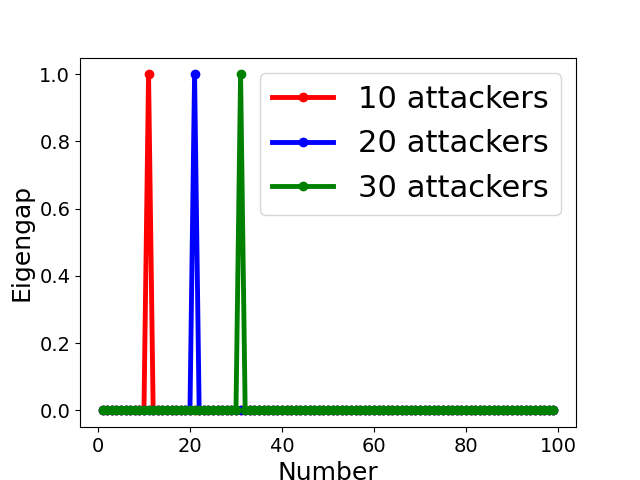}
		\end{minipage}
	}
    	\subfigure[]{
    		\begin{minipage}[b]{0.22\textwidth}
  		 	\includegraphics[width=1\textwidth]{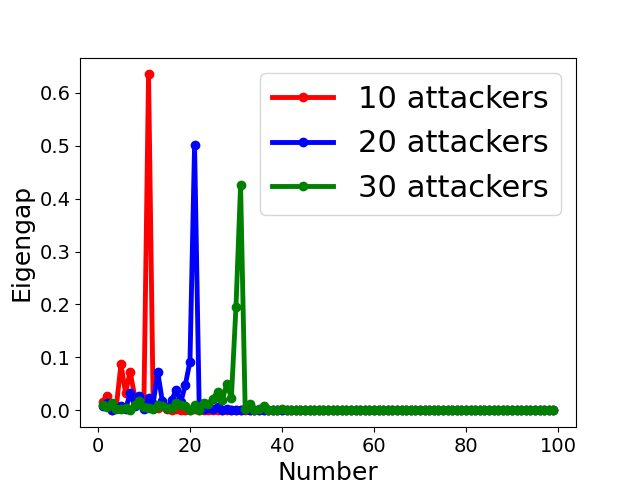}
    		\end{minipage}
    	}
    		\subfigure[]{
    		\begin{minipage}[b]{0.22\textwidth}
  		 	\includegraphics[width=1\textwidth]{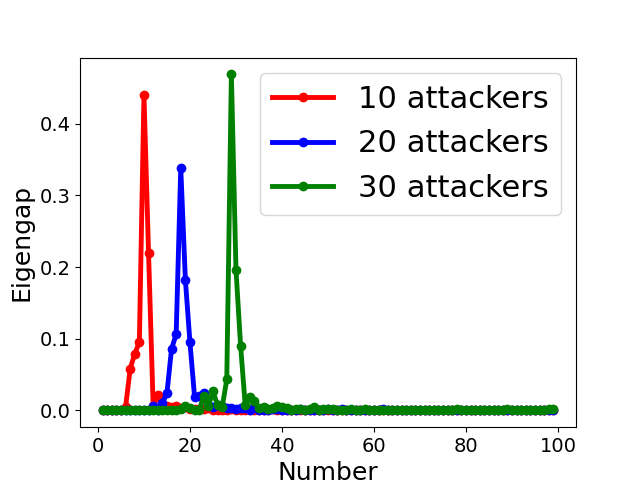}
    		\end{minipage}
    	}
    	\subfigure[]{
		\begin{minipage}[b]{0.22\textwidth}
			\includegraphics[width=1\textwidth]{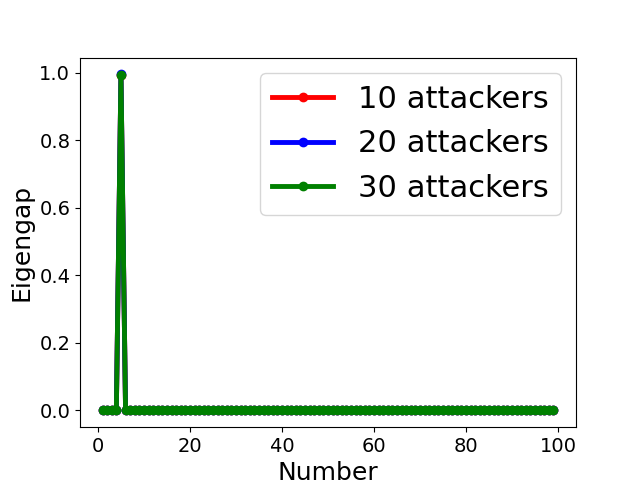}
		\end{minipage}
	}
	
	\subfigure[]{
		\begin{minipage}[b]{0.22\textwidth}
			\includegraphics[width=1\textwidth]{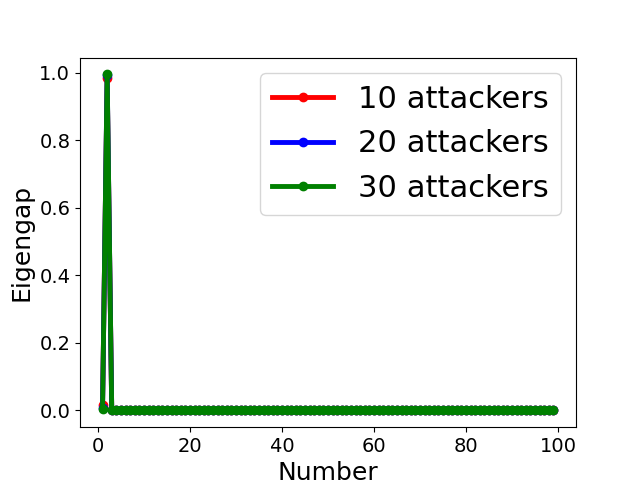}
		\end{minipage}
	}
    		\subfigure[]{
    		\begin{minipage}[b]{0.22\textwidth}
  		 	\includegraphics[width=1\textwidth]{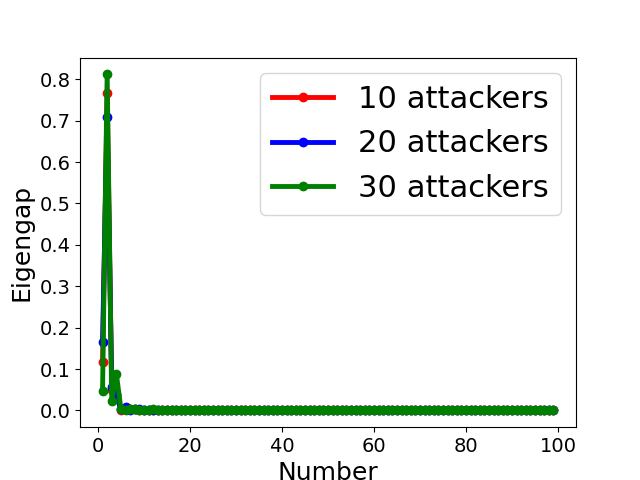}
    		\end{minipage}
    	}
    	 \subfigure[]{
    		\begin{minipage}[b]{0.22\textwidth}
  		 	\includegraphics[width=1\textwidth]{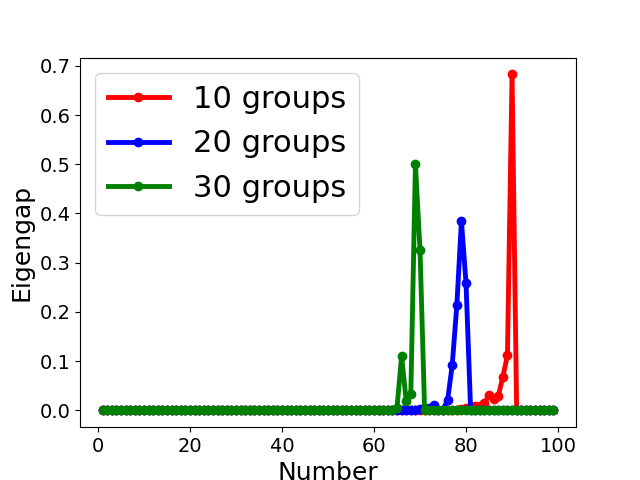}
    		\end{minipage}
    	}
	  \subfigure[]{
    		\begin{minipage}[b]{0.22\textwidth}
  		 	\includegraphics[width=1\textwidth]{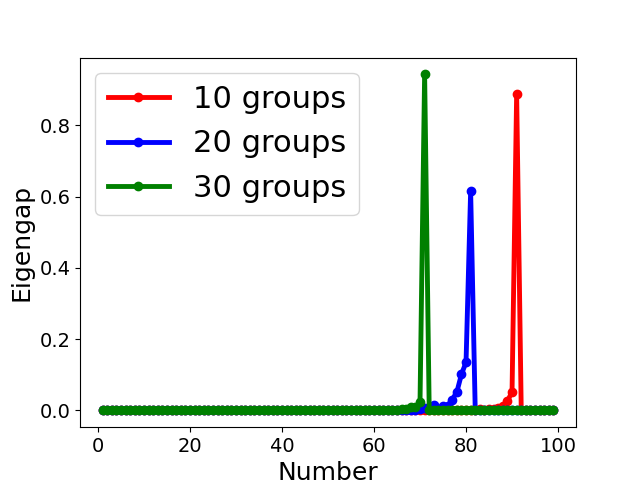}
    		\end{minipage}
    	}    

 	\caption{The change of eigengap for different attacks with different number of attackers under IID setting for MNIST dataset. From left to right, top to bottom, attack methods are Gaussian attack \cite{blanchard2017machine}, label flipping \cite{data2021byzantine}, sign flipping \cite{data2021byzantine}, our designed collusion attack (see Appendix A), same value attack \cite{li2019rsa}, Fang-v1 (design for trimmed mean) \cite{fang2020local}, Mimic attack \cite{karimireddy2020byzantine}, and Lie \cite{baruch2019little} respectively.}
	\vspace{-10pt}
	\label{fig:eigengap-clusternum1}
\end{figure*}

\subsubsection{Gaussian Kernel Determination}\label{subsubsec:gaussian}
The Gaussian kernel scaling parameter $\sigma$ in similarity function $A_{ij} = \text{exp}(-||\bg_i - \bg_j||^2/2\sigma^2)$ controls how rapidly the $A_{ij}$ falls off with the distance between $\bg_i$ and $\bg_j$. An appropriate $\sigma$ is crucial for distinguishing Byzantine and benign clients.
The following analysis shows that if the maximum eigengap $\delta$ (see Def. \ref{def:eigengap}) is sufficiently large, a small perturbation on the normalized adjacency matrix $L$ or will only affect the eigenvectors with bounded influence, and thus the clustering of $c$ clusters via top $c$ eigenvectors are stable.

\begin{prop}[Stability \cite{stewart1990matrix}] \label{prop:matrix-perturb}
Let $\lambda$, $Y$ and $\delta$ be eigenvalue, principle eigenvectors and the maximum eigengap of $L$ separately. Define a matrix small perturbation for $L$ as $\tilde{L} = L +E$ so that $||E||_2$ is small enough, let $\tilde{\lambda}$, $\tilde{Y}$ be eigenvalue and principle eigenvectors of $\tilde{L}$. If the maximum eigengap $\delta$ is large enough, then 
% \begin{equation}
%     ||\lambda -\tilde{\lambda}|| \leq ||E||
% \end{equation}

\begin{equation}
    ||Y -\tilde{Y}|| \leq \frac{4||E||}{\delta - \sqrt{2}||E||}
\end{equation}
\end{prop}
% \begin{rmk}
% The proposition \ref{prop:matrix-perturb} illustrates that the large eigengap stabilize the eigenvectors and further obtain a stable clustering with $c$ clusters (using top $c$ eigenvectors as line 11 in Algo. \ref{algo:fedcut}). 
% \end{rmk}

\begin{figure}[htbp]
\vspace{-6pt}
	\centering
	\subfigure{
		\begin{minipage}[b]{0.22\textwidth}
			\includegraphics[width=1\textwidth]{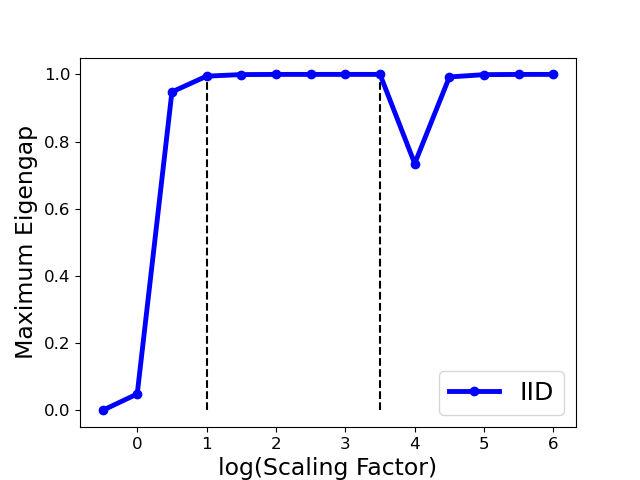}
		\end{minipage}
	}
    	\subfigure{
    		\begin{minipage}[b]{0.22\textwidth}
  		 	\includegraphics[width=1\textwidth]{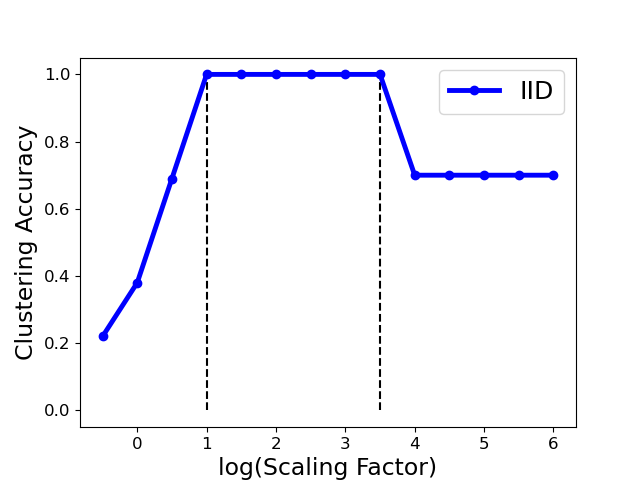}
    		\end{minipage}
    	}

	\caption{The change of the maximum eigengap of normalized adjacency matrix and clustering accuracy with different scaling factor $\sigma$ for the Gaussian attack \cite{li2019rsa}, where the normalized adjacency matrix $L = D^{-1/2}AD^{-1/2}$, $D = \text{diag}(\text{Sum}(A))$ and $A_{ij}= \text{exp}(-||\bg_i - \bg_j||^2/2\sigma^2)$. The clustering accuracy is calculated by NCut on the normalized adjacency matrix with 100 client including 70 benign clients and 30 Byzantine clients.}
	\vspace{-10pt}
	\label{fig:eigengap-gamma}
\end{figure}
It is noted that the error bound $||Y -\tilde{Y}||$ in the principle eigenvectors is affected by two factors: the perturbation $||E||$ in the affinity matrix, and the maximal eigengap $\delta$. While the perturbation is unknown and unable to modify for a given matrix $\tilde{L}$, one can seek to maximize the eigengap $\delta$ so that the recovered principle eigenvectors are as stable as possible. According to proposition \ref{prop:matrix-perturb}, we select the optimal parameter $\sigma_*$ by maximizing the maximum eigengap, i.e., 
\begin{equation}
    \sigma^* = \text{argmax}_{\sigma} \delta(\sigma)
\end{equation}
We implement this strategy in Algo. \ref{algo:para_deter}. Fig. \ref{fig:eigengap-gamma} shows that one can select $\sigma$ as such that the largest eigengap is up to its maximal (e.g., $log(\sigma)$ is in the range (1,3.5) for the Gaussian attack), provided that the clustering accuracy achieves the largest value (100\%) in this range. Also noted that the clustering accuracy drops in the extreme case (i.e., the $\sigma \geq 1e^4$) even the maximum eigengap is still large. Therefore, we select the optimal $\sigma$ at one reasonable range with removing the extreme case (see detailed range in Sect. \ref{subsec:experiment-setup}).

% \begin{rmk}
% The choose of $\sigma^*$ is important and some bad $\sigma^*$ causes the difficultly to distinguish the Byzantine and benign clients. For instance, as $\sigma$ tends to zero, $A_{ij}$ goes to zero ($\bg_i \neq \bg_j$). In this case, the computed connection is very low for all clients. As $\sigma$ tends to infinity, $A_{ij}$ goes to 1. Then all clients could be regarded as strongly-connected colluders. It is hard to distinguish the Byzantine and benign clients in both two cases no matter what the attack type is (\textbf{S1}-\textbf{S4} in Sect.\ref{subsec:fail-case}).

% \end{rmk}

\subsubsection{Determination of the Number of Clusters} \label{subsubsec:determine of cluster number}
The assumption made by many clustering-based Byzantine-resilient methods that model updates uploaded by all Byzantine attackers form \textit{a single group} might be violated in the presence of collusion (as shown by the toy example in Sect. \ref{subsec:fail-case} and experimental results in Sect. \ref{sec:experiment}). Therefore, it is necessary to discover the number of clusters, and one possible way is to analyze the spectrum of the adjacency matrix, normalized adjacency matrix or correlation matrix \cite{shen2010spectral}. According to the spectral property in Sect. \ref{subsec:spectral property}, we determine the clustering number $c$ (see line 3-11 in Algo. \ref{algo:para_deter}) based on the position of the largest eigengap of normalized adjacency matrix, i.e., 
\begin{equation}
    c = \text{argmax}_{k}|\lambda_k-\lambda_{k+1} |
\end{equation}

% We analyze the spectrum of normalized adjacency matrix to determine clustering numbers as follows. First, the analysis given in \cite{ng2001spectral,von2007tutorial} shows that one could estimate $c$ by counting the number of eigenvalues equaling 1 as Lemma \ref{lemma:clusternum-eigengap}.
% \begin{lem}\label{lemma:clusternum-eigengap} \cite{von2007tutorial}
% Let G be an undirected graph with non-negative weights. Then the multiplicity c of the eigenvalue 1 of $L$ equals the number of connected components $B_1,\cdots, B_c$ in the graph.
% \end{lem}

% Moreover, Shen et al. \cite{shen2010spectral} pointed out that $i$ can be viewed as the appropriate candidate for the number of intrinsic communities if the $i_{th}$ eigengap is the largest. We determine the clustering number $c$ using Algo. \ref{algo:para_deter}. 

Fig. \ref{fig:eigengap-clusternum1} displays the change of eigengap under different attacks, which shows that the index of the largest eigengap is a good estimation of the number of clusters. Specifically, 
\begin{itemize}
    \item for Non-Collusion attack (i.e., Gaussian attack \cite{blanchard2017machine}, label flipping \cite{data2021byzantine} and sign flipping \cite{data2021byzantine}), the largest eigengap in Fig. \ref{fig:eigengap-clusternum1}(a), (b) and (c) consisting of 10, 20 and 30 attackers lies in 11, 21, and 31 with benign clients forming one group and 10, 20 and 30 individual attackers;
    \item for Collusion-diff attack, the largest eigengap of Fig. \ref{fig:eigengap-clusternum1}(e) and (f) lies in between second and third largest eigenvalues. The position of the largest eigengap indicates that there are two clusters (one cluster represents benign clients while the other cluster represents colluders);
    \item for Collusion-mimic attack, the largest eigengap in Fig. \ref{fig:eigengap-clusternum1}(a), (b), and (c) of 10, 20 and 30 attackers lies in 90, 80, and 70 indicating that cluster numbers are 90, 80, and 70 (colluders form one group while other benign clients form 90, 80, and 70 groups separately);
\end{itemize}

% there are four attacks in the first row of Fig. \ref{fig:eigengap-clusternum1} where the largest eigengap lies in between second and third largest eigenvalues. Also, colluders form one group for these four attacks (see Sect. \ref{subsec:colluder-graph}). Consequently, the position of the largest eigengap indicates that there are two clusters (one cluster represents benign client while the other cluster represents colluders). 

In addition, Our empirical results in Appendix B also demonstrates the effectiveness of estimating the number of communities via the largest eigengap even for heterogeneous dataset among clients.
% change of eigengap is less stable with the increase of Non-IID extent (see Appendix D) because the increased heterogeneity of data reduces intra-cluster pairwise similarities between benign clients. 
% and, consequently, makes it harder to distinguish the benign and byzantine clients according to the adjacency matrix.

% and there are two observations:
% \begin{itemize}
%     \item First, the index of the largest eigengap is a good estimation of the number of clusters. For example, there are four attacks in the first row of Fig. \ref{fig:eigengap-clusternum1} where the largest eigengap lies in between second and third largest eigenvalues. Also, colluders form one group for these four attacks (see Sect. \ref{subsec:colluder-graph}). Consequently, the position of the largest eigengap indicates that there are two clusters (one cluster represents benign client while the other cluster represents colluders).
%     %each attacker forms one group so that attackers numbers are cluster numbers of attackers (see details in Sect. \ref{subsec:spectral property}). Consequently, the position of the largest eigengap indicates the clusters numbers.
%     \item Second, the change of eigengap is less stable with the increase of Non-IID extent because the increased heterogeneity of data reduces intra-cluster pairwise similarities between benign clients and, consequently, makes it harder to distinguish the benign and byzantine clients according to the adjacency matrix.
% \end{itemize}

\subsubsection{Mimic Colluders Detection}\label{subsubsec:mimic}
Colluders may mimic the behaviours of some benign clients towards over-emphasizing some clients and render other model updates useless. Therefore, it is hard to distinguish the mimic colluders and benign clients. The existing byzantine-resilient methods could not defend the Collusion-mimic attack (see Sect. \ref{subsec:fail-case} for \textbf{S3} case). We leverage the spectral property to detect the mimic colluders if the position of the largest eigengap is larger than $\frac{K}{2}$ (Proposition \ref{prop:attack-type}). Specifically, when the position of the largest eigengap is larger than $\frac{K}{2}$, we pick the clusters with number of clients larger than two as mimic colluders sets (see line 12-15 in Algo. \ref{algo:para_deter}).

\section{Convergence Analysis} \label{sec:converg}
We prove the convergence of such an iterative FedCut algorithm. The proof for Theorem \ref{thm:thm2} uses the similar technique as \cite{data2021byzantine} (see details in Appendix C).

% \begin{assumption} \label{assum:connected}
% The graph of benign clients $G_R$ is connected.
% \end{assumption}

\begin{assumption} \label{assum: diff-cluster}
For malicious updates $\bg_b$ provided that $||\bg_b- \nabla F||> \ka$, the difference between the mean of benign updates and colluders' updates has at least $C\ka$ distance,  where $C$ is a large constant, i.e.,
\begin{equation*}
    \| \bg_b -  \nabla F \| > C\ka.
\end{equation*}

%the updates of clients in same parties has at most $C\ka$ distance (C>2).
\end{assumption}

% \begin{assumption} \label{assum: Bounded gradient dissimilarity}
% We define benign mean model update across
% clients to be $\bar{\bg} = \frac{1}{|V_R|}\sum_{i \in V_R} \bg_{i}$
% where $V_R$ is the set of benign clients, hence the variance across client updates as $\mathbb{E}||\bg_i-\bar{\bg}|| \leq \ka$
% across all rounds of training.
% \end{assumption}

% \begin{rmk}
% $\ka$ in Assumption \ref{assum: Bounded gradient dissimilarity} has also been used earlier to bound heterogeneity in
% datasets \cite{li2019communication}; Specifically, when the data is homogeneous, we have $\ka = 0$ in Assumption \ref{assum: Bounded gradient dissimilarity}.
% \end{rmk}

\begin{rmk}
$C$ in the Assumption \ref{assum: diff-cluster} is greatly larger than 1, which demonstrates a large distance between colluders and the mean of benign updates. If the single attacker uploads gradients within $\ka$ distance from the mean of benign updates, we can regard this attacker as one benign client. For mimic colluders who introduce the statistical bias, we detect these colluders and remove them according to spectral heuristics (see Sect. \ref{subsubsec:mimic}).
% 1) for some parties which distance from benign clients is smaller than $C\ka$, even if they are misclustered, the $||\hat{\bg} - \bg||$ \footnote{$||\cdot||$ represent $\ell_2$ norm.} is bounded by $\tilde{\calO}(C^2\alpha^2\ka)$ by induction. \\
% 2) If we do NCut for each iteration separately, here $C$ is selected by smallest values over all iterations so that $C$ is small. FedCut averages the normalized adjacency matrix for multiple iterations so that $C$ is average case instead of minimum case. Consequently, $C$ in FedCut is larger than $C$ in NCut, which result in being more stable than NCut.
\end{rmk}

\begin{thm} \label{thm:thm1}
Suppose an $0< \alpha < \frac{1}{2}$ fraction of clients are Byzantine attackers. If Assumption \ref{assum: Bounded gradient dissimilarity} and \ref{assum: diff-cluster} holds, we can find the estimate of $\hat{\bg}$ according to line 11 in Algo. \ref{algo:byz-robust} and with the probability $1-\tilde{O}(\sqrt{Z_1})$, such that $||\hat{\bg} - \nabla F|| \leq \calO(\alpha \ka)$, where $Z_1 =  \frac{4K \sqrt{\alpha(1-\alpha)Z^{\frac{C^2}{4}}}}{\delta - K \sqrt{2\alpha(1-\alpha)Z^{\frac{C^2}{4}}}}, Z = \text{exp}(-2 \ka^2/\sigma^2)$.
\end{thm}

\begin{rmk}
The Theorem \ref{thm:thm1} illustrates that the distance between estimated gradients $\hat{\bg}$ via FedCut and benign averaged gradients $ \bar{\bg}$ is bounded by $\calO(\alpha \ka)$. Moreover, $\alpha(1-\alpha)$ is increasing w.r.t. $\alpha$ when $\alpha <\frac{1}{2}$, and $Z_1$ is increasing w.r.t. $\alpha(1-\alpha)$ . Therefore, $Z_1$ is increasing w.r.t. $\alpha$, which indicates that the probability $1-\tilde{O}(Z_1)$ decreases as the number of Byzantine attackers increases. In addition, for the larger $C$, $Z_1$ tends to zero since $Z<1$. Consequently, $||\hat{\bg} - \nabla F|| \leq \calO( \ka)$ with a high probability close to 1. 
\end{rmk}

\begin{assumption} \label{assum:bound local variance}
The stochastic gradients sampled from any local dataset have uniformly bounded variance over $\calD_i$ for all benign clients, i.e., there exists a
finite $\sigma_0$, such that for all $\bx_{i,j} \in \calD_i , i\in [K]$, we have
 \begin{equation}
     \mathbb{E}_{j} ||[\nabla F_i(\bw_i, \bx_{i,j})) - \nabla F_i(\bw_i)||^2 \leq \sigma_0^2,
 \end{equation}
 %, where $F_i(\bw) = \sum_{j=1}^{|D_i|} f(\bw, x_{i,j})$.
where $ \nabla F_i(\bw_i) = \mathbb{E}_{j}\nabla F_i(\bw_i, x_{i,j}))$.
\end{assumption}
\begin{rmk}
The difference between Assumption \ref{assum: Bounded gradient dissimilarity} and \ref{assum:bound local variance} is that the former bounds the variance across gradient estimates within the same client while the latter bounds the variance between model updates across clients.
\end{rmk}

% \begin{assumption}
% There exists at least one global minima x such that $F(x^*) \leq F(x)$ for any x.
% \end{assumption}
\begin{assumption} \label{assu:l-smooth}
We assume that F(x) is L-smooth and has $\mu$-strong convex.
% \begin{equation}
%     <\nabla F(x), y-x> + \frac{\mu}{2}||y-x||^2 \leq F(y)- F(x) \leq <\nabla F(x), y-x> + \frac{L}{2}||y-x||^2
% \end{equation}
\end{assumption}

% This Theorem illustrates the aggregation output of our temporal Ncut is bounded by ...

\begin{thm}\label{thm:thm2}
Suppose an $0< \alpha<\frac{1}{2}$ fraction of clients are corrupted.
For a global objective function $F:R^d\to R$, the server obtains a sequence of iterates $\{\bw^t: t\in[0:T]\}$ (see Algo. 1) when run with a fixed step-size $\eta< \min\{\frac{1}{4L},\frac{1}{\mu} \}$. If Assumption \ref{assum: Bounded gradient dissimilarity},  \ref{assum: diff-cluster}, \ref{assum:bound local variance} and \ref{assu:l-smooth} holds, the sequence of average iterates $\{\bw^t: t\in[0:T]\}$
satisfy the following convergence guarantees: 
\begin{equation}
    \left\|\bw^{T} - \bw^*\right\|^2 \leq (1-\frac{C_1\mu}{L})^T\left\|\bw^{0} - \bw^*\right\|^2 + \frac{\Gamma}{\mu^2},
\end{equation}

where $\Gamma = \calO(\sigma_0^2 + \ka^2 + \alpha^2 \ka^2)$, $C_1$ is a constant and $\bw^*$ is the global optimal weights in federated learning.
\end{thm}
Theorem \ref{thm:thm2} provides the convergence guarantee of FedCut framework in the strong convex case. As $T$ tends to infinity, the upper bound of $\|\bw^{T} - \bw^*\|^2$ becomes large with the increasing of $\sigma_0$ (variance of gradients within the same client), $\ka$ (variance between model updates across benign clients) and $\alpha$ (the ratio of Byzantine attackers).

% \subsection{Clustering accuracy and spectral gap}
% \subsection{Majority vote for model performance}
% \subsection{Convergence of our unified framework}

\section{Experimental results} \label{sec:experiment}
\begin{table*}[bp] 

\vspace{-4pt}

\centering 
%\caption{Efficiency: Training time of one iteration on workers comparison between Distributed Learning, SplitFed, FL-DP, CryptoNets, BatchCryptoNets and FDL-PP}
\caption{Model performances of different Byzantine-resilient methods under different Byzantine attacks (with IID setting 30 Byzantine clients for classification of MNIST, Fashion MNIST and CIFAR10). Moreover, the baseline of FedAvg \cite{mcmahan2017communication} without any attacks achieve the model performance 92.5\%, 90.1 \% and 69.4\% for MNIST, Fashion MNIST and CIFAR10 respectively.}
% of one iteration on workers comparison between Distributed Learning, SplitFed, FL-DP, CryptoNets, BatchCryptoNets and FDL-PP}
\footnotesize
\setlength{\tabcolsep}{0.9mm}
\renewcommand\arraystretch{1.5}
\begin{tabular}{|c||c|c|c|c|c|c|c|c|c|c|}
\hline

&& 
Krum\cite{blanchard2017machine} &GeoMedian\cite{chen2017distributed}&Median\cite{yin2018byzantine}&Trimmed\cite{yin2018byzantine}& Bulyan\cite{guerraoui2018hidden}&FLtrust\cite{cao2020fltrust}& DnC\cite{shejwalkar2021manipulating}&Kmeans\cite{shen2016auror}&FedCut(Ours)\\ \hline \hline

\multirow{9}{*}{\begin{tabular}[c]{@{}c@{}} 
M\\N
\\ I\\S\\T \end{tabular} }

 &No attack  &90.5$\pm${0.1}  &92.4$\pm${0.0}  &90.5$\pm${0.1}  &90.3$\pm${0.1}  &89.6$\pm${0.1}  &89.4$\pm${0.1}  &92.2$\pm${0.3}  &92.4$\pm${0.3}  &92.5$\pm${0.1}

 \\ \cline{2-11} &Lie \cite{baruch2019little}  &90.5$\pm${0.0}  &84.4$\pm${0.6}  &84.0$\pm${0.3}  &83.6$\pm${0.8}  &73.6$\pm${1.0}  &89.5$\pm${0.1}  &92.2$\pm${0.3}  &92.2$\pm${0.2}  &92.3$\pm${0.2}

 \\ \cline{2-11} &Fang-v1 \cite{fang2020local}  &90.2$\pm${0.1}  &45.8$\pm${0.5}  &43.4$\pm${1.5}  &37.8$\pm${2.6}  &85.5$\pm${0.3}  &75.8$\pm${3.5}  &92.3$\pm${0.3}  &92.3$\pm${0.2}  &92.3$\pm${0.1}

 \\ \cline{2-11} &Fang-v2 \cite{fang2020local}  &41.1$\pm${7.4}  &53.4$\pm${2.1}  &39.0$\pm${2.3}  &42.4$\pm${2.8}  &18.8$\pm${6.1}  &79.7$\pm${0.8}  &92.0$\pm${0.6}  &88.2$\pm${0.4}  &92.3$\pm${0.1}

 \\ \cline{2-11} &Same value \cite{li2019rsa}  &90.6$\pm${0.1}  &85.6$\pm${0.0}  &77.0$\pm${0.4}  &75.1$\pm${0.6}  &88.8$\pm${0.4}  &89.8$\pm${0.3}  &92.3$\pm${0.2}  &92.2$\pm${0.3}  &92.2$\pm${0.1}

 \\ \cline{2-11} &Gaussian \cite{blanchard2017machine}  &90.4$\pm${0.1}  &92.4$\pm${0.1}  &90.6$\pm${0.1}  &90.8$\pm${0.1}  &89.6$\pm${0.3}  &87.0$\pm${1.0}  &74.3$\pm${0.4}  &22.4$\pm${4.3}  &92.3$\pm${0.1}

 \\ \cline{2-11} &sign flipping \cite{data2021byzantine} &90.6$\pm${0.1}  &91.3$\pm${0.1}  &90.2$\pm${0.1}  &90.2$\pm${0.1}  &89.6$\pm${0.2}  &72.2$\pm${5.5}  &90.7$\pm${0.1}  &49.9$\pm${1.3}  &91.9$\pm${0.2}

 \\ \cline{2-11} &label flipping \cite{data2021byzantine} &90.4$\pm${0.1}  &89.4$\pm${0.1}  &85.6$\pm${0.2}  &85.4$\pm${0.4}  &89.5$\pm${0.4}  &89.3$\pm${0.6}  &92.2$\pm${0.2}  &92.1$\pm${0.2}  &92.1$\pm${0.4}

 \\ \cline{2-11} &minic  \cite{karimireddy2020byzantine} &86.1$\pm${0.4}  &86.3$\pm${0.9}  &88.3$\pm${0.3}  &89.7$\pm${0.2}  &85.5$\pm${0.5}  &90.5$\pm${0.1}  &92.2$\pm${0.4}  &92.1$\pm${0.1}  &92.3$\pm${0.1}

 \\ \cline{2-11} &Collusion (Ours)  &90.3$\pm${0.2}  &83.9$\pm${2.1}  &79.5$\pm${0.7}  &78.0$\pm${0.8}  &89.7$\pm${0.3}  &89.4$\pm${0.1}  &90.2$\pm${0.3}  &43.2$\pm${20.8}  &92.3$\pm${0.1} 

 \\ \cline{2-11} &Averaged 
 &85.1$\pm${0.9}  &80.5$\pm${0.7}  &76.8$\pm${0.6}  &76.3$\pm${0.9}  &80.0$\pm${1.0}  &85.3$\pm${1.2}  &90.1$\pm${0.3}  &75.7$\pm${2.8}  &\textbf{92.2$\pm${0.2} }

 \\ \cline{2-11} &Worst-case 
 &41.1$\pm${7.4}  &45.8$\pm${0.5}  &39.0$\pm${2.3}  &37.8$\pm${2.6}  &18.8$\pm${6.1}  &72.2$\pm${5.5}  &74.3$\pm${0.4}  &22.4$\pm${4.3}  &\textbf{91.9$\pm${0.2} }

\\ \hline \hline

\multirow{9}{*}{\begin{tabular}[c]{@{}c@{}} 
F\\M\\N
\\ I\\S\\T \end{tabular} }

 &No attack  &84.6$\pm${0.5}  &89.5$\pm${0.3}  &88.2$\pm${0.4}  &88.6$\pm${0.4}  &87.3$\pm${0.3}  &88.2$\pm${0.5}  &90.3$\pm${0.2}  &90.0$\pm${0.3}  &90.1$\pm${0.3}

 \\ \cline{2-11} &Lie \cite{baruch2019little}  &85.3$\pm${0.6}  &60.7$\pm${12.9}  &50.5$\pm${22.7}  &68.9$\pm${6.9}  &55.5$\pm${10.3}  &75.8$\pm${4.3}  &78.4$\pm${5.0}  &89.0$\pm${0.8}  &89.7$\pm${0.6}

 \\ \cline{2-11} &Fang-v1 \cite{fang2020local}  &85.4$\pm${0.3}  &82.5$\pm${1.0}  &73.1$\pm${4.1}  &69.1$\pm${4.6}  &84.8$\pm${0.9}  &86.6$\pm${1.3}  &89.5$\pm${0.6}  &89.7$\pm${0.3}  &89.7$\pm${0.5}

 \\ \cline{2-11} &Fang-v2 \cite{fang2020local}  &10.0$\pm${11.1}  &12.6$\pm${4.8}  &64.2$\pm${3.3}  &46.6$\pm${16.6}  &10.5$\pm${2.4}  &80.1$\pm${2.3}  &71.5$\pm${10.4}  &69.3$\pm${1.4}  &89.3$\pm${0.4}

 \\ \cline{2-11} &Same value \cite{li2019rsa}  &85.6$\pm${0.6}  &59.9$\pm${14.3}  &69.8$\pm${0.8}  &66.5$\pm${3.6}  &86.7$\pm${0.3}  &88.1$\pm${0.7}  &89.4$\pm${0.6}  &88.8$\pm${0.6}  &89.6$\pm${0.4}

 \\ \cline{2-11} &Gaussian \cite{blanchard2017machine}  &85.5$\pm${0.3}  &89.5$\pm${0.5}  &87.5$\pm${0.5}  &88.2$\pm${0.4}  &87.1$\pm${0.8}  &87.9$\pm${0.7}  &69.8$\pm${7.7}  &17.9$\pm${27.3}  &89.7$\pm${0.2}

 \\ \cline{2-11} &sign flipping \cite{data2021byzantine} &85.2$\pm${0.8}  &85.9$\pm${1.2}  &86.2$\pm${0.6}  &86.3$\pm${0.5}  &87.1$\pm${0.4}  &84.0$\pm${1.4}  &76.9$\pm${1.1}  &52.8$\pm${3.0}  &87.5$\pm${0.9}

 \\ \cline{2-11} &label flipping \cite{data2021byzantine} &85.5$\pm${0.2}  &89.6$\pm${0.3}  &72.6$\pm${5.1}  &78.5$\pm${4.8}  &87.1$\pm${0.3}  &87.9$\pm${0.6}  &89.8$\pm${0.3}  &89.9$\pm${0.4}  &89.1$\pm${0.3}

 \\ \cline{2-11} &minic  \cite{karimireddy2020byzantine} &79.7$\pm${0.6}  &81.2$\pm${0.4}  &83.0$\pm${1.3}  &86.7$\pm${0.8}  &79.8$\pm${0.7}  &87.7$\pm${0.2}  &88.8$\pm${0.2}  &89.1$\pm${0.4}  &89.4$\pm${0.7}

 \\ \cline{2-11} &Collusion (Ours)  &84.7$\pm${0.5}  &35.9$\pm${22.7}  &71.2$\pm${6.3}  &69.5$\pm${1.8}  &86.8$\pm${0.6}  &88.4$\pm${0.3}  &85.5$\pm${1.7}  &76.9$\pm${13.6}  &90.0$\pm${0.5} 

 \\ \cline{2-11} &Averaged 
 &77.2$\pm${1.6}  &68.7$\pm${5.8}  &74.6$\pm${4.5}  &74.9$\pm${4.0}  &75.3$\pm${1.7}  &85.5$\pm${1.2}  &83.0$\pm${2.8}  &75.3$\pm${4.8}  &\textbf{89.4$\pm${0.5} }

 \\ \cline{2-11} &Worst-case 
 &10.0$\pm${11.1}  &12.6$\pm${4.8}  &50.5$\pm${22.7}  &46.6$\pm${16.6}  &10.5$\pm${2.4}  &75.8$\pm${4.3}  &69.8$\pm${7.7}  &17.9$\pm${27.3}  &\textbf{87.5$\pm${0.9} }

\\ \hline \hline

\multirow{9}{*}{\begin{tabular}[c]{@{}c@{}} 
C\\I
\\ F\\A\\R\\10 \end{tabular} }

 &No attack  &64.6$\pm${0.1}  &64.9$\pm${2.1}  &39.0$\pm${25.2}  &66.6$\pm${1.6}  &29.5$\pm${6.6}  &68.3$\pm${0.0}  &69.7$\pm${0.9}  &68.6$\pm${0.6}  &68.0$\pm${0.4}

 \\ \cline{2-11} &Lie \cite{baruch2019little}  &63.9$\pm${2.1}  &10.0$\pm${0.2}  &10.2$\pm${0.2}  &9.8$\pm${0.2}  &10.0$\pm${0.1}  &12.4$\pm${4.2}  &11.9$\pm${3.3}  &20.4$\pm${22.0}  &68.4$\pm${0.4}

 \\ \cline{2-11} &Fang-v1 \cite{fang2020local}  &64.4$\pm${0.5}  &63.9$\pm${0.3}  &11.2$\pm${2.1}  &17.5$\pm${12.3}  &17.3$\pm${3.4}  &68.1$\pm${0.8}  &66.8$\pm${0.3}  &68.0$\pm${0.9}  &67.7$\pm${3.9}

 \\ \cline{2-11} &Fang-v2 \cite{fang2020local}  &9.9$\pm${0.1}  &15.0$\pm${4.1}  &11.0$\pm${1.3}  &11.3$\pm${1.4}  &10.0$\pm${0.1}  &61.8$\pm${4.2}  &67.0$\pm${0.6}  &53.4$\pm${12.7}  &68.5$\pm${1.5}

 \\ \cline{2-11} &Same value \cite{li2019rsa}  &64.8$\pm${0.7}  &50.2$\pm${0.3}  &10.0$\pm${0.0}  &10.0$\pm${0.0}  &54.8$\pm${0.2}  &29.2$\pm${33.3}  &65.1$\pm${5.2}  &65.3$\pm${11.3}  &66.4$\pm${1.8}

 \\ \cline{2-11} &Gaussian \cite{blanchard2017machine}  &63.5$\pm${1.1}  &62.1$\pm${4.3}  &48.4$\pm${4.8}  &65.4$\pm${0.2}  &17.1$\pm${6.2}  &68.1$\pm${0.6}  &31.8$\pm${2.3}  &12.6$\pm${1.3}  &65.3$\pm${0.7}

 \\ \cline{2-11} &sign flipping \cite{data2021byzantine} &65.0$\pm${0.5}  &58.6$\pm${12.5}  &16.6$\pm${5.6}  &51.9$\pm${1.7}  &40.2$\pm${14.9}  &28.7$\pm${32.4}  &47.5$\pm${9.4}  &17.7$\pm${22.6}  &63.1$\pm${5.5}

 \\ \cline{2-11} &label flipping \cite{data2021byzantine} &63.7$\pm${1.4}  &36.1$\pm${22.7}  &18.6$\pm${3.3}  &50.2$\pm${2.1}  &13.0$\pm${2.7}  &61.3$\pm${2.5}  &57.1$\pm${0.8}  &58.3$\pm${1.5}  &63.8$\pm${2.0}

 \\ \cline{2-11} &minic  \cite{karimireddy2020byzantine} &44.1$\pm${0.9}  &65.0$\pm${0.5}  &56.4$\pm${1.0}  &66.0$\pm${0.8}  &42.2$\pm${2.8}  &48.0$\pm${32.9}  &65.9$\pm${1.7}  &66.6$\pm${3.6}  &68.8$\pm${0.2}

 \\ \cline{2-11} &Collusion (Ours)  &63.3$\pm${0.6}  &10.0$\pm${0.1}  &9.7$\pm${0.1}  &10.2$\pm${0.0}  &19.4$\pm${1.2}  &67.9$\pm${0.7}  &29.5$\pm${13.4}  &10.5$\pm${0.3}  &67.0$\pm${1.7} 

 \\ \cline{2-11} &Averaged 
 &56.7$\pm${0.8}  &43.6$\pm${4.7}  &23.1$\pm${4.4}  &35.9$\pm${2.0}  &25.3$\pm${3.8}  &51.4$\pm${11.2}  &51.2$\pm${3.8}  &44.1$\pm${7.7}  &\textbf{66.7$\pm${1.8} }

 \\ \cline{2-11} &Worst-case 
 &9.9$\pm${0.1}  &10.0$\pm${0.2}  &9.7$\pm${0.1}  &9.8$\pm${0.2}  &10.0$\pm${0.1}  &12.4$\pm${4.2}  &11.9$\pm${3.3}  &10.5$\pm${0.3}  &\textbf{63.1$\pm${5.5}}

\\ \hline

\end{tabular}
\label{tab:main-result}
\vspace{-6pt}
\end{table*}

This section illustrates the proposed FedCut method's experimental results compared with eight existing Byzantine-robust methods. We refer to Appendix A and B for the full report of extensive experimental results. 
%Experiment setups are briefly summarized below. 
%only a  Extensive experiments 

\subsection{Setup and Evaluation Metrics} \label{subsec:experiment-setup}

\begin{itemize}[itemsep=.0\baselineskip]
    \item \textbf{Models}: \textit{logistic regression}, \textit{LeNet} \cite{lecun2015lenet} and \textit{AlexNet} \cite{iandola2016squeezenet} three models are used in all experiment settings. (see results for other models and datasets in supplementary material). 

    \item \textbf{Datasets}: \textit{MNIST} \cite{lecun1998gradient}, \textit{Fashion-MNIST} \cite{xiao2017fashion} and \textit{CIFAR10} \cite{krizhevsky2009learning} are used for image classification tasks. To simulate Non-IID settings, class labels assigned to clients follows a Dirichlet distribution $Dir(\beta)$ \cite{li2021federated}. 

    %\item \textbf{Federated learning}: 
    \item \textbf{Federated Learning Settings:} We simulate a horizontal federated learning system with K = 100 clients in a stand-alone machine with 8 Tesla V100-SXM2 32 GB GPUs and 72 cores of Intel(R) Xeon(R) Gold 61xx CPUs. In each communication round, the clients update the weight updates, and the server adopts \textit{Fedavg} \cite{mcmahan2017communication} algorithm to aggregate the model updates. The detailed experimental hyper-parameters are listed in Appendix A.

    \item \textbf{Byzantine attacks}: 
   We set 10\%, 20\% and 30\% clients,
   i.e., 10, 20, and 30 out of 100 clients are Byzantine attackers. The following attacking methods are used in experiments: 
   \begin{itemize}
       \item the \textit{same value attack}: model updates of attackers
        are replaced by the all ones' vector;
      \item the \textit{sign flipping attack}: local gradients of attackers are shifted by a scaled value -4;
     \item the \textit{gaussian attack}: local gradients at clients are replaced by independent Gaussian random vectors $\calN(0,200)$;
     \item  the \textit{Lie attack}: it was designed in \cite{baruch2019little};
     \item  the \textit{Fang-v1 attack} and \textit{Fang-v2 attack}: they were designed in \cite{fang2020local} for coordinate-wise trimmed mean \cite{yin2018byzantine} and Krum \cite{blanchard2017machine} respectively;
     \item Mimic attack: Colluders may mimic the behaviours of some benign clients towards over-emphasizing some clients and under-representing others \cite{karimireddy2020byzantine}.
     \item Our designed \textit{multi-collusion attack}: adversaries are separated into 4 groups, and the same group has similar values. For example, each group is sampled from $\calN(\mu+ \mu_i, 0.0001)$, and different groups have different $\mu_i$, where $\mu$ is the mean of uploaded gradients of all other benign clients.
   \end{itemize}
%   the \textit{gaussian attack}, where local gradients at clients are replaced by independent Gaussian random vectors $\calN(0,200)$; the \textit{Lie attack}, which was designed in \cite{baruch2019little}; the \textit{Fang-v1 attack}, which was designed in \cite{fang2020local}  for coordinate-wise trimmed mean \cite{yin2018byzantine} and Krum \cite{blanchard2017machine} (Fang-v2 attack); the \textit{Fang-v2 attack}, which was designed in \cite{fang2020local} for Krum \cite{blanchard2017machine}; Our designed \textit{collusion attack} that adversaries are separated into 4 groups and same group has the similar values. 

    % Up to \textit{100 clients} participate in a \textit{horizontal federated learning}. Up to 30\% of clients are Byzantine attackers who may launch one of \textit{eight} attacking methods i.e., Lie \cite{baruch2019little}, Fang-v1 (desgin for trimmed mean) \cite{fang2020local}, fang-v2 (design for krum) \cite{fang2020local}, same value attack \cite{li2019rsa}, Gaussian attack \cite{blanchard2017machine}, sign flipping \cite{data2021byzantine}, label flipping \cite{data2021byzantine} and our designed collusion attack with four-party colluders, which can be classified as \textit{single-party collusion} and \textit{multi-party collusion} according to graph analysis in Sect. \ref{subsec:spectral property}. 
    
    \item \textbf{Byzantine-resilient methods}: \textit{Nine} existing methods i.e., Statistic-based methods: Krum \cite{blanchard2017machine}, Median \cite{yin2018byzantine}, Trimmed Mean \cite{yin2018byzantine}, Bulyan \cite{guerraoui2018hidden} and DnC \cite{shejwalkar2021manipulating}, Serve-evaluating methods: FLtrust \cite{cao2020fltrust}, Clustering-based methods: Kmeans \cite{shen2016auror} and the proposed method \textit{FedCut} are compared in terms of following metrics. 
    
    \item \textbf{Gaussian kernel scaling parameters of FedCut}: the preselected set of $\sigma$: $\{\sigma_1, \cdots, \sigma_n\}$ in Algo. \ref{algo:para_deter} is the geometric sequence with common ratio $2$ of $(1, 10\sqrt{10})$ for MNIST, $(1, 10\sqrt{10})$ for Fashion-MNIST and $(0.1, 10\sqrt{10})$ for CIFAR10.

\end{itemize}
\begin{figure*}[bp]
\vspace{-3pt}
	\centering
	\subfigure{
		\begin{minipage}[b]{0.3\textwidth}
\includegraphics[width=1\textwidth]{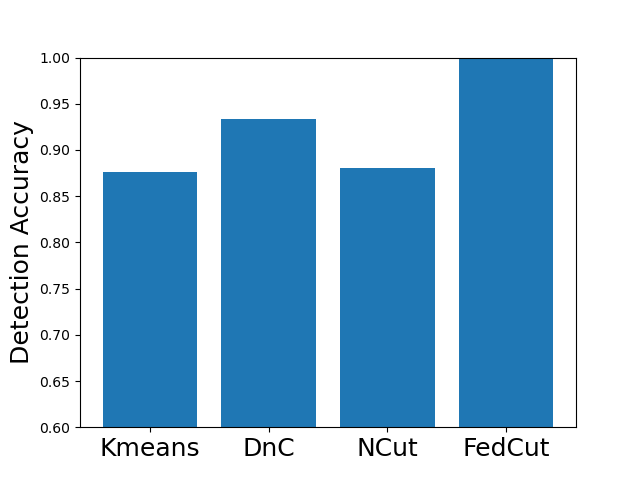}
		\end{minipage}
	}
    	\subfigure{
    		\begin{minipage}[b]{0.3\textwidth}
  		 	\includegraphics[width=1\textwidth]{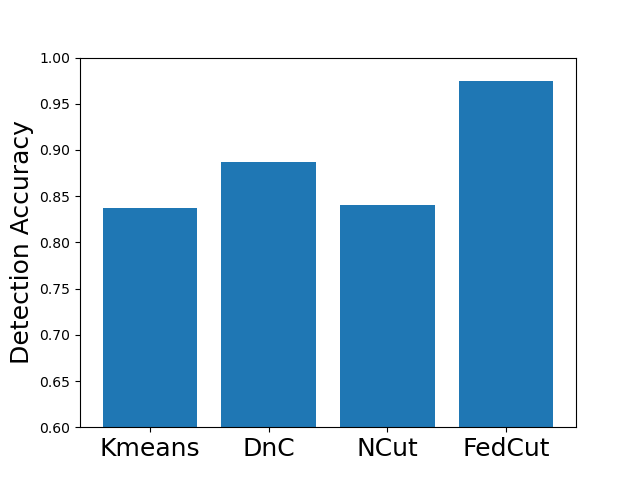}
    		\end{minipage}
    	}
    		\subfigure{
    		\begin{minipage}[b]{0.3\textwidth}
  		 	\includegraphics[width=1\textwidth]{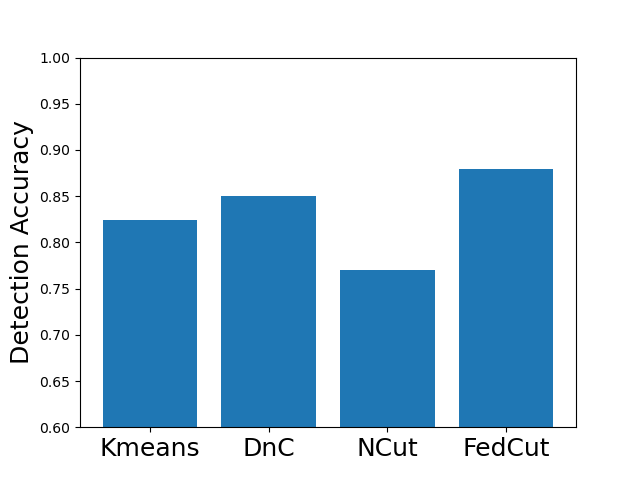}
    		\end{minipage}
    	}

	\subfigure{
		\begin{minipage}[b]{0.3\textwidth}
\includegraphics[width=1\textwidth]{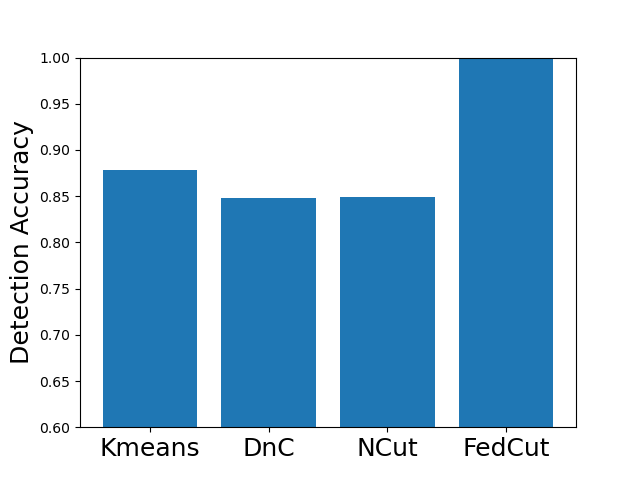}
		\end{minipage}
	}
    	\subfigure{
    		\begin{minipage}[b]{0.3\textwidth}
  		 	\includegraphics[width=1\textwidth]{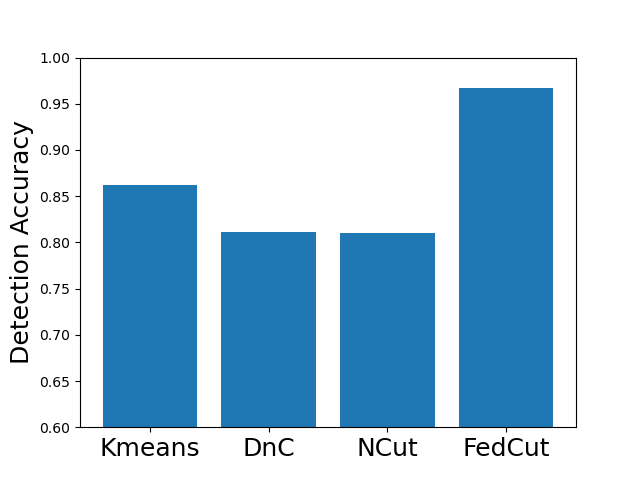}
    		\end{minipage}
    	}
    		\subfigure{
    		\begin{minipage}[b]{0.3\textwidth}
  		 	\includegraphics[width=1\textwidth]{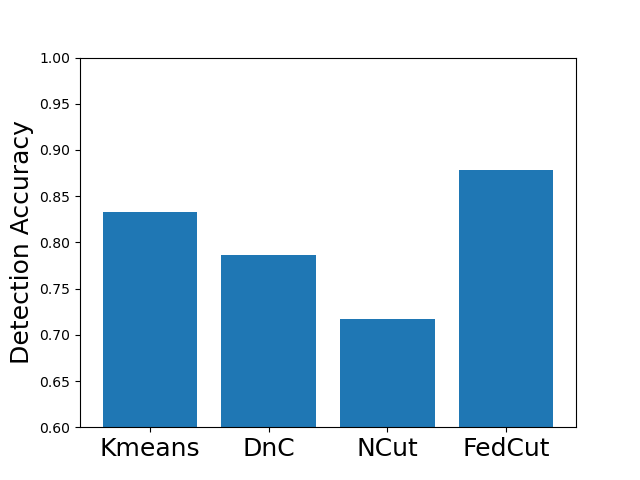}
    		\end{minipage}
    	}
    	
	\caption{The Detection accuracy of different Byzantine-resilient methods (Kmeans, DnC, NCut and FedCut) among all 8 attacks with different Non-IID extent (left: IID, middle: Non-IID with $\beta =0.5$, right: Non-IID with $\beta=0.1$) on MNIST (first row) and Fashion MNIST (second row). }
	\label{fig:detection-acc}
	\vspace{-3pt}
\end{figure*}

\textbf{Evaluation metric:} two types of metrics are used in our evaluation.
\begin{itemize}[itemsep=.0\baselineskip]
    \item \textit{Model Performance (MP)} of the federated model is used to evaluate defending capabilities of different methods. 
    In order to elucidate robustness of each defending method, we also report respective \textit{averaged} and \textit{worst-case} model performances under all possible attacks. 
    
    %In order to cater for difference of baseline MPs between IID and non-IID settings, we use \textit{normalized model performance (NMP)}, which is defined as the \textit{ratio} of the MP over MP of FedAvg\cite{?} under no attacks, to indicate defending capabilities of different methods. Finally, we may use \textit{averaged NMP (NMP)} under various settings to report the averaged defending capabilities of different methods. 

    \item For methods including DnC \cite{shejwalkar2021manipulating}, Kmeans \cite{shen2016auror} and FedCut, which detect benign clients before aggregation, \textit{detection accuracy} among all iterations is used to quantify their defending capabilities. Note that this accuracy is measured against \textit{ground-truth Byzantine attackers' membership}, and it should not be confused with the image classification accuracy of the main task. 
\end{itemize}

\subsection{Comparison with Other Byzantine-resilient Methods}

Tab. \ref{tab:main-result} summarizes Model Performances (MP) of 8 existing methods as well as the proposed FedCut method for classification of MNIST, Fashion MNIST, and CIFAR10 using logistic regression, LeNet, and AlexNet respectively under IID setting and 30 attackers (see Appendix B for more results with other settings). There are four noticeable observations. 
\begin{itemize}
    \item FedCut performs robustly under almost all attacks with worst-case MP above 92.0\% for MNIST, 87.5\% for Fashion MNIST and 63.1\% for CIFAR10.
    % , except for the \textit{sign flipping} attack, which caused a minor \textit{worst-case MP degradation} by 1.6\% with MP down to 87.5\%. 
    This robust performance is in sharp contrast with eight existing methods, each of which has a \textit{worst-case MP degradation} ranging from 18.2\% (i.e., DnC by Gaussian attack reaching 74.3\% MP)
    %20\% (i.e DnC by \textit{sign flipping} attack under IID setting) 
    to 73.7\% (i.e., Bulyan by Fang-v2 merely having 18.8\% MP) on MNIST. 
    \item  In terms of averaged MP, it is clearly observed that FedCut outperformed eight existing methods by noticeable margins ranging between 2.1\% to 16.5\% on MNIST, 3.9\% to 20.7\% on Fashion MNIST and 10.0\% to 43.6\% on CIFAR10. 
    \item Robust Statistics based methods (e.g., Krum, GeoMedian, Median, and Trimmed Mean) were overwhelmed by Collusion-diff and Collusion-mimic attacks (such as the MP drops 49.1\% for Median under Fang-v1 attack and the MP drops 51.4\% for Krum under Fang-v2 attack), incurred a significant MP degradation on MNIST. In contrast, all Collusion attacks cause a minor MP loss of less than 0.5\% to FedCut. 
    \item Existing clustering-based method, i.e., Kmeans performs robustly in the presence of single-group colluders but with MP degradation less than 2\% (e.g., Kmeans for Same Value attack on MNIST, Fashion MNIST and CIFAR10), but incurred significant MP degradation more than 50\% in the face of multi-group colluders attackers (e.g., Kmeans for the Multi-Collusion attack on MNIST and Gaussian attack on Fashion MNIST). In contrast, multi-group collusion doesn't cause significant MP loss to the proposed FedCut, which adopts spectral heuristics to make an estimation of the number of colluder groups (see Sect. \ref{subsec: spectal-heu}). Correspondingly, the detection accuracy of FedCut is higher than other methods, as Fig. \ref{fig:detection-acc} shows. For instance, the detection accuracy of FedCut on MNIST in the IID setting is 99.9\%, while the detection accuracy of DnC and Kmeans drop to 93.6\% and 87.5\% separately.

\end{itemize}

% Finally, the computational complexity of our FedCut is  $\tilde{O}(K^2d+ K^3)$, which is much less than $\tilde{O}(d^3+K^2d)$ for methods leveraging SVD \cite{data2021byzantine, shejwalkar2021manipulating}.

% \begin{figure*}[htbp]
% 	\centering
% 	\subfigure{
% 		\begin{minipage}[b]{0.3\textwidth}
% \includegraphics[width=1\textwidth]{imgs/detection_acc/MNIST-iid-30.png}
% 		\end{minipage}
% 	}
%     	\subfigure{
%     		\begin{minipage}[b]{0.3\textwidth}
%   		 	\includegraphics[width=1\textwidth]{imgs/detection_acc/MNIST-noniid0.5-30.png}
%     		\end{minipage}
%     	}
%     		\subfigure{
%     		\begin{minipage}[b]{0.3\textwidth}
%   		 	\includegraphics[width=1\textwidth]{imgs/detection_acc/MNIST-noniid0.1-30.png}
%     		\end{minipage}
%     	}

% 	\subfigure{
% 		\begin{minipage}[b]{0.3\textwidth}
% \includegraphics[width=1\textwidth]{imgs/detection_acc/FMNIST-iid-30.png}
% 		\end{minipage}
% 	}
%     	\subfigure{
%     		\begin{minipage}[b]{0.3\textwidth}
%   		 	\includegraphics[width=1\textwidth]{imgs/detection_acc/FMNIST-noniid0.5-30.png}
%     		\end{minipage}
%     	}
%     		\subfigure{
%     		\begin{minipage}[b]{0.3\textwidth}
%   		 	\includegraphics[width=1\textwidth]{imgs/detection_acc/FMNIST-noniid0.1-30.png}
%     		\end{minipage}
%     	}
    	
% 	\caption{\small Average Detection accuracy of different Byzantine-resilient methods (Kmeans, DnC, NCut and FedCut) among all 8 attacks with different Non-IID extent (left: IID, middle: Non-IID with $\beta =0.5$, right: Non-IID with $\beta=0.1$) on MNIST (first row) and Fashion MNIST (second row). }
% 	\label{fig:detection-acc}
% \end{figure*}
\begin{figure*}[bp]
\vspace{-2pt}
	\centering
	\subfigure{
		\begin{minipage}[b]{0.4\textwidth}
\includegraphics[width=1\textwidth]{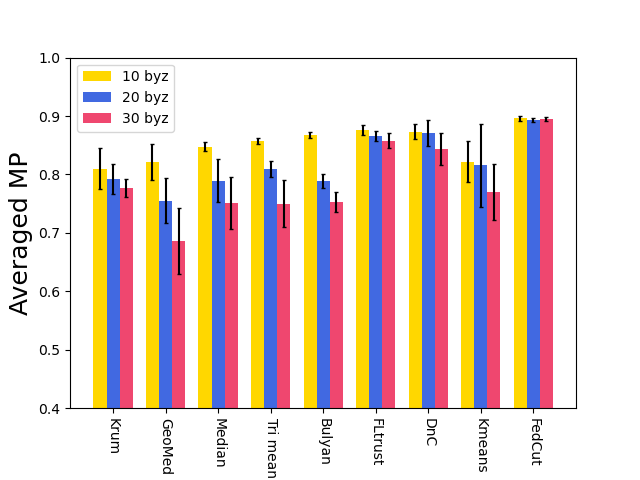}
		\end{minipage}
	}
    	\subfigure{
    		\begin{minipage}[b]{0.4\textwidth}
  		 	\includegraphics[width=1\textwidth]{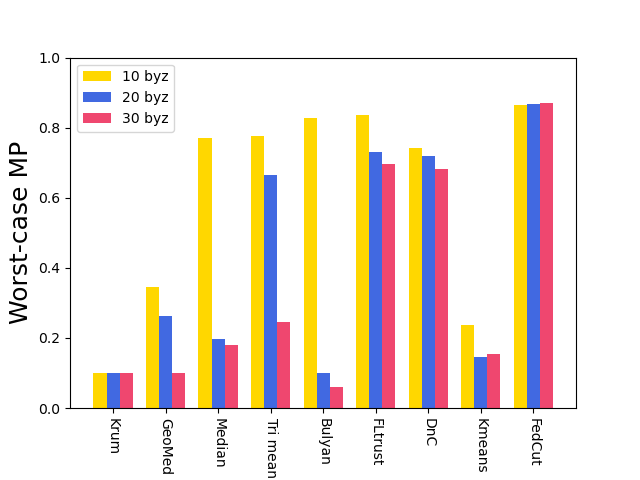}
    		\end{minipage}
    	}
    	
%     		\subfigure{
% 		\begin{minipage}[b]{0.22\textwidth}
% \includegraphics[width=1\textwidth]{imgs/avg-worst-MP/lenet_avg_iid.png}
% 		\end{minipage}
% 	}
%     	\subfigure{
%     		\begin{minipage}[b]{0.22\textwidth}
%   		 	\includegraphics[width=1\textwidth]{imgs/avg-worst-MP/lenet_worst_iid.png}
%     		\end{minipage}
%     	}

	\caption{ Averaged and worst-case model performance among all 8 attacks of different Byzantine-resilient methods with IID setting for different byzantine numbers (yellow: 10 attackers, blue: 20 attackers, green: 30 attackers) on Fashion MNIST. }
	\vspace{-2pt}
	\label{fig:robust-byznum}
\end{figure*}

\begin{figure*}[bp]
\vspace{-3pt}
	\centering
% 	\subfigure{
% 		\begin{minipage}[b]{0.45\textwidth}
% \includegraphics[width=1\textwidth]{imgs/avg-worst-MP/lr_avg_30.png}
% 		\end{minipage}
% 	}
%     	\subfigure{
%     		\begin{minipage}[b]{0.45\textwidth}
%   		 	\includegraphics[width=1\textwidth]{imgs/avg-worst-MP/lr_worst_30.png}
%     		\end{minipage}
%     	}
    	
    		\subfigure{
		\begin{minipage}[b]{0.4\textwidth}
\includegraphics[width=1\textwidth]{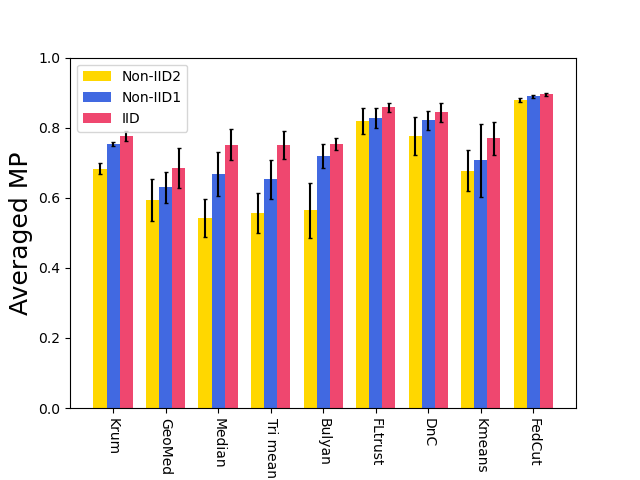}
		\end{minipage}
	}
    	\subfigure{
    		\begin{minipage}[b]{0.4\textwidth}
  		 	\includegraphics[width=1\textwidth]{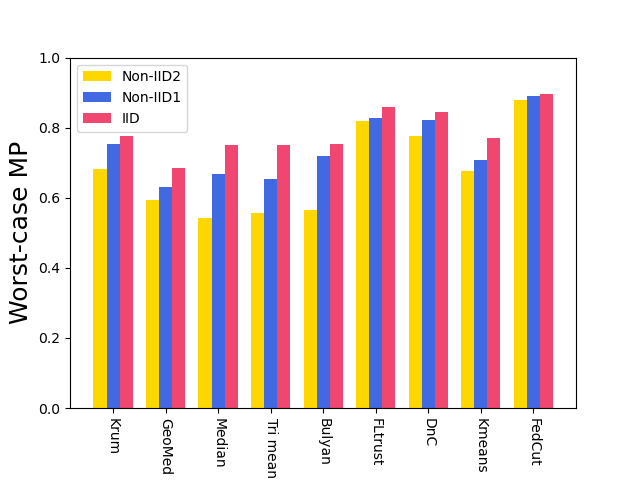}
    		\end{minipage}
    	}
    	
%     		\subfigure{
% 		\begin{minipage}[b]{0.22\textwidth}
% \includegraphics[width=1\textwidth]{imgs/avg-worst-MP/alnex_avg_30.png}
% 		\end{minipage}
% 	}
%     	\subfigure{
%     		\begin{minipage}[b]{0.22\textwidth}
%   		 	\includegraphics[width=1\textwidth]{imgs/avg-worst-MP/alnex_worst_30.png}
%     		\end{minipage}
%     	}

	\caption{ Averaged and worst-case model performance among all 8 attacks of different Byzantine-resilient methods with 30 Byzantine attackers for different Non-IID extent (yellow: Non-IID with $\beta = 0.1$, blue: Non-IID with $\beta = 0.5$, green: IID) on Fashion MNIST. }
	\label{fig:robust-heter}
	\vspace{-2pt}
\end{figure*}

\subsection{Robustness}
In this subsection, we demonstrate the robustness of our FedCut framework under different byzantine numbers, the Non-IID extent of clients' local data, and multiple collusion attacks.

\subsubsection{Robustness under Varying Numbers of Attackers.}
Fig. \ref{fig:robust-byznum} shows the model performance (MP) for different Byzantine-resilient methods under different types of attacks for different Byzantine attackers, i.e., 10, 20, and 30 attackers. The result shows that the MP of FedCut doesn't drop with the increase of Byzantine numbers while the MP of others drops seriously (e.g., the averaged MP of Median \cite{yin2018byzantine} drops to 74.6\% from 88.6\% in Fashion MNIST dataset). Clearly, our proposed method, FedCut, is robust for the the number of byzantine attackers.

\subsubsection{Robustness under Heterogeneous dataset}
Fig. \ref{fig:robust-heter} shows the model performance (MP) for different Byzantine resilient methods under different types of attacks for different Non-IID extent of clients' datasets. The result demonstrates that the MP of FedCut drops no more than 3\% as the clients' local dataset becomes more heterogeneous while the MP of others drops seriously (e.g., the averaged MP of trimmed mean \cite{yin2018byzantine} drops to 55.4\% from 74.9\% in Fashion MNIST dataset)

\begin{table*}[htbp] 
\vspace{-6pt}
\centering 
\setlength{\tabcolsep}{1.0mm}
\renewcommand\arraystretch{1.5}
\caption{Computation Time complexity for different Byzantine resilient methods, where $K$ is the number of client, $d$ is the dimension of updates, $T'$ is training time of the server' data for FLtrust.}
\begin{tabular}{|c||c|c|c|c|c|}
\hline

\textbf{Defense}&
Krum  \cite{blanchard2017machine}  &GeoMedian \cite{chen2017distributed}&Median \cite{yin2018byzantine}&Trimmed \cite{yin2018byzantine}&Bulyan  \cite{guerraoui2018hidden} \\ \hline 
\begin{tabular}[c]{@{}c@{}}\textbf{Time}\\ \textbf{Complexity}\end{tabular}&$\calO(K^2d)$&$\calO(K^2d)$&$\calO(K^2d +)$&$\calO(K^2d)$&$\calO( Kd)$\\ \hline \hline
\textbf{Defense}&FLtrust  \cite{cao2020fltrust}&DnC \cite{shejwalkar2021manipulating}&Kmeans \cite{shen2016auror}&FedCut (Ours)&  \\ \hline 

\begin{tabular}[c]{@{}c@{}}\textbf{Time}\\ \textbf{Complexity}\end{tabular}&$\calO( Kd + T’)$&$\calO(d^3+ K^2d)$&$\calO(K^2d)$&$\calO(K^2d + K^3)$& \\ \hline

\end{tabular}
\label{tab:time-complexity}
\vspace{-6pt}
\end{table*}

\subsubsection{Robustness against the Multi-Collusion attack}
In the former part, we design the 'Multi-Collusion attack' in which adversaries are separated into four groups, and the same group has similar values. For example, each group is sampled from $\calN(\mu+ \mu_i, 0.0001)$, and different groups have different $\mu_i$, where $\mu$ is the mean of uploaded gradients of all other benign clients. We further design more 'collusion attacks' with a various number of colluders' parties as follows: 30 attackers are separated into 1) two groups with each group having 15 attackers; 2) three groups with each group having ten attackers 3) four groups with each group having 8, 8, 7, 7 attackers respectively; 4) five groups with each group has six attackers.

Tab. \ref{tab:collusion1} and \ref{tab:collusion2} displays the model performances under different 'collusion attacks,' which illustrates that various collusion attacks do not influence FedCut and NCut while other byzantine resilient methods such as Kmeans are affected seriously by collusion attack. Note that DnC performs well with 2-groups collusion, but MP drops when colluders' groups increase.

\begin{table}[htbp] 
\vspace{-5pt}
\tiny
\centering 
\setlength{\tabcolsep}{0.5mm}
\renewcommand\arraystretch{2.0}
\caption{Model performances of different Byzantine-resilient methods under different groups of collusion attacks (with Non-IID setting $\beta = 0.5$ and 30 attackers for classification of MNIST). $i-$ groups represent the number of colluders.}
\begin{tabular}{|c||c|c|c|c|c|c|c|c|c|}
\hline

&\begin{tabular}[c]{@{}c@{}} 
Krum \\ \cite{blanchard2017machine} \end{tabular} &\begin{tabular}[c]{@{}c@{}}GeoMedian\\ \cite{chen2017distributed}\end{tabular}&\begin{tabular}[c]{@{}c@{}}Median\\ \cite{yin2018byzantine}\end{tabular}&\begin{tabular}[c]{@{}c@{}}Trimmed\\ \cite{yin2018byzantine}\end{tabular}&\begin{tabular}[c]{@{}c@{}} Bulyan \\ \cite{guerraoui2018hidden}\end{tabular}&\begin{tabular}[c]{@{}c@{}} FLtrust \\ \cite{cao2020fltrust}\end{tabular}&\begin{tabular}[c]{@{}c@{}} DnC\\ \cite{shejwalkar2021manipulating}\end{tabular}&\begin{tabular}[c]{@{}c@{}}Kmeans\\ \cite{shen2016auror}\end{tabular}&\begin{tabular}[c]{@{}c@{}}FedCut\\ (Ours)\end{tabular} \\ \hline \hline

2-groups
 &89.2$\pm${0.4}  &79.2$\pm${0.0}  &71.9$\pm${0.1}  &69.0$\pm${0.2}  &89.1$\pm${0.2}  &88.5$\pm${0.1}  &\textbf{91.9$\pm${0.1}}  &45.7$\pm${0.0}   &91.7$\pm${0.0}

\\ \hline 3-groups
  &88.9$\pm${0.1}  &81.1$\pm${0.0}  &73.4$\pm${0.0}  &70.7$\pm${0.1}  &88.9$\pm${0.1}  &88.0$\pm${0.1}  &87.2$\pm${0.8}  &45.9$\pm${0.1} &\textbf{91.6$\pm${0.1}}

\\ \hline 4-groups
  &89.0$\pm${0.5}  &82.2$\pm${1.2}  &73.6$\pm${1.2}  &72.3$\pm${1.2}  &86.7$\pm${1.2}  &89.2$\pm${0.5}  &89.4$\pm${0.5}  &31.2$\pm${2.2}  &\textbf{92.1$\pm${0.1}}

\\ \hline 5-groups
 &88.8$\pm${1.1}  &84.2$\pm${0.1}  &78.0$\pm${0.1}  &75.1$\pm${0.0}  &88.7$\pm${0.1}  &88.4$\pm${0.1}  &89.0$\pm${0.3}  &54.1$\pm${0.0}   &\textbf{91.7$\pm${0.1}}

\\ \hline

\end{tabular}
\label{tab:collusion2}
\vspace{-6pt}
\end{table}

\begin{table}[htbp] 
\tiny
\centering 
\setlength{\tabcolsep}{0.5mm}
\renewcommand\arraystretch{2.0}
\caption{Model performances of different Byzantine-resilient methods under different groups of collusion attacks (with IID and 30 attackers for classification of MNIST). $i-$ groups represent the number of colluders' clusters.}
\begin{tabular}{|c||c|c|c|c|c|c|c|c|c|c|}
\hline

&\begin{tabular}[c]{@{}c@{}} 
Krum \\ \cite{blanchard2017machine} \end{tabular} &\begin{tabular}[c]{@{}c@{}}GeoMedian\\ \cite{chen2017distributed}\end{tabular}&\begin{tabular}[c]{@{}c@{}}Median\\ \cite{yin2018byzantine}\end{tabular}&\begin{tabular}[c]{@{}c@{}}Trimmed\\ \cite{yin2018byzantine}\end{tabular}&\begin{tabular}[c]{@{}c@{}} Bulyan \\ \cite{guerraoui2018hidden}\end{tabular}&\begin{tabular}[c]{@{}c@{}} FLtrust \\ \cite{cao2020fltrust}\end{tabular}&\begin{tabular}[c]{@{}c@{}} DnC\\ \cite{shejwalkar2021manipulating}\end{tabular}&\begin{tabular}[c]{@{}c@{}}Kmeans\\ \cite{shen2016auror}\end{tabular}&\begin{tabular}[c]{@{}c@{}}FedCut\\ (Ours)\end{tabular} \\ \hline \hline

2-groups
 &90.8$\pm${0.0}  &80.2$\pm${0.4}  &73.3$\pm${0.1}  &72.7$\pm${0.3}  &90.1$\pm${0.0}  &88.8$\pm${0.1}  &\textbf{92.0$\pm${0.1}} &46.4$\pm${0.1}  &\textbf{91.8$\pm${0.0}}

\\ \hline 3-groups
  &90.6$\pm${0.1}  &82.3$\pm${0.4}  &77.2$\pm${0.7}  &75.7$\pm${0.3}  &90.1$\pm${0.1}  &88.9$\pm${0.1}  &89.0$\pm${1.4}  &49.5$\pm${0.1} &\textbf{91.8$\pm${0.0}}

\\ \hline 4-groups
 &90.3$\pm${0.2}  &83.9$\pm${2.1}  &79.5$\pm${0.7}  &78.0$\pm${0.8}  &89.7$\pm${0.3}  &89.4$\pm${0.1}  &90.2$\pm${0.3}  &43.2$\pm${8.4}  &\textbf{92.2$\pm${0.2}}

\\ \hline 5-groups
&90.7$\pm${0.1}  &84.4$\pm${0.4}  &79.3$\pm${0.1}  &78.9$\pm${0.1}  &90.1$\pm${0.1}  &88.6$\pm${0.6}  &89.4$\pm${0.1}  &51.9$\pm${0.1} &\textbf{91.9$\pm${0.0}}

\\ \hline

\end{tabular}

\label{tab:collusion1}

\end{table}

\subsection{Computation Complexity}

We compare the computation time complexity for different Byzantine resilient methods for each iteration on the server as Tab. \ref{tab:time-complexity}. It shows that
% 1) the time complexity of robust statistics based Byzantine-resilient methods except DnC is at most $\calO(K^2d)$, e.g., Krum 
% computation complexity for Bulyan \cite{guerraoui2018hidden} is the least ($Kd$); 2) 
the computation complexity of our proposed method FedCut is $\calO(K^2d + K^3)$, where implementing the singular vector decomposition (SVD) needs $\calO(K^3)$ \cite{halko2011finding} and computing the adjacency matrix requires $\calO(K^2d)$.
The time complexity of FedCut is comparable to statistic-based methods such as Krum \cite{blanchard2017machine} ($\calO(K^2d)$) and median \cite{yin2018byzantine} ($\calO(K^2d)$) when the dimension of model $d$ is greatly larger than number of participants $K$.

% In the case of large $K$, some modified methods such as randomized SVD \cite{halko2011finding} improved the time complexity from $\calO(K^3)$ to $\calO(K)$ if the normalized adjacency matrix has low rank. Actually, the normalized adjacency matrix is sparse in practice since server may choose partial model updates to aggregate in FL.

Empirically, the proposed FedCut method allows the training of a DNN model (AlexNet) with 100 clients to run in three hours, under a simulated environment (see Appendix A for details). We regard this time complexity is reasonable for practical federated learning applications across multiple institutions.  For cross-devices application scenarios in which $K$ might be up to millions, randomized SVD \cite{halko2011finding} can be adopted to improve the time complexity from $\calO(K^3)$ to $\calO(K)$ if the normalized adjacency matrix has low rank. Applying randomized SVD to cross-device FL scenarios is one of our future work.

% Actually, the server only select partial $p \ll K$ clients' updates to aggregate in FL
% 3) DnC \cite{shejwalkar2021manipulating} consumes much computing time as $d$ becomes large.

\section{Discussion and Conclusion}
This paper proposed a novel spectral analysis framework, called FedCut, to detect Byzantine colluders robustly and efficiently from the graph perspective. Specifically, our proposed algorithm FedCut ensures the optimal separation of benign clients and colluders in the spatial-temporal graph constructed from uploaded model updates over different iterations. We analyze existing Byzantine attacks and Byzantine-resilient methods in the lens of spectral analysis. It was shown that existing Byzantine-resilient methods may suffer from failure cases in the face of Byzantine colluders. Moreover, spectral heuristics are used in FedCut to determine the number of colluder groups, the scaling factor and mimic colluders, which significantly improves the Byzantine tolerance in colluders detection. 

% We believe this disclosed connection deserves follow-up exploration in understanding group behavior between benign clients and that of colluders. 

Our extensive experiments on multiple datasets and theoretical convergence analysis demonstrate that our proposed framework can achieve drastic improvements over the FL baseline in terms of model performance under various Byzantine attacks.
% Finally, benign clients demonstrate an intriguing \textit{small-world network} (SWN) property which is reminiscent of typical features observed in biological and social networks \cite{watts1998collective}.

Finally, the proposed and many existing Byzantine-resilient methods assume that the server could access to clients' model updates to compute the similarity. However, the leakage of model updates allows a semi-honest server to infer some information of the private data \cite{zhu2019deep}. Therefore, Byzantine problem should be considered when implementing FL for privacy-preserving applications. In future work, we will explore to what extent FedCut can be used in conjunction with Differential Privacy \cite{abadi2016deep} or Homomorphic Encryption mechanisms \cite{zhang2020batchcrypt}. 

\newpage 
\newpage

\bibliographystyle{plain}
\bibliography{refer}

\clearpage
% %%%%%%%%%%%%%%%%%%%%%%%%%%%%%%%%%%%%%%%%%%%%%%%%%%%%%%%%%%%%
\begin{appendices}

\section{Experiment Setting}
\label{appendix:expr}
This section illustrates the experiment settings of the empirical study of the proposed FedCut framework in comparison with existing Byzantine resilient methods. 

\textbf{Model Architectures} The architectures we investigated include the well-known Logistic regression, LeNet \cite{lecun2015lenet} and AlexNet \cite{krizhevsky2012imagenet}. 

\textbf{Dataset} We evaluate Federated learning model performances with classification tasks on standard MNIST, Fashion MNIST \cite{xiao2017fashion} and CIFAR10 dataset. The MNIST database of 10-class handwritten digits has a training set of 60,000 examples, and a test set of 10,000 examples. Fashion MNIST is fashion product of MNIST including 60,000 images and the test set has 10,000 images. The CIFAR-10 dataset consists of 60000 $32\times 32$ colour images in 10 classes, with 6000 images per class. Respectively, we conduct stand image classification tasks of MNIST, Fashion MNIST and CIFAR10 with logistic regression, LeNet and AlexNet. 
According to the way we split the dataset for clients in federated learning, the experiments are divided into IID setting and Non-IID setting. Specifically, We consider \textit{lable-skew} Non-IID federated learning setting, where we assume each client's training examples are drawn with class labels following a \textit{dirichlet distribution} ($Dir(\beta)$) \cite{li2021federated}, with which $\beta > 0$ is the concentration parameter controlling the identicalness among users (we use $\beta=0.1$ \& 0.5 in the following).

\textbf{Federated Learning Settings} We simulate a horizontal federated learning system with K = 100 clients in a stand-alone machine with 8 Tesla V100-SXM2 32 GB GPUs and 72 cores of Intel(R) Xeon(R) Gold 61xx CPUs. In each communication round, the clients update the weight updates, and the server adopts \textit{Fedavg} \cite{mcmahan2017communication} algorithm to aggregate the model updates. The detailed experimental hyper-parameters are listed in Tab. \ref{tab:train-params}.

\begin{table*}[htbp] \huge \renewcommand\arraystretch{1.5}
	\resizebox{1\textwidth}{!}{
		\begin{tabular}{l|ccc}
			\hline
			Hyper-parameter & Logistic Regression& LeNet & AlexNet \\ \hline
			Activation function & ReLU & ReLU & ReLU \\
			Optimization method & Adam & Adam& Adam\\
			
			Learning rate & 0.0005 & 0.001& 0.001\\
			Weight Decay & 0.001 & 0.002& 0.001\\
			Batch size & 32 & 32& 64\\
			Data Distribution & IID and Non-IID ($\beta=0.5,0.1)$ & IID and Non-IID ($\beta=0.5,0.1)$&  IID\\
			Iterations & 2000 & 3000& 3000 \\
			Number of Clients & 100 & 100& 30 \\
			 Numbers of Attackers & 10, 20, 30 & 10, 20, 30& 6\\
			\hline
	\end{tabular}}
	
	\caption{Training parameters}%Training parameters for Federated AlexNet$_{\mathbf{p}}$ and ResNet$_{\mathbf{p}}$-18, respectively ($\dagger$ the learning rate is scheduled as 0.01, 0.001 and 0.0001 between epochs [1-100], [101-150] and [151-200] respectively). }
	\label{tab:train-params}
\end{table*}

\textbf{Attack:}
We set 10\%, 20\% and 30\% clients,
i.e., 10, 20 and 30 out of 100 clients be Byzantine attackers. The following attacking methods are used in experiments:
\begin{enumerate}
\item the 'same value attack', where model updates of attackers
are replaced by the all ones vector; 
\item the 'label flipping attack', where attackers use the wrong label to generate the gradients to upload; 
\item the 'sign flipping attack', where local gradients of attackers are shifted by a scaled value -4; 
\item 
the 'gaussian attack', where local gradients at clients are replaced by independent Gaussian random vectors $\calN(0,200)$. 
\item the ‘Lie attack’, which was designed in \cite{baruch2019little}; 
\item  the 'Fang-v1 attack', which was designed in \cite{fang2020local}  for coordinate-wise trimmed mean \cite{yin2018byzantine} and Krum \cite{blanchard2017machine} (Fang-v2 attack);
\item the 'Fang-v2 attack', which was designed in \cite{fang2020local}  for Krum \cite{blanchard2017machine};
\item Our designed 'collusion attack' that adversaries are separated into 4 groups and same group has the similar values. For example, each groups is sampled from $\calN(\mu+ \mu_i, 0.0001)$ and different groups has different $\mu_i$, where $\mu$ is mean of uploaded gradients of all other benign clients.
\end{enumerate}

% We implement two types of attack: single-party collusion attack. Non-collusion attack includes: 1) the same-value attack, where gradients of the corrupt clients
% are replaced by the all ones vector; 2)  the label flipping attack, where corrupt clients use the wrong label to generate the gradients to upload; 3) the sign flipping, where local gradients of corrupt clients are shifted by a scaled value -4; 4) the random gradient attack, where local gradients at clients are replaced by independent Gaussian random vectors $\cal N(0,200)$. Single-collusion attack includes 1) the ‘Baruch attack’, which was designed in \cite{baruch2019little}; 2) the ‘Fang attack’, which was designed in \cite{fang2020local} specifically for coordinate-wise trimmed mean \cite{yin2018byzantine} and Krum \cite{blanchard2017machine}. We design the collusion attack that adversaries are separated into different groups and same group has the similar values. For example, each groups is sampled from $\calN(\mu+ \mu_i, 0.0001)$ and different groups has different $\mu_i$, where $\mu$ is mean of uploaded gradients of all other benign clients. Specifically, ...

\textbf{Byzantine-resilient methods}
The following Byzantine resilent methods are evaluated in experiments:
\begin{enumerate}
    \item Krum \cite{blanchard2017machine}, we adopt codes from \url{https://github.com/vrt1shjwlkr/NDSS21-Model-Poisoning/tree/main/femnist}
    \item Median \cite{yin2018byzantine}, we adopt codes from \url{https://github.com/Liepill/RSA-Byzantine}
    \item Trimmed Mean \cite{yin2018byzantine}, we adopt codes from \url{https://github.com/vrt1shjwlkr/NDSS21-Model-Poisoning/tree/main/femnist}
    \item Bulyan \cite{guerraoui2018hidden}, we adopt codes from \url{https://github.com/vrt1shjwlkr/NDSS21-Model-Poisoning/tree/main/femnist}
    \item DnC \cite{shejwalkar2021manipulating}, we adopt codes from \url{https://github.com/vrt1shjwlkr/NDSS21-Model-Poisoning/tree/main/femnist}
    \item FLtrust \cite{cao2020fltrust}, we adopt codes from \url{https://people.duke.edu/~zg70/code/fltrust.zip}
    \item Kmeans \cite{shen2016auror}, we adopt codes from \url{https://scikit-learn.org/stable/modules/generated/sklearn.cluster.KMeans.html}
    \item NCut \cite{ng2001spectral}, we adopt NCut in each iteration  separately for ablation study.
    \item FedCut, the proposed FedCut method applied to the Spatial-Temporal network.
\end{enumerate}

\clearpage

\section{Ablation Study} \label{sec:appB}

In this section, we report more experimental results on different extent Non-IID (e.g., $\beta$=0.1 and $\beta$=0.5), number of byzantine attackers (e.g., 10, 20 and 30), model (e.g., logistic regression, LeNet and AlexNet) and dataset (e.g., MNIST, Fashion MNIST and CIFAR10). 
%Moreover, we design more types of collusion attack and compare our proposed FedCut to NCut without considering multiple iterations.
% \begin{itemize}
%     \item AlexNet Cifar10
%     \item Non iid, byzantine attacker
%     \item more collusion
%     \item Time help ncut
%     \item detection acc auc
%     \item iterative filtering methods
% \end{itemize}
\subsection{More Experimental Results}
Tab. \ref{tab:MP1}-\ref{tab:MP19} show the model performance (MP) for different Byzantine resilient methods under different types of attacks in various settings. All results show that the proposed method FedCut achieve the best model performances compared to other byzantine resilient methods when the Non-IID extent increases from $\beta=0.5$ to 0.1 and byzantine attackers numbers increases from 10 to 30. As shown by the ablation study, the MP of NCut is is not as stable as that of the proposed FedCut (e.g., MP drops to 31.7 \% with Non-IID setting and 30 attackers under Fang-v1 attack).

\begin{figure*}[bp]
	\centering
	\vspace{-5pt}
	\subfigure{
		\begin{minipage}[b]{0.3\textwidth}
\includegraphics[width=1\textwidth]{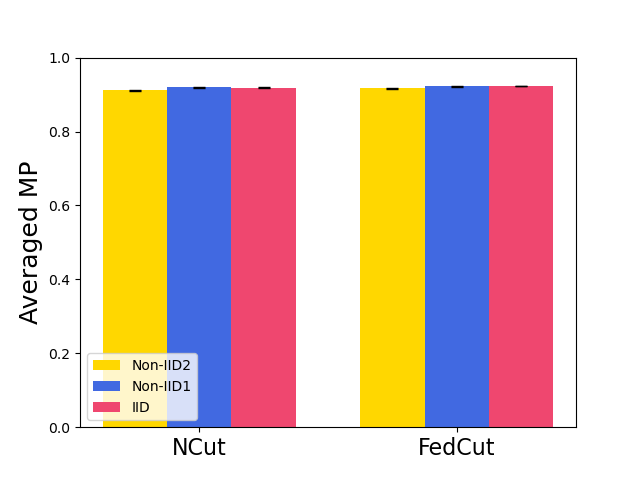}
		\end{minipage}
	}
    	\subfigure{
    		\begin{minipage}[b]{0.3\textwidth}
  		 	\includegraphics[width=1\textwidth]{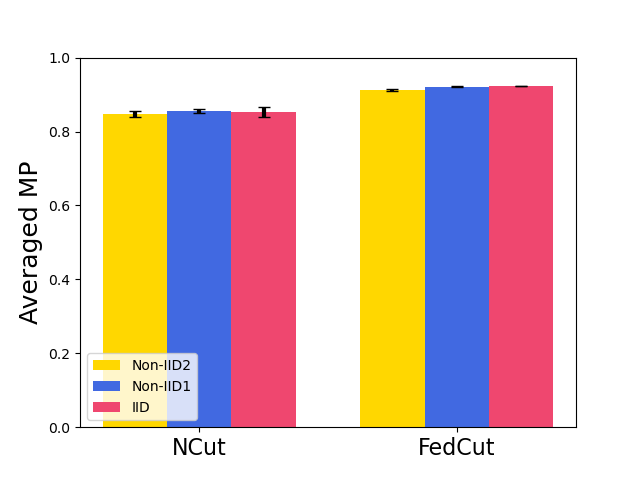}
    		\end{minipage}
    	}
    		\subfigure{
    		\begin{minipage}[b]{0.3\textwidth}
  		 	\includegraphics[width=1\textwidth]{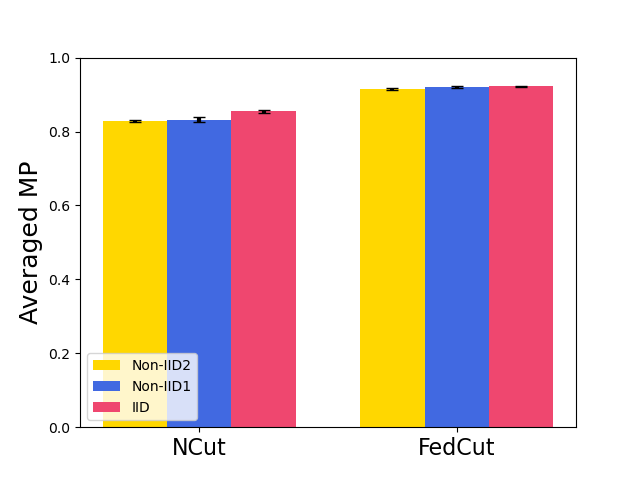}
    		\end{minipage}
    	}

	\subfigure{
		\begin{minipage}[b]{0.3\textwidth}
\includegraphics[width=1\textwidth]{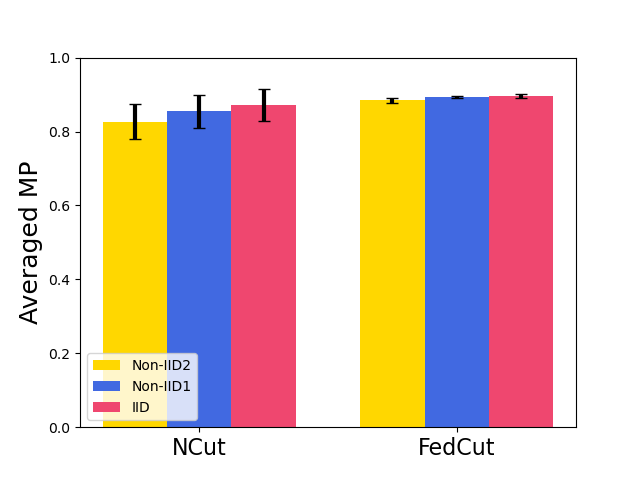}
		\end{minipage}
	}
    	\subfigure{
    		\begin{minipage}[b]{0.3\textwidth}
  		 	\includegraphics[width=1\textwidth]{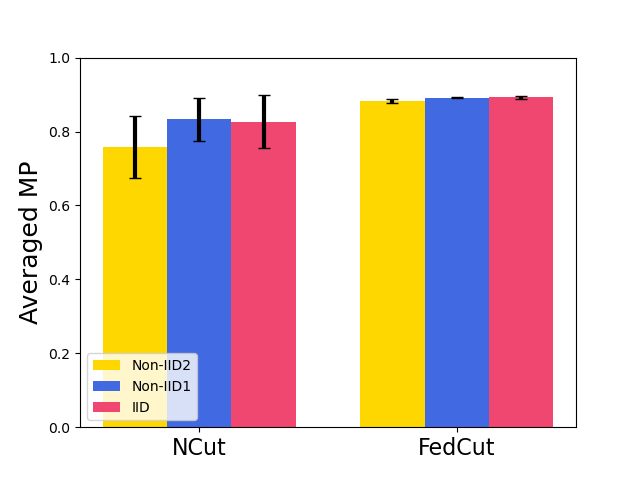}
    		\end{minipage}
    	}
    		\subfigure{
    		\begin{minipage}[b]{0.3\textwidth}
  		 	\includegraphics[width=1\textwidth]{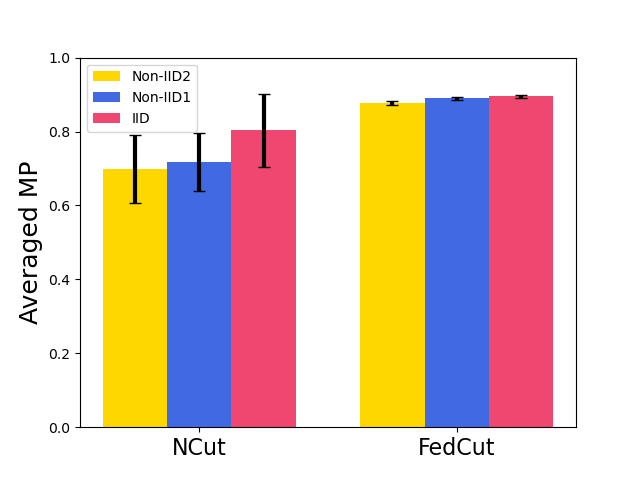}
    		\end{minipage}
    	}
    	
	\caption{Average MP of FedCut and Ncut (FedCut without temporal consistency) among all 8 attacks with different byzantine numbers (left: 10 attackers, middle: 20 attackers and right: 30 attackers) on MNIST (first row) and Fashion MNIST (second row). }
	\label{fig:fedcutvsncut}
	\vspace{-5pt}
\end{figure*}

\begin{figure*}[htbp]
	\centering
	\subfigure{
		\begin{minipage}[b]{0.22\textwidth}
			\includegraphics[width=1\textwidth]{imgs/clusternum-eigengap/LR_mnist_all_ones_noniid_1000000.0_spectral_iteration_1000gamma_0.0008.png}
		\end{minipage}
	}
    	\subfigure{
    		\begin{minipage}[b]{0.22\textwidth}
  		 	\includegraphics[width=1\textwidth]{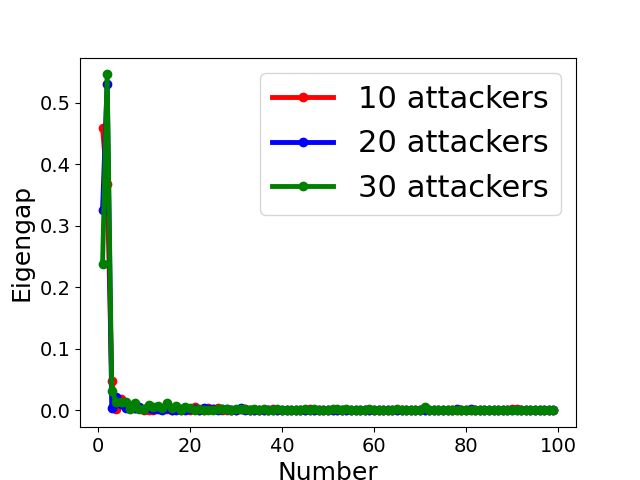}
    		\end{minipage}
    	}
    		\subfigure{
    		\begin{minipage}[b]{0.22\textwidth}
  		 	\includegraphics[width=1\textwidth]{imgs/clusternum-eigengap/LR_mnist_trimmed_attack_noniid_1000000.0_spectral_iteration_1000gamma_0.06.png}
    		\end{minipage}
    	}
    	\subfigure{
		\begin{minipage}[b]{0.22\textwidth}
			\includegraphics[width=1\textwidth]{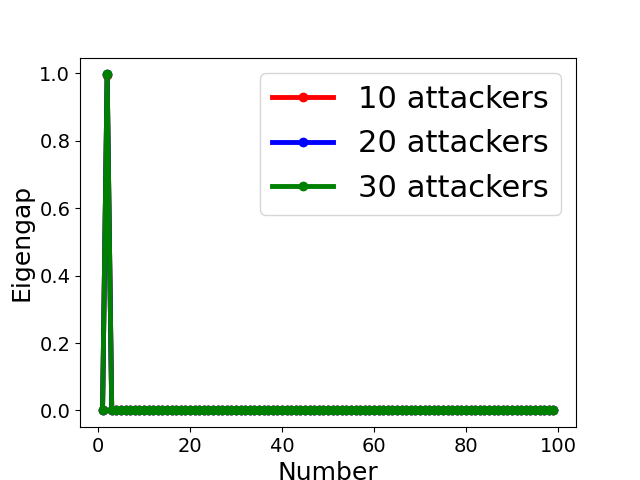}
		\end{minipage}
	}
	
		\subfigure{
		\begin{minipage}[b]{0.22\textwidth}
			\includegraphics[width=1\textwidth]{imgs/clusternum-eigengap/LR_mnist_gaussian_noniid_1000000.0_spectral_iteration_1000gamma_0.0008.png}
		\end{minipage}
	}
    	\subfigure{
    		\begin{minipage}[b]{0.22\textwidth}
  		 	\includegraphics[width=1\textwidth]{imgs/clusternum-eigengap/LR_mnist_label_flipping_noniid_1000000.0_spectral_iteration_1000gamma_0.09.png}
    		\end{minipage}
    	}
    		\subfigure{
    		\begin{minipage}[b]{0.22\textwidth}
  		 	\includegraphics[width=1\textwidth]{imgs/clusternum-eigengap/LR_mnist_sign_flipping_noniid_1000000.0_spectral_iteration_1000gamma_0.2.png}
    		\end{minipage}
    	}
    	\subfigure{
		\begin{minipage}[b]{0.22\textwidth}
			\includegraphics[width=1\textwidth]{imgs/clusternum-eigengap/LR_mnist_collusion_noniid_1000000.0_spectral_iteration_1000gamma_0.01.png}
		\end{minipage}
	}
	\caption{The change of eigengap for different attacks with different number of attackers under IID setting for MNIST dataset.}
	\vspace{-10pt}
	\label{fig:eigengap-clusternum1-app}
\end{figure*}

\begin{figure*}[htbp]
	\centering
	\subfigure{
		\begin{minipage}[b]{0.22\textwidth}
			\includegraphics[width=1\textwidth]{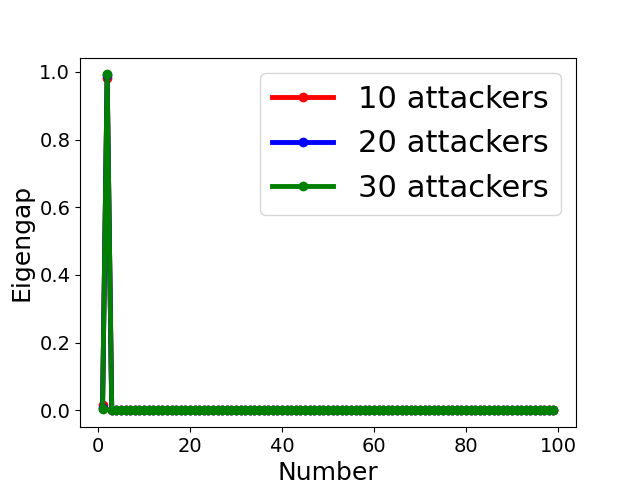}
		\end{minipage}
	}
    	\subfigure{
    		\begin{minipage}[b]{0.22\textwidth}
  		 	\includegraphics[width=1\textwidth]{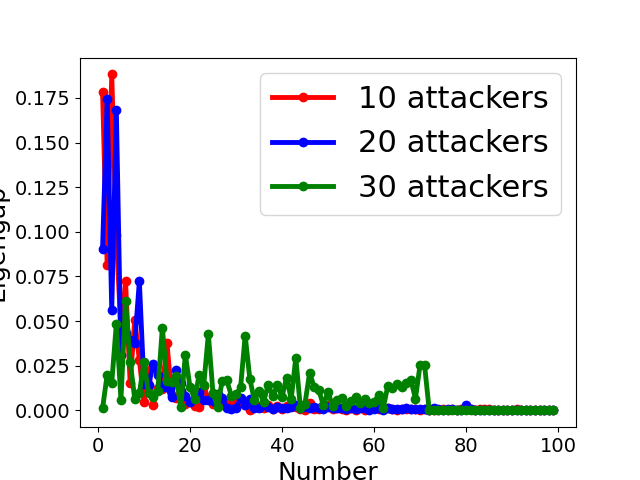}
    		\end{minipage}
    	}
    		\subfigure{
    		\begin{minipage}[b]{0.22\textwidth}
  		 	\includegraphics[width=1\textwidth]{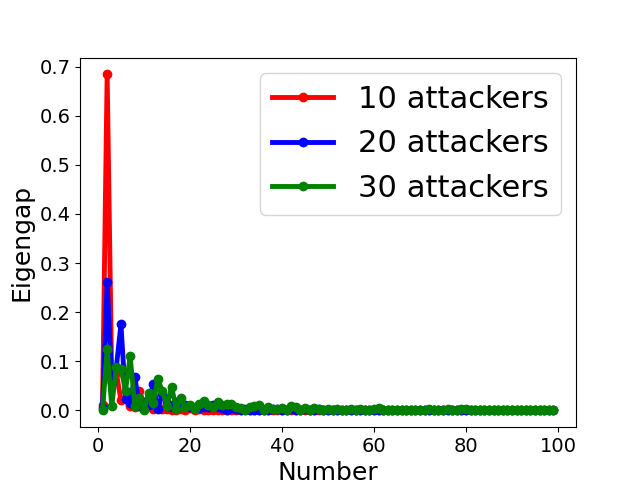}
    		\end{minipage}
    	}
    	\subfigure{
		\begin{minipage}[b]{0.22\textwidth}
			\includegraphics[width=1\textwidth]{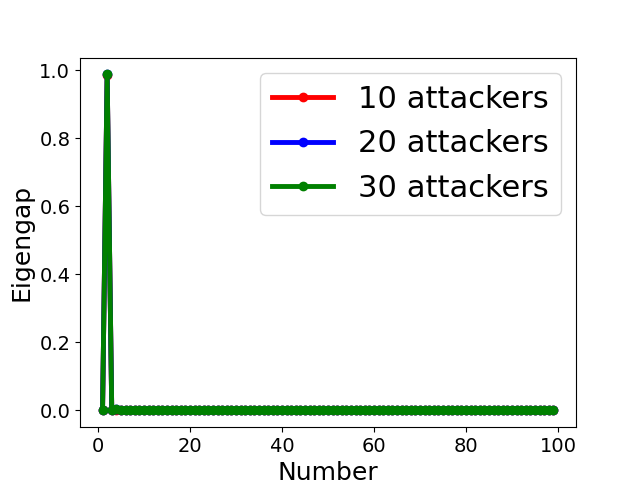}
		\end{minipage}
	}
	
		\subfigure{
		\begin{minipage}[b]{0.22\textwidth}
			\includegraphics[width=1\textwidth]{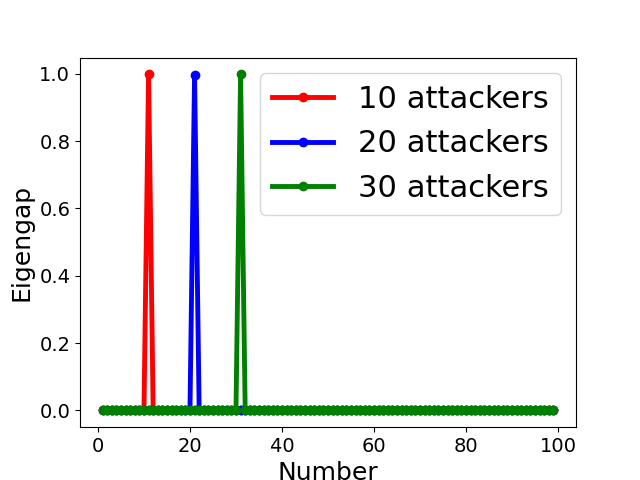}
		\end{minipage}
	}
    	\subfigure{
    		\begin{minipage}[b]{0.22\textwidth}
  		 	\includegraphics[width=1\textwidth]{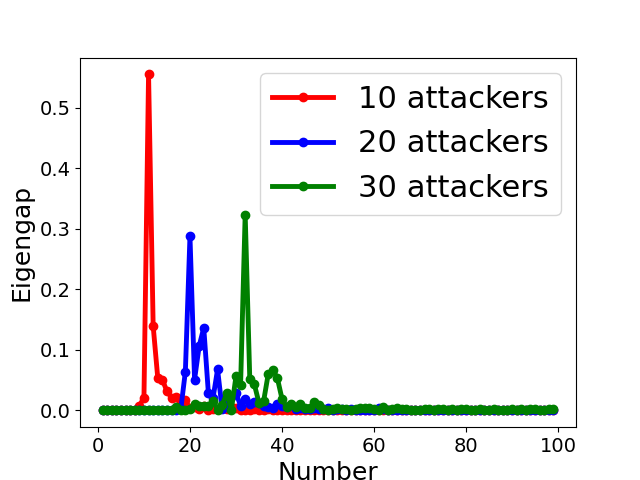}
    		\end{minipage}
    	}
    		\subfigure{
    		\begin{minipage}[b]{0.22\textwidth}
  		 	\includegraphics[width=1\textwidth]{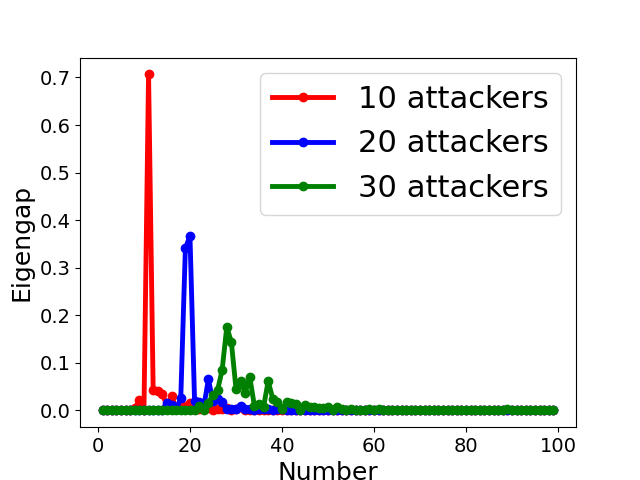}
    		\end{minipage}
    	}
    	\subfigure{
		\begin{minipage}[b]{0.22\textwidth}
			\includegraphics[width=1\textwidth]{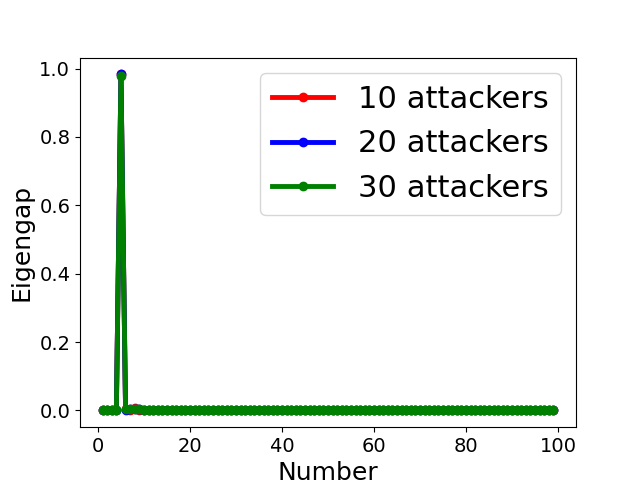}
		\end{minipage}
	}
	\caption{The change of eigengap for different attacks with different number of attackers under Non-IID with $\beta =0.1$ setting for MNIST dataset.}
	\vspace{-10pt}
	\label{fig:eigengap-clusternum2-app}
\end{figure*}

\begin{figure*}[htbp]
\vspace{-6pt}
	\centering
	\subfigure{
		\begin{minipage}[b]{0.23\textwidth}
			\includegraphics[width=1\textwidth]{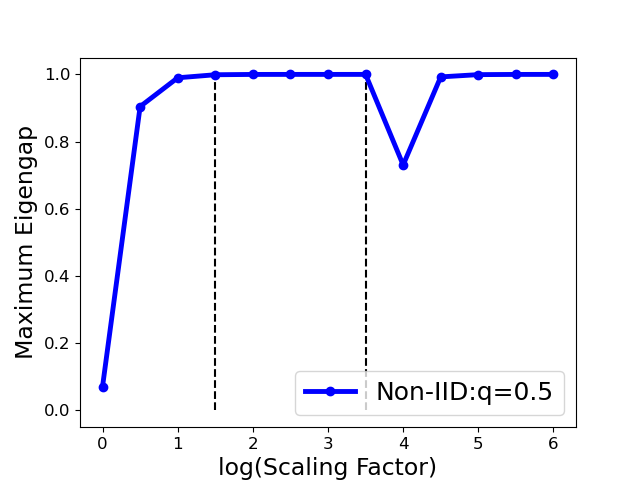}
		\end{minipage}
	}
    	\subfigure{
    		\begin{minipage}[b]{0.23\textwidth}
  		 	\includegraphics[width=1\textwidth]{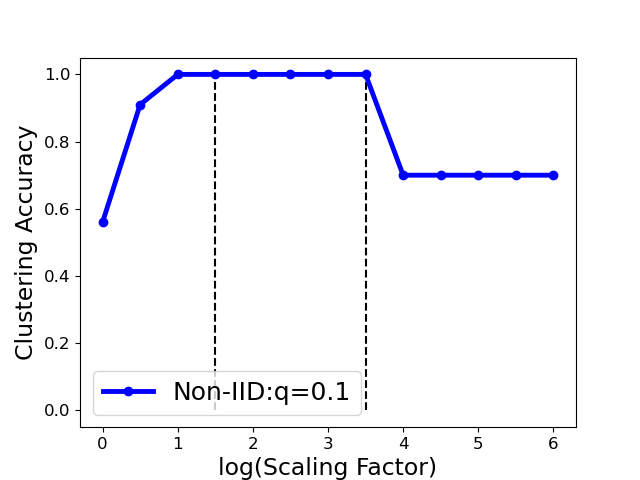}
    		\end{minipage}
    	}
    \subfigure{
		\begin{minipage}[b]{0.23\textwidth}
			\includegraphics[width=1\textwidth]{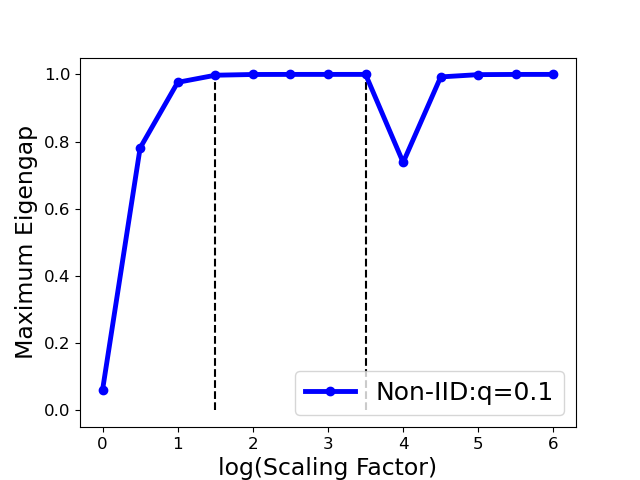}
		\end{minipage}
	}
    	\subfigure{
    		\begin{minipage}[b]{0.23\textwidth}
  		 	\includegraphics[width=1\textwidth]{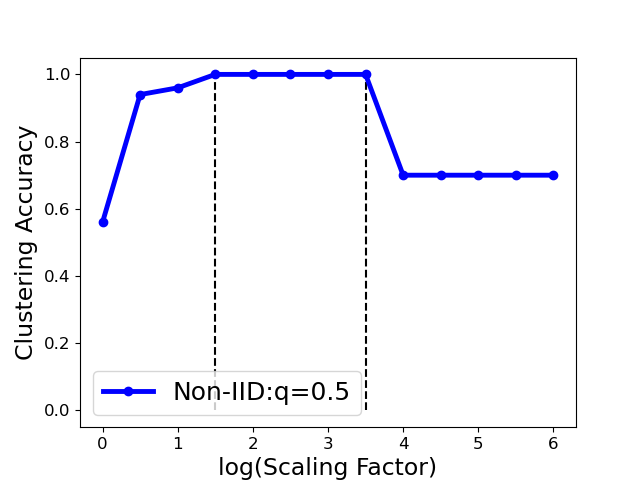}
    		\end{minipage}
    	}
	
	\caption{The change of the maximum eigengap of normalized adjacency matrix and clustering accuracy with different scaling factor $\sigma$ for the Gaussian attack \cite{li2019rsa} on different Non-IID dataset ($q=0.1,0.5)$, where the normalized adjacency matrix $L = D^{-1/2}AD^{-1/2}$, $D = \text{diag}(\text{Sum}(A))$ and $A_{ij}= \text{exp}(-||\bg_i - \bg_j||^2/2\sigma^2)$. The clustering accuracy is calculated by NCut on the normalized adjacency matrix with 100 client including 70 benign clients and 30 Byzantine clients.}
	\vspace{-10pt}
	\label{fig:eigengap-gamma-app}
\end{figure*}

\begin{table*}[bp]

\centering 
\scriptsize
\setlength{\tabcolsep}{1.0mm}
\renewcommand\arraystretch{1.5}
\begin{tabular}{|c||c|c|c|c|c|c|c|c|c|c|c|}
\hline

&&\begin{tabular}[c]{@{}c@{}} 
Krum \\ \cite{blanchard2017machine} \end{tabular} &\begin{tabular}[c]{@{}c@{}}GeoMeidan\\ \cite{chen2017distributed}\end{tabular}&\begin{tabular}[c]{@{}c@{}}Median\\ \cite{yin2018byzantine}\end{tabular}&\begin{tabular}[c]{@{}c@{}}Trimmed\\ \cite{yin2018byzantine}\end{tabular}&\begin{tabular}[c]{@{}c@{}} Bulyan \\ \cite{guerraoui2018hidden}\end{tabular}&\begin{tabular}[c]{@{}c@{}} FLtrust \\ \cite{cao2020fltrust}\end{tabular}&\begin{tabular}[c]{@{}c@{}} DnC\\ \cite{shejwalkar2021manipulating}\end{tabular}&\begin{tabular}[c]{@{}c@{}}Kmeans\\ \cite{shen2016auror}\end{tabular}&\begin{tabular}[c]{@{}c@{}}FedCut\\ (Ours)\end{tabular} \\ \hline \hline

\multirow{9}{*}{\begin{tabular}[c]{@{}c@{}} 
F\\M\\N
\\ I\\S\\T \end{tabular} }

 &No attack  &75.7$\pm${1.1}  &88.5$\pm${0.9}  &76.5$\pm${4.0}  &81.6$\pm${1.4}  &75.5$\pm${1.5}  &85.6$\pm${1.6}  &89.1$\pm${0.4}  &88.7$\pm${0.7}  &88.9$\pm${1.0}

 \\ \cline{2-11} &Lie \cite{baruch2019little}  &73.8$\pm${0.7}  &78.9$\pm${1.7}  &40.0$\pm${28.6}  &47.4$\pm${22.8}  &72.1$\pm${5.0}  &83.0$\pm${2.9}  &75.3$\pm${3.8}  &81.1$\pm${0.9}  &88.7$\pm${0.8}

 \\ \cline{2-11} &Fang-v1 \cite{fang2020local}  &75.4$\pm${2.3}  &86.7$\pm${1.1}  &64.7$\pm${7.7}  &62.3$\pm${17.1}  &74.4$\pm${3.3}  &83.9$\pm${3.0}  &89.9$\pm${0.2}  &76.6$\pm${8.1}  &89.2$\pm${1.4}

 \\ \cline{2-11} &Fang-v2 \cite{fang2020local}  &10.0$\pm${0.0}  &43.4$\pm${5.7}  &68.8$\pm${3.6}  &59.5$\pm${13.0}  &56.8$\pm${11.3}  &84.9$\pm${0.7}  &88.4$\pm${1.9}  &88.8$\pm${0.4}  &88.9$\pm${0.1}

 \\ \cline{2-11} &Same value \cite{li2019rsa}  &75.8$\pm${1.2}  &75.0$\pm${1.6}  &78.3$\pm${4.9}  &78.5$\pm${3.5}  &76.7$\pm${2.9}  &85.9$\pm${2.0}  &89.5$\pm${0.2}  &89.2$\pm${0.5}  &89.3$\pm${0.6}

 \\ \cline{2-11} &Gaussian \cite{blanchard2017machine}  &77.2$\pm${0.8}  &88.2$\pm${0.7}  &78.2$\pm${30.8}  &74.8$\pm${7.9}  &64.8$\pm${15.5}  &86.3$\pm${0.9}  &80.4$\pm${4.2}  &40.6$\pm${24.2}  &89.6$\pm${0.4}

 \\ \cline{2-11} &sign flipping \cite{data2021byzantine} &75.9$\pm${2.4}  &87.4$\pm${1.1}  &78.9$\pm${2.9}  &72.4$\pm${14.3}  &67.7$\pm${15.6}  &85.3$\pm${1.3}  &85.5$\pm${1.3}  &54.9$\pm${10.0}  &84.4$\pm${0.9}

 \\ \cline{2-11} &label flipping \cite{data2021byzantine} &76.6$\pm${0.9}  &88.1$\pm${0.9}  &69.0$\pm${7.5}  &78.7$\pm${1.6}  &60.2$\pm${20.9}  &85.8$\pm${1.3}  &89.4$\pm${0.5}  &87.8$\pm${1.0}  &88.0$\pm${0.4}

 \\ \cline{2-11} &minic  \cite{karimireddy2020byzantine} &75.5$\pm${1.4}  &87.9$\pm${0.6}  &77.1$\pm${2.1}  &86.0$\pm${0.3}  &79.3$\pm${0.3}  &87.4$\pm${0.4}  &88.9$\pm${0.3}  &88.9$\pm${0.2}  &88.9$\pm${0.7}

 \\ \cline{2-11} &Collusion (Ours)  &75.9$\pm${1.6}  &34.9$\pm${11.2}  &72.7$\pm${1.1}  &69.3$\pm${2.2}  &71.8$\pm${4.4}  &87.1$\pm${0.1}  &66.6$\pm${6.6}  &56.6$\pm${22.3}  &89.1$\pm${0.4} 

 \\ \cline{2-11} &Averaged 
 &69.2$\pm${1.2}  &75.9$\pm${2.6}  &70.4$\pm${9.3}  &71.1$\pm${8.4}  &69.9$\pm${8.1}  &85.5$\pm${1.4}  &84.3$\pm${1.9}  &75.3$\pm${6.8}  &88.5$\pm${0.7} 

 \\ \cline{2-11} &Worst-case 
 &10.0$\pm${0.0}  &34.9$\pm${11.2}  &40.0$\pm${28.6}  &47.4$\pm${22.8}  &56.8$\pm${11.3}  &83.0$\pm${2.9}  &66.6$\pm${6.6}  &40.6$\pm${24.2}  &84.4$\pm${0.9}

\\ \hline

\end{tabular}
\caption{Model performances of different Byzantine-resilient methods under different attacks (with Non-IID setting $\beta = 0.1$ and 10 attackers for classification of Fashion MNIST).}
\label{tab:MP1}

\end{table*}

\begin{table*}[bp]

\centering 
\scriptsize
\setlength{\tabcolsep}{1.0mm}
\renewcommand\arraystretch{1.5}
\begin{tabular}{|c||c|c|c|c|c|c|c|c|c|c|c|}
\hline

&&\begin{tabular}[c]{@{}c@{}} 
Krum \\ \cite{blanchard2017machine} \end{tabular} &\begin{tabular}[c]{@{}c@{}}GeoMeidan\\ \cite{chen2017distributed}\end{tabular}&\begin{tabular}[c]{@{}c@{}}Median\\ \cite{yin2018byzantine}\end{tabular}&\begin{tabular}[c]{@{}c@{}}Trimmed\\ \cite{yin2018byzantine}\end{tabular}&\begin{tabular}[c]{@{}c@{}} Bulyan \\ \cite{guerraoui2018hidden}\end{tabular}&\begin{tabular}[c]{@{}c@{}} FLtrust \\ \cite{cao2020fltrust}\end{tabular}&\begin{tabular}[c]{@{}c@{}} DnC\\ \cite{shejwalkar2021manipulating}\end{tabular}&\begin{tabular}[c]{@{}c@{}}Kmeans\\ \cite{shen2016auror}\end{tabular}&\begin{tabular}[c]{@{}c@{}}FedCut\\ (Ours)\end{tabular} \\ \hline \hline

\multirow{9}{*}{\begin{tabular}[c]{@{}c@{}} 
F\\M\\N
\\ I\\S\\T \end{tabular} }

 &No attack  &83.2$\pm${0.6}  &89.9$\pm${0.3}  &84.5$\pm${0.9}  &84.3$\pm${1.9}  &83.6$\pm${2.3}  &87.4$\pm${0.9}  &90.0$\pm${0.4}  &89.6$\pm${0.5}  &90.0$\pm${0.4}

 \\ \cline{2-11} &Lie \cite{baruch2019little}  &83.5$\pm${0.7}  &84.3$\pm${0.6}  &55.1$\pm${18.7}  &75.8$\pm${6.0}  &82.6$\pm${1.8}  &83.6$\pm${2.5}  &71.0$\pm${29.0}  &80.2$\pm${1.3}  &89.9$\pm${0.6}

 \\ \cline{2-11} &Fang-v1 \cite{fang2020local}  &83.7$\pm${0.4}  &89.1$\pm${0.4}  &79.8$\pm${3.4}  &79.7$\pm${3.5}  &84.1$\pm${1.8}  &85.3$\pm${1.6}  &90.2$\pm${0.3}  &89.6$\pm${1.5}  &89.6$\pm${0.3}

 \\ \cline{2-11} &Fang-v2 \cite{fang2020local}  &10.0$\pm${0.0}  &60.3$\pm${6.3}  &72.2$\pm${2.3}  &73.1$\pm${1.5}  &73.0$\pm${2.0}  &85.7$\pm${1.0}  &89.0$\pm${1.3}  &88.7$\pm${0.5}  &89.2$\pm${0.4}

 \\ \cline{2-11} &Same value \cite{li2019rsa}  &83.5$\pm${0.3}  &76.3$\pm${0.5}  &84.2$\pm${0.6}  &83.4$\pm${1.4}  &83.8$\pm${1.3}  &87.1$\pm${0.5}  &89.8$\pm${0.2}  &89.9$\pm${0.2}  &90.0$\pm${0.5}

 \\ \cline{2-11} &Gaussian \cite{blanchard2017machine}  &83.1$\pm${0.4}  &89.7$\pm${0.3}  &84.5$\pm${1.3}  &85.3$\pm${1.2}  &79.9$\pm${4.6}  &87.2$\pm${0.6}  &78.1$\pm${4.6}  &49.6$\pm${16.0}  &89.8$\pm${0.1}

 \\ \cline{2-11} &sign flipping \cite{data2021byzantine} &83.7$\pm${0.7}  &88.9$\pm${0.7}  &82.8$\pm${3.8}  &84.6$\pm${1.8}  &80.8$\pm${2.7}  &86.7$\pm${1.5}  &88.2$\pm${0.5}  &71.7$\pm${10.6}  &86.4$\pm${0.6}

 \\ \cline{2-11} &label flipping \cite{data2021byzantine} &83.6$\pm${0.8}  &89.5$\pm${0.4}  &81.1$\pm${3.4}  &83.4$\pm${3.4}  &81.1$\pm${2.5}  &87.5$\pm${0.6}  &90.0$\pm${0.2}  &89.6$\pm${0.2}  &89.8$\pm${0.3}

 \\ \cline{2-11} &minic  \cite{karimireddy2020byzantine} &83.6$\pm${0.5}  &89.4$\pm${0.1}  &81.5$\pm${6.3}  &88.4$\pm${0.6}  &85.1$\pm${0.8}  &88.7$\pm${0.6}  &89.4$\pm${0.1}  &89.7$\pm${0.2}  &89.2$\pm${0.3}

 \\ \cline{2-11} &Collusion (Ours)  &84.3$\pm${0.6}  &27.8$\pm${24.8}  &75.2$\pm${1.9}  &75.5$\pm${1.4}  &81.9$\pm${2.8}  &88.7$\pm${0.6}  &84.2$\pm${0.7}  &54.5$\pm${12.6}  &90.1$\pm${0.1} 

 \\ \cline{2-11} &Averaged 
 &76.2$\pm${0.5}  &78.5$\pm${3.4}  &78.1$\pm${4.3}  &81.4$\pm${2.3}  &81.6$\pm${2.3}  &86.8$\pm${1.0}  &86.0$\pm${3.7}  &79.3$\pm${4.4}  &89.4$\pm${0.4} 

 \\ \cline{2-11} &Worst-case 
 &10.0$\pm${0.0}  &27.8$\pm${24.8}  &55.1$\pm${18.7}  &73.1$\pm${1.5}  &73.0$\pm${2.0}  &83.6$\pm${2.5}  &71.0$\pm${29.0}  &49.6$\pm${16.0}  &86.4$\pm${0.6}

\\ \hline

\end{tabular}
\caption{Model performances of different Byzantine-resilient methods under different attacks (with Non-IID setting $\beta = 0.5$ and 10 attackers for classification of Fashion MNIST).}

\end{table*}

\begin{table*}[htbp]

\centering 
\scriptsize
\setlength{\tabcolsep}{1.0mm}
\renewcommand\arraystretch{1.5}
\begin{tabular}{|c||c|c|c|c|c|c|c|c|c|c|c|}
\hline

&&\begin{tabular}[c]{@{}c@{}} 
Krum \\ \cite{blanchard2017machine} \end{tabular} &\begin{tabular}[c]{@{}c@{}}GeoMeidan\\ \cite{chen2017distributed}\end{tabular}&\begin{tabular}[c]{@{}c@{}}Median\\ \cite{yin2018byzantine}\end{tabular}&\begin{tabular}[c]{@{}c@{}}Trimmed\\ \cite{yin2018byzantine}\end{tabular}&\begin{tabular}[c]{@{}c@{}} Bulyan \\ \cite{guerraoui2018hidden}\end{tabular}&\begin{tabular}[c]{@{}c@{}} FLtrust \\ \cite{cao2020fltrust}\end{tabular}&\begin{tabular}[c]{@{}c@{}} DnC\\ \cite{shejwalkar2021manipulating}\end{tabular}&\begin{tabular}[c]{@{}c@{}}Kmeans\\ \cite{shen2016auror}\end{tabular}&\begin{tabular}[c]{@{}c@{}}FedCut\\ (Ours)\end{tabular} \\ \hline \hline

\multirow{9}{*}{\begin{tabular}[c]{@{}c@{}} 
F\\M\\N
\\ I\\S\\T \end{tabular} }

 &No attack  &84.7$\pm${0.7}  &90.0$\pm${0.4}  &87.9$\pm${0.4}  &88.2$\pm${0.4}  &87.0$\pm${0.6}  &88.0$\pm${0.6}  &90.2$\pm${0.2}  &89.8$\pm${0.4}  &90.2$\pm${0.1}

 \\ \cline{2-11} &Lie \cite{baruch2019little}  &85.1$\pm${0.3}  &66.8$\pm${22.2}  &81.9$\pm${2.2}  &82.7$\pm${0.9}  &85.1$\pm${0.5}  &85.2$\pm${1.1}  &80.0$\pm${3.8}  &78.6$\pm${3.2}  &90.0$\pm${0.4}

 \\ \cline{2-11} &Fang-v1 \cite{fang2020local}  &84.9$\pm${0.6}  &88.7$\pm${0.7}  &86.5$\pm${0.4}  &86.6$\pm${0.6}  &86.9$\pm${0.4}  &87.4$\pm${1.0}  &89.9$\pm${0.4}  &90.1$\pm${0.2}  &87.8$\pm${2.1}

 \\ \cline{2-11} &Fang-v2 \cite{fang2020local}  &47.7$\pm${31.0}  &78.2$\pm${1.8}  &78.0$\pm${0.3}  &84.7$\pm${0.8}  &84.8$\pm${1.0}  &84.3$\pm${0.8}  &89.6$\pm${0.9}  &89.3$\pm${0.2}  &90.3$\pm${0.4}

 \\ \cline{2-11} &Same value \cite{li2019rsa}  &85.3$\pm${0.5}  &76.9$\pm${0.4}  &85.5$\pm${0.5}  &85.1$\pm${0.7}  &87.2$\pm${0.5}  &88.2$\pm${0.9}  &90.2$\pm${0.6}  &90.3$\pm${0.2}  &90.3$\pm${0.1}

 \\ \cline{2-11} &Gaussian \cite{blanchard2017machine}  &85.1$\pm${0.9}  &90.0$\pm${0.4}  &87.8$\pm${0.5}  &88.1$\pm${0.4}  &87.5$\pm${0.4}  &88.3$\pm${0.8}  &76.3$\pm${5.2}  &30.2$\pm${26.5}  &90.3$\pm${0.2}

 \\ \cline{2-11} &sign flipping \cite{data2021byzantine} &85.0$\pm${0.4}  &89.4$\pm${0.3}  &87.9$\pm${0.4}  &88.2$\pm${0.4}  &87.1$\pm${0.4}  &86.5$\pm${1.9}  &89.3$\pm${0.6}  &87.7$\pm${0.6}  &89.4$\pm${1.0}

 \\ \cline{2-11} &label flipping \cite{data2021byzantine} &84.9$\pm${0.5}  &89.9$\pm${0.2}  &85.6$\pm${1.1}  &85.1$\pm${0.7}  &87.1$\pm${0.4}  &88.2$\pm${0.8}  &90.4$\pm${0.4}  &90.1$\pm${0.2}  &88.0$\pm${0.2}

 \\ \cline{2-11} &minic  \cite{karimireddy2020byzantine} &85.3$\pm${0.2}  &89.4$\pm${0.4}  &87.9$\pm${0.7}  &89.4$\pm${0.5}  &86.9$\pm${0.3}  &89.1$\pm${0.3}  &89.8$\pm${0.4}  &89.9$\pm${0.2}  &89.8$\pm${0.2}

 \\ \cline{2-11} &Collusion (Ours)  &85.3$\pm${0.5}  &54.7$\pm${4.5}  &79.1$\pm${1.3}  &78.1$\pm${0.4}  &87.5$\pm${0.4}  &88.7$\pm${0.2}  &85.0$\pm${0.5}  &67.6$\pm${2.9}  &90.1$\pm${0.0} 

 \\ \cline{2-11} &Averaged 
 &81.3$\pm${3.6}  &81.4$\pm${3.1}  &84.8$\pm${0.8}  &85.6$\pm${0.6}  &86.7$\pm${0.5}  &87.4$\pm${0.8}  &87.1$\pm${1.3}  &80.4$\pm${3.5}  &89.6$\pm${0.5} 

 \\ \cline{2-11} &Worst-case 
 &47.7$\pm${31.0}  &54.7$\pm${4.5}  &78.0$\pm${0.3}  &78.1$\pm${0.4}  &84.8$\pm${1.0}  &84.3$\pm${0.8}  &76.3$\pm${5.2}  &30.2$\pm${26.5}  &87.8$\pm${2.1}

\\ \hline

\end{tabular}
\caption{Model performances of different Byzantine-resilient methods under different attacks (with IID and 10 attackers for classification of Fashion MNIST).}

\end{table*}

\begin{table*}[htbp]

\centering 
\scriptsize
\setlength{\tabcolsep}{1.0mm}
\renewcommand\arraystretch{1.5}
\begin{tabular}{|c||c|c|c|c|c|c|c|c|c|c|c|}
\hline

&&\begin{tabular}[c]{@{}c@{}} 
Krum \\ \cite{blanchard2017machine} \end{tabular} &\begin{tabular}[c]{@{}c@{}}GeoMeidan\\ \cite{chen2017distributed}\end{tabular}&\begin{tabular}[c]{@{}c@{}}Median\\ \cite{yin2018byzantine}\end{tabular}&\begin{tabular}[c]{@{}c@{}}Trimmed\\ \cite{yin2018byzantine}\end{tabular}&\begin{tabular}[c]{@{}c@{}} Bulyan \\ \cite{guerraoui2018hidden}\end{tabular}&\begin{tabular}[c]{@{}c@{}} FLtrust \\ \cite{cao2020fltrust}\end{tabular}&\begin{tabular}[c]{@{}c@{}} DnC\\ \cite{shejwalkar2021manipulating}\end{tabular}&\begin{tabular}[c]{@{}c@{}}Kmeans\\ \cite{shen2016auror}\end{tabular}&\begin{tabular}[c]{@{}c@{}}FedCut\\ (Ours)\end{tabular} \\ \hline \hline

\multirow{9}{*}{\begin{tabular}[c]{@{}c@{}} 
F\\M\\N
\\ I\\S\\T \end{tabular} }

 &No attack  &72.0$\pm${2.5}  &88.4$\pm${0.7}  &70.1$\pm${11.0}  &81.8$\pm${0.3}  &81.8$\pm${0.4}  &86.1$\pm${1.1}  &88.9$\pm${0.8}  &88.7$\pm${0.6}  &89.4$\pm${0.4}

 \\ \cline{2-11} &Lie \cite{baruch2019little}  &75.4$\pm${1.2}  &64.4$\pm${9.9}  &35.5$\pm${25.0}  &41.3$\pm${19.5}  &26.9$\pm${28.5}  &76.8$\pm${4.0}  &72.1$\pm${6.8}  &88.1$\pm${0.6}  &88.1$\pm${2.3}

 \\ \cline{2-11} &Fang-v1 \cite{fang2020local}  &76.4$\pm${1.5}  &78.5$\pm${0.8}  &27.9$\pm${10.9}  &35.8$\pm${9.2}  &71.6$\pm${5.5}  &84.1$\pm${2.9}  &88.7$\pm${0.8}  &89.7$\pm${0.4}  &89.2$\pm${0.1}

 \\ \cline{2-11} &Fang-v2 \cite{fang2020local}  &10.0$\pm${0.0}  &10.9$\pm${1.2}  &18.8$\pm${18.4}  &32.9$\pm${5.8}  &23.6$\pm${9.0}  &84.9$\pm${0.4}  &82.2$\pm${8.1}  &65.8$\pm${10.3}  &89.0$\pm${0.3}

 \\ \cline{2-11} &Same value \cite{li2019rsa}  &76.6$\pm${1.9}  &64.5$\pm${8.9}  &57.8$\pm${11.9}  &66.7$\pm${3.7}  &81.2$\pm${0.7}  &86.1$\pm${1.0}  &89.5$\pm${0.4}  &89.5$\pm${0.4}  &89.4$\pm${0.2}

 \\ \cline{2-11} &Gaussian \cite{blanchard2017machine}  &76.4$\pm${0.4}  &88.0$\pm${0.7}  &76.4$\pm${2.7}  &79.4$\pm${1.4}  &78.2$\pm${3.5}  &85.7$\pm${1.9}  &76.4$\pm${6.6}  &20.7$\pm${28.9}  &89.3$\pm${0.1}

 \\ \cline{2-11} &sign flipping \cite{data2021byzantine} &75.5$\pm${1.2}  &85.3$\pm${5.5}  &73.0$\pm${5.9}  &54.5$\pm${25.3}  &80.5$\pm${8.5}  &83.4$\pm${2.5}  &76.7$\pm${1.2}  &61.1$\pm${3.4}  &82.9$\pm${0.5}

 \\ \cline{2-11} &label flipping \cite{data2021byzantine} &76.5$\pm${0.7}  &88.3$\pm${1.0}  &66.8$\pm${13.1}  &75.5$\pm${2.9}  &77.3$\pm${10.4}  &85.6$\pm${1.2}  &87.3$\pm${2.0}  &87.0$\pm${1.2}  &88.5$\pm${0.5}

 \\ \cline{2-11} &minic  \cite{karimireddy2020byzantine} &75.3$\pm${1.4}  &84.6$\pm${1.4}  &60.4$\pm${13.3}  &81.1$\pm${3.6}  &66.9$\pm${4.7}  &87.4$\pm${0.5}  &88.0$\pm${0.1}  &88.7$\pm${0.4}  &89.1$\pm${0.2}

 \\ \cline{2-11} &Collusion (Ours)  &75.0$\pm${1.6}  &55.9$\pm${11.4}  &67.7$\pm${2.0}  &70.7$\pm${0.4}  &71.7$\pm${2.2}  &86.4$\pm${0.5}  &86.1$\pm${0.4}  &32.4$\pm${20.0}  &89.4$\pm${0.4} 

 \\ \cline{2-11} &Averaged 
 &68.9$\pm${1.2}  &70.9$\pm${4.2}  &55.4$\pm${11.4}  &62.0$\pm${7.2}  &66.0$\pm${7.3}  &84.7$\pm${1.6}  &83.6$\pm${2.7}  &71.2$\pm${6.6}  &88.4$\pm${0.5} 

 \\ \cline{2-11} &Worst-case 
 &10.0$\pm${0.0}  &10.9$\pm${1.2}  &18.8$\pm${18.4}  &32.9$\pm${5.8}  &23.6$\pm${9.0}  &76.8$\pm${4.0}  &72.1$\pm${6.8}  &20.7$\pm${28.9}  &82.9$\pm${0.5}

\\ \hline

\end{tabular}
\caption{Model performances of different Byzantine-resilient methods under different attacks (with Non-IID setting $\beta = 0.1$ and 20 attackers for classification of Fashion MNIST).}

\end{table*}

\begin{table*}[htbp]

\centering 
\scriptsize
\setlength{\tabcolsep}{1.0mm}
\renewcommand\arraystretch{1.5}
\begin{tabular}{|c||c|c|c|c|c|c|c|c|c|c|c|}
\hline

&&\begin{tabular}[c]{@{}c@{}} 
Krum \\ \cite{blanchard2017machine} \end{tabular} &\begin{tabular}[c]{@{}c@{}}GeoMeidan\\ \cite{chen2017distributed}\end{tabular}&\begin{tabular}[c]{@{}c@{}}Median\\ \cite{yin2018byzantine}\end{tabular}&\begin{tabular}[c]{@{}c@{}}Trimmed\\ \cite{yin2018byzantine}\end{tabular}&\begin{tabular}[c]{@{}c@{}} Bulyan \\ \cite{guerraoui2018hidden}\end{tabular}&\begin{tabular}[c]{@{}c@{}} FLtrust \\ \cite{cao2020fltrust}\end{tabular}&\begin{tabular}[c]{@{}c@{}} DnC\\ \cite{shejwalkar2021manipulating}\end{tabular}&\begin{tabular}[c]{@{}c@{}}Kmeans\\ \cite{shen2016auror}\end{tabular}&\begin{tabular}[c]{@{}c@{}}FedCut\\ (Ours)\end{tabular} \\ \hline \hline

\multirow{9}{*}{\begin{tabular}[c]{@{}c@{}} 
F\\M\\N
\\ I\\S\\T \end{tabular} }

 &No attack  &82.3$\pm${0.7}  &90.0$\pm${0.3}  &85.6$\pm${1.0}  &86.5$\pm${0.4}  &85.5$\pm${2.8}  &88.1$\pm${0.6}  &89.4$\pm${1.0}  &88.4$\pm${1.0}  &90.2$\pm${0.2}

 \\ \cline{2-11} &Lie \cite{baruch2019little}  &83.5$\pm${1.0}  &53.7$\pm${13.5}  &36.1$\pm${15.6}  &51.9$\pm${26.7}  &54.5$\pm${12.9}  &76.3$\pm${6.0}  &73.3$\pm${7.8}  &87.9$\pm${1.0}  &89.9$\pm${0.1}

 \\ \cline{2-11} &Fang-v1 \cite{fang2020local}  &84.8$\pm${0.7}  &85.5$\pm${1.1}  &75.5$\pm${2.5}  &71.8$\pm${7.3}  &83.1$\pm${0.9}  &84.4$\pm${1.4}  &89.2$\pm${1.3}  &88.8$\pm${1.2}  &89.4$\pm${0.2}

 \\ \cline{2-11} &Fang-v2 \cite{fang2020local}  &10.0$\pm${8.3}  &18.9$\pm${7.1}  &62.4$\pm${5.1}  &55.8$\pm${2.8}  &59.3$\pm${5.4}  &84.9$\pm${0.8}  &83.0$\pm${8.3}  &72.4$\pm${5.2}  &89.5$\pm${0.2}

 \\ \cline{2-11} &Same value \cite{li2019rsa}  &83.7$\pm${0.8}  &64.8$\pm${8.2}  &77.7$\pm${2.7}  &71.5$\pm${5.5}  &86.2$\pm${1.0}  &86.9$\pm${0.7}  &88.8$\pm${1.2}  &89.0$\pm${1.2}  &89.4$\pm${0.1}

 \\ \cline{2-11} &Gaussian \cite{blanchard2017machine}  &82.9$\pm${0.4}  &89.6$\pm${0.5}  &84.6$\pm${2.5}  &86.2$\pm${1.1}  &85.7$\pm${0.7}  &88.4$\pm${0.4}  &77.5$\pm${7.2}  &38.2$\pm${34.5}  &89.5$\pm${0.3}

 \\ \cline{2-11} &sign flipping \cite{data2021byzantine} &82.9$\pm${1.4}  &87.7$\pm${0.9}  &85.0$\pm${1.0}  &85.0$\pm${0.9}  &85.1$\pm${1.4}  &84.4$\pm${1.8}  &85.0$\pm${1.5}  &74.3$\pm${7.4}  &86.0$\pm${0.1}

 \\ \cline{2-11} &label flipping \cite{data2021byzantine} &82.4$\pm${1.1}  &89.9$\pm${0.4}  &79.4$\pm${3.4}  &80.4$\pm${1.8}  &80.3$\pm${6.6}  &88.1$\pm${0.2}  &88.4$\pm${1.2}  &88.6$\pm${1.5}  &89.3$\pm${0.3}

 \\ \cline{2-11} &minic  \cite{karimireddy2020byzantine} &81.1$\pm${2.6}  &87.5$\pm${1.1}  &83.0$\pm${0.4}  &87.8$\pm${0.3}  &76.3$\pm${1.5}  &88.7$\pm${0.3}  &89.3$\pm${0.3}  &89.4$\pm${0.3}  &88.9$\pm${0.6}

 \\ \cline{2-11} &Collusion (Ours)  &84.5$\pm${1.0}  &41.7$\pm${17.0}  &72.0$\pm${1.2}  &61.4$\pm${15.3}  &81.4$\pm${5.0}  &88.7$\pm${0.7}  &86.9$\pm${0.5}  &56.5$\pm${9.7}  &89.9$\pm${0.3} 

 \\ \cline{2-11} &Averaged 
 &75.8$\pm${1.8}  &70.9$\pm${5.0}  &74.1$\pm${3.5}  &73.8$\pm${6.2}  &77.7$\pm${3.8}  &85.9$\pm${1.3}  &85.1$\pm${3.0}  &77.4$\pm${6.3}  &89.2$\pm${0.2} 

 \\ \cline{2-11} &Worst-case 
 &10.0$\pm${8.3}  &18.9$\pm${7.1}  &36.1$\pm${15.6}  &51.9$\pm${26.7}  &54.5$\pm${12.9}  &76.3$\pm${6.0}  &73.3$\pm${7.8}  &38.2$\pm${34.5}  &86.0$\pm${0.1}

\\ \hline

\end{tabular}
\caption{Model performances of different Byzantine-resilient methods under different attacks (with Non-IID setting $\beta = 0.5$ and 20 attackers for classification of Fashion MNIST).}

\end{table*}

\begin{table*}[htbp]

\centering 
\scriptsize
\setlength{\tabcolsep}{1.0mm}
\renewcommand\arraystretch{1.5}
\begin{tabular}{|c||c|c|c|c|c|c|c|c|c|c|c|}
\hline

&&\begin{tabular}[c]{@{}c@{}} 
Krum \\ \cite{blanchard2017machine} \end{tabular} &\begin{tabular}[c]{@{}c@{}}GeoMeidan\\ \cite{chen2017distributed}\end{tabular}&\begin{tabular}[c]{@{}c@{}}Median\\ \cite{yin2018byzantine}\end{tabular}&\begin{tabular}[c]{@{}c@{}}Trimmed\\ \cite{yin2018byzantine}\end{tabular}&\begin{tabular}[c]{@{}c@{}} Bulyan \\ \cite{guerraoui2018hidden}\end{tabular}&\begin{tabular}[c]{@{}c@{}} FLtrust \\ \cite{cao2020fltrust}\end{tabular}&\begin{tabular}[c]{@{}c@{}} DnC\\ \cite{shejwalkar2021manipulating}\end{tabular}&\begin{tabular}[c]{@{}c@{}}Kmeans\\ \cite{shen2016auror}\end{tabular}&\begin{tabular}[c]{@{}c@{}}FedCut\\ (Ours)\end{tabular} \\ \hline \hline

\multirow{9}{*}{\begin{tabular}[c]{@{}c@{}} 
F\\M\\N
\\ I\\S\\T \end{tabular} }

 &No attack  &84.6$\pm${0.6}  &90.1$\pm${0.4}  &88.3$\pm${0.2}  &88.0$\pm${0.6}  &87.5$\pm${0.2}  &88.3$\pm${0.6}  &89.3$\pm${1.0}  &89.2$\pm${0.7}  &89.1$\pm${0.9}

 \\ \cline{2-11} &Lie \cite{baruch2019little}  &84.7$\pm${0.7}  &76.3$\pm${18.9}  &69.0$\pm${26.6}  &75.2$\pm${3.1}  &82.4$\pm${0.2}  &79.0$\pm${4.2}  &79.7$\pm${4.6}  &89.0$\pm${1.0}  &88.8$\pm${1.0}

 \\ \cline{2-11} &Fang-v1 \cite{fang2020local}  &85.3$\pm${0.4}  &86.4$\pm${0.7}  &80.8$\pm${1.4}  &80.3$\pm${2.0}  &87.5$\pm${0.3}  &86.6$\pm${0.9}  &88.5$\pm${1.0}  &89.2$\pm${0.6}  &89.7$\pm${0.2}

 \\ \cline{2-11} &Fang-v2 \cite{fang2020local}  &38.0$\pm${21.1}  &46.4$\pm${9.7}  &70.9$\pm${1.4}  &71.0$\pm${2.1}  &17.2$\pm${12.5}  &82.9$\pm${0.8}  &85.0$\pm${6.4}  &86.6$\pm${2.3}  &90.0$\pm${0.2}

 \\ \cline{2-11} &Same value \cite{li2019rsa}  &85.0$\pm${0.5}  &74.0$\pm${2.2}  &78.9$\pm${2.0}  &77.8$\pm${1.2}  &87.1$\pm${0.4}  &88.4$\pm${0.4}  &89.3$\pm${0.8}  &89.2$\pm${0.7}  &89.5$\pm${0.2}

 \\ \cline{2-11} &Gaussian \cite{blanchard2017machine}  &85.2$\pm${0.5}  &89.9$\pm${0.4}  &88.1$\pm${0.3}  &88.5$\pm${0.3}  &87.1$\pm${0.4}  &88.1$\pm${0.5}  &78.6$\pm${6.6}  &40.7$\pm${33.8}  &89.4$\pm${0.2}

 \\ \cline{2-11} &sign flipping \cite{data2021byzantine} &84.5$\pm${0.5}  &88.5$\pm${0.1}  &87.4$\pm${0.3}  &87.4$\pm${0.5}  &87.4$\pm${0.2}  &86.4$\pm${0.6}  &87.3$\pm${0.8}  &75.7$\pm${6.8}  &87.3$\pm${0.5}

 \\ \cline{2-11} &label flipping \cite{data2021byzantine} &85.5$\pm${0.6}  &89.9$\pm${0.6}  &82.3$\pm${0.8}  &82.0$\pm${0.5}  &87.3$\pm${0.2}  &87.3$\pm${0.9}  &89.2$\pm${1.1}  &89.8$\pm${0.7}  &89.7$\pm${0.1}

 \\ \cline{2-11} &minic  \cite{karimireddy2020byzantine} &81.8$\pm${0.8}  &86.3$\pm${0.8}  &86.4$\pm${0.3}  &88.5$\pm${0.2}  &80.0$\pm${0.4}  &88.4$\pm${0.3}  &89.1$\pm${0.2}  &89.2$\pm${0.9}  &89.8$\pm${0.4}

 \\ \cline{2-11} &Collusion (Ours)  &84.9$\pm${0.7}  &35.7$\pm${5.0}  &76.7$\pm${1.1}  &69.7$\pm${2.9}  &87.4$\pm${0.2}  &88.6$\pm${0.2}  &86.9$\pm${0.4}  &66.7$\pm${22.7}  &89.6$\pm${1.0} 

 \\ \cline{2-11} &Averaged 
 &79.9$\pm${2.6}  &76.4$\pm${3.9}  &80.9$\pm${3.4}  &80.8$\pm${1.3}  &79.1$\pm${1.5}  &86.4$\pm${0.9}  &86.3$\pm${2.3}  &80.5$\pm${7.0}  &89.3$\pm${0.5} 

 \\ \cline{2-11} &Worst-case 
 &38.0$\pm${21.1}  &35.7$\pm${5.0}  &69.0$\pm${26.6}  &69.7$\pm${2.9}  &17.2$\pm${12.5}  &79.0$\pm${4.2}  &78.6$\pm${6.6}  &40.7$\pm${33.8}  &87.3$\pm${0.5}

\\ \hline

\end{tabular}
\caption{Model performances of different Byzantine-resilient methods under different attacks (with IID setting and 20 attackers for classification of Fashion MNIST).}

\end{table*}

\begin{table*}[htbp]

\centering 
\scriptsize
\setlength{\tabcolsep}{1.0mm}
\renewcommand\arraystretch{1.5}
\begin{tabular}{|c||c|c|c|c|c|c|c|c|c|c|c|}
\hline

&&\begin{tabular}[c]{@{}c@{}} 
Krum \\ \cite{blanchard2017machine} \end{tabular} &\begin{tabular}[c]{@{}c@{}}GeoMeidan\\ \cite{chen2017distributed}\end{tabular}&\begin{tabular}[c]{@{}c@{}}Median\\ \cite{yin2018byzantine}\end{tabular}&\begin{tabular}[c]{@{}c@{}}Trimmed\\ \cite{yin2018byzantine}\end{tabular}&\begin{tabular}[c]{@{}c@{}} Bulyan \\ \cite{guerraoui2018hidden}\end{tabular}&\begin{tabular}[c]{@{}c@{}} FLtrust \\ \cite{cao2020fltrust}\end{tabular}&\begin{tabular}[c]{@{}c@{}} DnC\\ \cite{shejwalkar2021manipulating}\end{tabular}&\begin{tabular}[c]{@{}c@{}}Kmeans\\ \cite{shen2016auror}\end{tabular}&\begin{tabular}[c]{@{}c@{}}FedCut\\ (Ours)\end{tabular} \\ \hline \hline

\multirow{9}{*}{\begin{tabular}[c]{@{}c@{}} 
F\\M\\N
\\ I\\S\\T \end{tabular} }

 &No attack  &71.7$\pm${1.9}  &88.6$\pm${0.7}  &80.6$\pm${1.2}  &79.5$\pm${3.7}  &79.4$\pm${0.9}  &86.2$\pm${1.5}  &88.6$\pm${0.2}  &89.3$\pm${0.2}  &89.3$\pm${0.5}

 \\ \cline{2-11} &Lie \cite{baruch2019little}  &74.4$\pm${0.8}  &37.9$\pm${17.2}  &26.9$\pm${14.5}  &30.1$\pm${22.4}  &48.0$\pm${10.9}  &66.5$\pm${26.1}  &61.4$\pm${22.2}  &87.7$\pm${1.8}  &88.1$\pm${2.3}

 \\ \cline{2-11} &Fang-v1 \cite{fang2020local}  &74.5$\pm${1.9}  &61.3$\pm${4.6}  &19.5$\pm${5.4}  &21.3$\pm${5.3}  &16.7$\pm${33.8}  &82.0$\pm${3.0}  &79.1$\pm${2.4}  &89.1$\pm${0.4}  &88.6$\pm${0.6}

 \\ \cline{2-11} &Fang-v2 \cite{fang2020local}  &10.0$\pm${0.0}  &10.0$\pm${0.0}  &15.3$\pm${7.7}  &17.6$\pm${3.1}  &12.3$\pm${3.0}  &84.4$\pm${0.4}  &78.0$\pm${14.6}  &15.0$\pm${18.6}  &88.7$\pm${0.2}

 \\ \cline{2-11} &Same value \cite{li2019rsa}  &75.6$\pm${1.0}  &43.2$\pm${19.2}  &44.5$\pm${8.2}  &41.0$\pm${3.5}  &79.6$\pm${1.1}  &86.0$\pm${0.4}  &89.2$\pm${0.7}  &89.4$\pm${0.3}  &89.0$\pm${0.5}

 \\ \cline{2-11} &Gaussian \cite{blanchard2017machine}  &76.1$\pm${2.1}  &87.9$\pm${0.8}  &75.2$\pm${3.7}  &80.9$\pm${1.0}  &78.2$\pm${3.7}  &85.9$\pm${1.4}  &76.7$\pm${11.1}  &13.3$\pm${29.5}  &88.8$\pm${0.1}

 \\ \cline{2-11} &sign flipping \cite{data2021byzantine} &74.1$\pm${1.0}  &79.3$\pm${5.6}  &71.1$\pm${6.1}  &70.0$\pm${5.1}  &79.6$\pm${1.0}  &80.7$\pm${2.2}  &50.2$\pm${3.2}  &39.4$\pm${3.8}  &82.1$\pm${0.5}

 \\ \cline{2-11} &label flipping \cite{data2021byzantine} &76.5$\pm${1.5}  &87.3$\pm${1.0}  &67.7$\pm${7.0}  &70.2$\pm${5.1}  &80.1$\pm${14.9}  &85.9$\pm${1.6}  &83.7$\pm${6.8}  &80.6$\pm${3.2}  &86.7$\pm${1.9}

 \\ \cline{2-11} &minic  \cite{karimireddy2020byzantine} &72.6$\pm${5.5}  &70.3$\pm${2.2}  &73.8$\pm${0.7}  &75.9$\pm${7.8}  &53.6$\pm${3.6}  &86.9$\pm${0.7}  &86.9$\pm${0.5}  &88.9$\pm${0.2}  &88.9$\pm${0.1}

 \\ \cline{2-11} &Collusion (Ours)  &76.4$\pm${1.5}  &39.7$\pm${9.0}  &67.8$\pm${2.2}  &67.5$\pm${3.4}  &19.7$\pm${6.3}  &86.1$\pm${0.2}  &82.1$\pm${1.1}  &54.0$\pm${3.1}  &89.2$\pm${0.3} 

 \\ \cline{2-11} &Averaged 
 &68.2$\pm${1.7}  &60.6$\pm${6.0}  &54.2$\pm${5.7}  &55.4$\pm${6.0}  &54.7$\pm${7.9}  &83.1$\pm${3.8}  &77.6$\pm${6.3}  &64.7$\pm${6.1}  &87.9$\pm${0.7} 

 \\ \cline{2-11} &Worst-case 
 &10.0$\pm${0.0}  &10.0$\pm${0.0}  &15.3$\pm${7.7}  &17.6$\pm${3.1}  &12.3$\pm${3.0}  &66.5$\pm${26.1}  &50.2$\pm${3.2}  &13.3$\pm${29.5}  &82.1$\pm${0.5}

\\ \hline

\end{tabular}
\caption{Model performances of different Byzantine-resilient methods under different attacks (with Non-IID setting $\beta = 0.1$ and 30 attackers for classification of Fashion MNIST).}

\end{table*}

\begin{table*}[htbp]

\centering 
\scriptsize
\setlength{\tabcolsep}{1.0mm}
\renewcommand\arraystretch{1.5}
\begin{tabular}{|c||c|c|c|c|c|c|c|c|c|c|c|}
\hline

&&\begin{tabular}[c]{@{}c@{}} 
Krum \\ \cite{blanchard2017machine} \end{tabular} &\begin{tabular}[c]{@{}c@{}}GeoMeidan\\ \cite{chen2017distributed}\end{tabular}&\begin{tabular}[c]{@{}c@{}}Median\\ \cite{yin2018byzantine}\end{tabular}&\begin{tabular}[c]{@{}c@{}}Trimmed\\ \cite{yin2018byzantine}\end{tabular}&\begin{tabular}[c]{@{}c@{}} Bulyan \\ \cite{guerraoui2018hidden}\end{tabular}&\begin{tabular}[c]{@{}c@{}} FLtrust \\ \cite{cao2020fltrust}\end{tabular}&\begin{tabular}[c]{@{}c@{}} DnC\\ \cite{shejwalkar2021manipulating}\end{tabular}&\begin{tabular}[c]{@{}c@{}}Kmeans\\ \cite{shen2016auror}\end{tabular}&\begin{tabular}[c]{@{}c@{}}FedCut\\ (Ours)\end{tabular} \\ \hline \hline

\multirow{9}{*}{\begin{tabular}[c]{@{}c@{}} 
F\\M\\N
\\ I\\S\\T \end{tabular} }

 &No attack  &82.6$\pm${0.8}  &89.6$\pm${0.7}  &85.4$\pm${1.5}  &84.4$\pm${1.9}  &83.7$\pm${4.3}  &87.3$\pm${0.8}  &89.5$\pm${0.4}  &89.1$\pm${0.6}  &89.7$\pm${0.2}

 \\ \cline{2-11} &Lie \cite{baruch2019little}  &83.0$\pm${1.0}  &45.4$\pm${14.6}  &39.0$\pm${15.9}  &50.9$\pm${10.5}  &59.5$\pm${11.1}  &52.6$\pm${19.5}  &79.8$\pm${2.7}  &89.6$\pm${0.7}  &89.3$\pm${0.3}

 \\ \cline{2-11} &Fang-v1 \cite{fang2020local}  &83.3$\pm${0.3}  &76.6$\pm${2.2}  &47.6$\pm${10.2}  &52.1$\pm${5.9}  &81.2$\pm${3.1}  &85.0$\pm${1.6}  &85.5$\pm${0.2}  &89.6$\pm${0.3}  &88.8$\pm${0.8}

 \\ \cline{2-11} &Fang-v2 \cite{fang2020local}  &10.0$\pm${0.0}  &10.0$\pm${0.0}  &47.4$\pm${5.5}  &32.4$\pm${14.8}  &12.3$\pm${3.3}  &84.6$\pm${0.9}  &80.8$\pm${11.8}  &34.6$\pm${24.0}  &88.8$\pm${0.6}

 \\ \cline{2-11} &Same value \cite{li2019rsa}  &84.2$\pm${0.7}  &49.9$\pm${15.9}  &53.5$\pm${8.0}  &48.3$\pm${10.3}  &85.8$\pm${0.9}  &85.4$\pm${1.7}  &89.4$\pm${0.1}  &89.3$\pm${0.4}  &89.1$\pm${0.5}

 \\ \cline{2-11} &Gaussian \cite{blanchard2017machine}  &83.4$\pm${1.0}  &89.0$\pm${0.4}  &84.1$\pm${1.9}  &86.3$\pm${1.3}  &85.1$\pm${1.2}  &87.6$\pm${0.9}  &73.0$\pm${8.9}  &17.5$\pm${29.4}  &89.5$\pm${0.1}

 \\ \cline{2-11} &sign flipping \cite{data2021byzantine} &83.0$\pm${0.4}  &85.3$\pm${1.7}  &83.6$\pm${0.8}  &83.5$\pm${2.0}  &83.1$\pm${1.6}  &79.7$\pm${2.1}  &59.6$\pm${4.3}  &29.7$\pm${23.1}  &87.9$\pm${1.6}

 \\ \cline{2-11} &label flipping \cite{data2021byzantine} &83.6$\pm${0.3}  &89.7$\pm${0.5}  &67.4$\pm${15.1}  &75.5$\pm${1.6}  &85.1$\pm${1.1}  &87.4$\pm${1.1}  &88.9$\pm${1.2}  &89.3$\pm${0.4}  &88.5$\pm${0.4}

 \\ \cline{2-11} &minic  \cite{karimireddy2020byzantine} &77.0$\pm${1.7}  &78.3$\pm${0.7}  &80.6$\pm${1.0}  &84.1$\pm${0.7}  &63.0$\pm${7.4}  &88.3$\pm${0.3}  &87.9$\pm${0.6}  &88.9$\pm${0.6}  &89.4$\pm${0.5}

 \\ \cline{2-11} &Collusion (Ours)  &83.6$\pm${0.4}  &10.4$\pm${8.3}  &65.2$\pm${4.8}  &63.3$\pm${7.1}  &85.3$\pm${0.6}  &88.6$\pm${0.5}  &83.6$\pm${0.3}  &63.0$\pm${24.1}  &89.0$\pm${0.5} 

 \\ \cline{2-11} &Averaged 
 &75.4$\pm${0.7}  &62.4$\pm${4.5}  &65.4$\pm${6.5}  &66.1$\pm${5.6}  &72.4$\pm${3.5}  &82.7$\pm${2.9}  &81.8$\pm${3.1}  &68.1$\pm${10.4}  &89.0$\pm${0.6} 

 \\ \cline{2-11} &Worst-case 
 &10.0$\pm${0.0}  &10.0$\pm${0.0}  &39.0$\pm${15.9}  &32.4$\pm${14.8}  &12.3$\pm${3.3}  &52.6$\pm${19.5}  &59.6$\pm${4.3}  &17.5$\pm${29.4}  &87.9$\pm${1.6}

\\ \hline

\end{tabular}
\caption{Model performances of different Byzantine-resilient methods under different attacks (with Non-IID setting $\beta = 0.5$ and 30 attackers for classification of Fashion MNIST).}

\end{table*}

\begin{table*}[htbp]

\centering 
\scriptsize
\setlength{\tabcolsep}{1.0mm}
\renewcommand\arraystretch{1.5}
\begin{tabular}{|c||c|c|c|c|c|c|c|c|c|c|c|}
\hline

&&\begin{tabular}[c]{@{}c@{}} 
Krum \\ \cite{blanchard2017machine} \end{tabular} &\begin{tabular}[c]{@{}c@{}}GeoMeidan\\ \cite{chen2017distributed}\end{tabular}&\begin{tabular}[c]{@{}c@{}}Median\\ \cite{yin2018byzantine}\end{tabular}&\begin{tabular}[c]{@{}c@{}}Trimmed\\ \cite{yin2018byzantine}\end{tabular}&\begin{tabular}[c]{@{}c@{}} Bulyan \\ \cite{guerraoui2018hidden}\end{tabular}&\begin{tabular}[c]{@{}c@{}} FLtrust \\ \cite{cao2020fltrust}\end{tabular}&\begin{tabular}[c]{@{}c@{}} DnC\\ \cite{shejwalkar2021manipulating}\end{tabular}&\begin{tabular}[c]{@{}c@{}}Kmeans\\ \cite{shen2016auror}\end{tabular}&\begin{tabular}[c]{@{}c@{}}FedCut\\ (Ours)\end{tabular} \\ \hline \hline

\multirow{9}{*}{\begin{tabular}[c]{@{}c@{}} 
M\\N
\\ I\\S\\T \end{tabular} }

 &No attack  &64.2$\pm${6.8}  &92.1$\pm${0.1}  &66.9$\pm${2.3}  &66.1$\pm${2.1}  &55.6$\pm${5.9}  &89.8$\pm${0.2}  &91.7$\pm${0.2}  &91.8$\pm${0.3}  &92.2$\pm${0.0}

 \\ \cline{2-11} &Lie \cite{baruch2019little}  &61.6$\pm${7.7}  &90.5$\pm${0.2}  &75.8$\pm${2.8}  &73.8$\pm${3.7}  &63.1$\pm${2.2}  &89.6$\pm${0.1}  &92.1$\pm${0.2}  &89.8$\pm${0.2}  &92.2$\pm${0.1}

 \\ \cline{2-11} &Fang-v1 \cite{fang2020local}  &62.7$\pm${5.9}  &89.8$\pm${0.2}  &50.9$\pm${4.5}  &42.9$\pm${6.6}  &57.5$\pm${2.0}  &79.3$\pm${1.5}  &92.0$\pm${0.3}  &68.4$\pm${5.3}  &92.0$\pm${0.1}

 \\ \cline{2-11} &Fang-v2 \cite{fang2020local}  &58.4$\pm${4.2}  &88.2$\pm${0.1}  &40.9$\pm${3.4}  &45.0$\pm${2.7}  &41.6$\pm${5.2}  &78.5$\pm${2.4}  &91.4$\pm${0.4}  &91.6$\pm${0.1}  &92.3$\pm${0.2}

 \\ \cline{2-11} &Same value \cite{li2019rsa}  &63.0$\pm${6.6}  &90.5$\pm${0.2}  &58.8$\pm${2.2}  &58.7$\pm${2.2}  &55.2$\pm${3.2}  &89.8$\pm${0.2}  &92.1$\pm${0.2}  &92.1$\pm${0.3}  &92.2$\pm${0.1}

 \\ \cline{2-11} &Gaussian \cite{blanchard2017machine}  &63.6$\pm${4.5}  &91.9$\pm${0.1}  &68.5$\pm${3.1}  &67.1$\pm${4.9}  &54.5$\pm${3.8}  &89.6$\pm${0.4}  &83.3$\pm${0.3}  &43.8$\pm${3.4}  &92.3$\pm${0.1}

 \\ \cline{2-11} &sign flipping \cite{data2021byzantine} &66.2$\pm${7.9}  &91.5$\pm${0.1}  &66.8$\pm${3.4}  &68.5$\pm${5.1}  &53.6$\pm${3.0}  &79.6$\pm${0.6}  &90.3$\pm${0.2}  &64.7$\pm${1.3}  &91.5$\pm${0.2}

 \\ \cline{2-11} &label flipping \cite{data2021byzantine} &64.8$\pm${2.1}  &91.3$\pm${0.1}  &57.3$\pm${2.5}  &58.1$\pm${5.4}  &53.0$\pm${7.2}  &89.9$\pm${0.4}  &90.5$\pm${0.2}  &86.4$\pm${0.7}  &88.8$\pm${0.1}

 \\ \cline{2-11} &minic  \cite{karimireddy2020byzantine} &72.5$\pm${3.1}  &91.3$\pm${0.1}  &74.8$\pm${0.9}  &88.0$\pm${0.4}  &65.1$\pm${4.2}  &90.2$\pm${0.1}  &91.4$\pm${0.1}  &91.1$\pm${0.2}  &92.1$\pm${0.2}

 \\ \cline{2-11} &Collusion (Ours)  &66.6$\pm${4.9}  &89.7$\pm${0.6}  &69.7$\pm${7.3}  &70.8$\pm${6.4}  &58.5$\pm${9.0}  &89.6$\pm${0.4}  &89.2$\pm${1.9}  &79.1$\pm${0.3}  &92.1$\pm${0.2} 

 \\ \cline{2-11} &Averaged 
 &64.4$\pm${5.4}  &90.7$\pm${0.2}  &63.0$\pm${3.2}  &63.9$\pm${4.0}  &55.8$\pm${4.6}  &86.6$\pm${0.6}  &90.4$\pm${0.4}  &79.9$\pm${1.2}  &91.8$\pm${0.1} 

 \\ \cline{2-11} &Worst-case 
 &58.4$\pm${4.2}  &88.2$\pm${0.1}  &40.9$\pm${3.4}  &42.9$\pm${6.6}  &41.6$\pm${5.2}  &78.5$\pm${2.4}  &83.3$\pm${0.3}  &43.8$\pm${3.4}  &88.8$\pm${0.1}

\\ \hline

\end{tabular}
\caption{Model performances of different Byzantine-resilient methods under different attacks (with Non-IID setting $\beta = 0.1$ and 10 attackers for classification of MNIST).}

\end{table*}

\begin{table*}[htbp]

\centering 
\scriptsize
\setlength{\tabcolsep}{1.0mm}
\renewcommand\arraystretch{1.5}
\begin{tabular}{|c||c|c|c|c|c|c|c|c|c|c|c|}
\hline

&&\begin{tabular}[c]{@{}c@{}} 
Krum \\ \cite{blanchard2017machine} \end{tabular} &\begin{tabular}[c]{@{}c@{}}GeoMeidan\\ \cite{chen2017distributed}\end{tabular}&\begin{tabular}[c]{@{}c@{}}Median\\ \cite{yin2018byzantine}\end{tabular}&\begin{tabular}[c]{@{}c@{}}Trimmed\\ \cite{yin2018byzantine}\end{tabular}&\begin{tabular}[c]{@{}c@{}} Bulyan \\ \cite{guerraoui2018hidden}\end{tabular}&\begin{tabular}[c]{@{}c@{}} FLtrust \\ \cite{cao2020fltrust}\end{tabular}&\begin{tabular}[c]{@{}c@{}} DnC\\ \cite{shejwalkar2021manipulating}\end{tabular}&\begin{tabular}[c]{@{}c@{}}Kmeans\\ \cite{shen2016auror}\end{tabular}&\begin{tabular}[c]{@{}c@{}}FedCut\\ (Ours)\end{tabular} \\ \hline \hline

\multirow{9}{*}{\begin{tabular}[c]{@{}c@{}} 
M\\N
\\ I\\S\\T \end{tabular} }

 &No attack  &89.0$\pm${0.3}  &92.4$\pm${0.1}  &85.3$\pm${0.5}  &84.9$\pm${0.1}  &81.7$\pm${1.0}  &89.0$\pm${0.4}  &92.2$\pm${0.3}  &92.0$\pm${0.4}  &92.4$\pm${0.1}

 \\ \cline{2-11} &Lie \cite{baruch2019little}  &89.3$\pm${0.3}  &91.5$\pm${0.2}  &87.7$\pm${0.2}  &88.1$\pm${0.5}  &85.1$\pm${0.7}  &88.8$\pm${0.7}  &92.3$\pm${0.2}  &91.6$\pm${0.0}  &92.4$\pm${0.1}

 \\ \cline{2-11} &Fang-v1 \cite{fang2020local}  &89.2$\pm${0.3}  &91.4$\pm${0.2}  &72.5$\pm${1.6}  &72.3$\pm${1.3}  &80.1$\pm${0.7}  &78.2$\pm${2.0}  &92.3$\pm${0.2}  &92.2$\pm${0.2}  &92.3$\pm${0.0}

 \\ \cline{2-11} &Fang-v2 \cite{fang2020local}  &38.1$\pm${5.5}  &88.2$\pm${0.2}  &70.2$\pm${1.6}  &70.2$\pm${2.3}  &60.5$\pm${1.5}  &80.9$\pm${0.6}  &91.9$\pm${0.3}  &91.9$\pm${0.2}  &92.4$\pm${0.1}

 \\ \cline{2-11} &Same value \cite{li2019rsa}  &89.3$\pm${0.4}  &91.2$\pm${0.2}  &79.3$\pm${0.3}  &78.5$\pm${1.1}  &81.1$\pm${0.6}  &88.5$\pm${0.9}  &92.3$\pm${0.3}  &92.3$\pm${0.2}  &92.4$\pm${0.1}

 \\ \cline{2-11} &Gaussian \cite{blanchard2017machine}  &89.2$\pm${0.4}  &92.3$\pm${0.2}  &85.7$\pm${0.2}  &85.5$\pm${0.3}  &80.6$\pm${0.6}  &88.5$\pm${0.5}  &83.5$\pm${0.7}  &43.3$\pm${5.3}  &92.3$\pm${0.1}

 \\ \cline{2-11} &sign flipping \cite{data2021byzantine} &89.1$\pm${0.2}  &92.2$\pm${0.1}  &86.0$\pm${0.1}  &85.6$\pm${0.5}  &81.0$\pm${0.8}  &78.1$\pm${0.8}  &92.0$\pm${0.1}  &81.9$\pm${8.3}  &91.9$\pm${0.3}

 \\ \cline{2-11} &label flipping \cite{data2021byzantine} &89.3$\pm${0.5}  &91.8$\pm${0.2}  &82.2$\pm${0.2}  &82.6$\pm${1.0}  &81.3$\pm${1.6}  &88.9$\pm${0.3}  &91.8$\pm${0.1}  &90.4$\pm${0.5}  &91.3$\pm${1.3}

 \\ \cline{2-11} &minic  \cite{karimireddy2020byzantine} &89.9$\pm${0.3}  &92.1$\pm${0.2}  &88.6$\pm${0.4}  &91.4$\pm${0.1}  &85.4$\pm${0.3}  &90.6$\pm${0.1}  &92.2$\pm${0.1}  &92.1$\pm${0.1}  &92.4$\pm${0.1}

 \\ \cline{2-11} &Collusion (Ours)  &89.2$\pm${0.4}  &90.7$\pm${0.7}  &84.5$\pm${2.4}  &84.1$\pm${2.6}  &82.3$\pm${3.7}  &88.7$\pm${0.2}  &90.1$\pm${1.7}  &78.1$\pm${5.3}  &92.2$\pm${0.0} 

 \\ \cline{2-11} &Averaged 
 &84.2$\pm${0.9}  &91.4$\pm${0.2}  &82.2$\pm${0.8}  &82.3$\pm${1.0}  &79.9$\pm${1.2}  &86.0$\pm${0.7}  &91.1$\pm${0.4}  &84.6$\pm${2.1}  &92.2$\pm${0.2} 

 \\ \cline{2-11} &Worst-case 
 &38.1$\pm${5.5}  &88.2$\pm${0.2}  &70.2$\pm${1.6}  &70.2$\pm${2.3}  &60.5$\pm${1.5}  &78.1$\pm${0.8}  &83.5$\pm${0.7}  &43.3$\pm${5.3}  &91.3$\pm${1.3}

\\ \hline

\end{tabular}
\caption{Model performances of different Byzantine-resilient methods under different attacks (with Non-IID setting $\beta = 0.5$ and 10 attackers for classification of MNIST).}

\end{table*}

\begin{table*}[htbp]

\centering 
\scriptsize
\setlength{\tabcolsep}{1.0mm}
\renewcommand\arraystretch{1.5}
\begin{tabular}{|c||c|c|c|c|c|c|c|c|c|c|c|}
\hline

&&\begin{tabular}[c]{@{}c@{}} 
Krum \\ \cite{blanchard2017machine} \end{tabular} &\begin{tabular}[c]{@{}c@{}}GeoMeidan\\ \cite{chen2017distributed}\end{tabular}&\begin{tabular}[c]{@{}c@{}}Median\\ \cite{yin2018byzantine}\end{tabular}&\begin{tabular}[c]{@{}c@{}}Trimmed\\ \cite{yin2018byzantine}\end{tabular}&\begin{tabular}[c]{@{}c@{}} Bulyan \\ \cite{guerraoui2018hidden}\end{tabular}&\begin{tabular}[c]{@{}c@{}} FLtrust \\ \cite{cao2020fltrust}\end{tabular}&\begin{tabular}[c]{@{}c@{}} DnC\\ \cite{shejwalkar2021manipulating}\end{tabular}&\begin{tabular}[c]{@{}c@{}}Kmeans\\ \cite{shen2016auror}\end{tabular}&\begin{tabular}[c]{@{}c@{}}FedCut\\ (Ours)\end{tabular} \\ \hline \hline

\multirow{9}{*}{\begin{tabular}[c]{@{}c@{}} 
M\\N
\\ I\\S\\T \end{tabular} }

 &No attack  &90.3$\pm${0.1}  &92.5$\pm${0.1}  &90.3$\pm${0.1}  &90.1$\pm${0.2}  &88.3$\pm${0.3}  &89.6$\pm${0.2}  &92.3$\pm${0.3}  &92.3$\pm${0.2}  &92.4$\pm${0.1}

 \\ \cline{2-11} &Lie \cite{baruch2019little}  &90.5$\pm${0.1}  &91.8$\pm${0.2}  &90.7$\pm${0.2}  &90.9$\pm${0.1}  &89.1$\pm${0.1}  &89.6$\pm${0.1}  &92.3$\pm${0.2}  &91.8$\pm${0.3}  &92.3$\pm${0.1}

 \\ \cline{2-11} &Fang-v1 \cite{fang2020local}  &90.4$\pm${0.1}  &91.7$\pm${0.2}  &83.6$\pm${0.2}  &82.8$\pm${0.4}  &88.0$\pm${0.2}  &77.1$\pm${2.1}  &92.4$\pm${0.2}  &91.9$\pm${0.3}  &92.3$\pm${0.1}

 \\ \cline{2-11} &Fang-v2 \cite{fang2020local}  &43.8$\pm${2.0}  &89.5$\pm${0.6}  &83.6$\pm${0.5}  &83.2$\pm${0.3}  &71.2$\pm${0.6}  &81.6$\pm${1.4}  &92.2$\pm${0.3}  &92.1$\pm${0.1}  &92.4$\pm${0.0}

 \\ \cline{2-11} &Same value \cite{li2019rsa}  &90.4$\pm${0.2}  &91.6$\pm${0.1}  &87.9$\pm${0.2}  &87.6$\pm${0.1}  &88.6$\pm${0.2}  &89.5$\pm${0.1}  &92.2$\pm${0.1}  &92.3$\pm${0.3}  &92.4$\pm${0.1}

 \\ \cline{2-11} &Gaussian \cite{blanchard2017machine}  &90.5$\pm${0.2}  &92.4$\pm${0.1}  &90.4$\pm${0.3}  &90.4$\pm${0.1}  &88.5$\pm${0.2}  &88.7$\pm${0.2}  &83.6$\pm${0.5}  &50.7$\pm${0.9}  &92.4$\pm${0.1}

 \\ \cline{2-11} &sign flipping \cite{data2021byzantine} &90.4$\pm${0.1}  &92.3$\pm${0.1}  &90.4$\pm${0.2}  &90.3$\pm${0.1}  &88.6$\pm${0.1}  &75.6$\pm${3.8}  &92.2$\pm${0.2}  &91.5$\pm${0.1}  &92.2$\pm${0.1}

 \\ \cline{2-11} &label flipping \cite{data2021byzantine} &90.4$\pm${0.2}  &91.8$\pm${0.1}  &89.4$\pm${0.3}  &89.2$\pm${0.2}  &88.6$\pm${0.2}  &89.2$\pm${0.5}  &92.2$\pm${0.1}  &92.2$\pm${0.3}  &92.2$\pm${0.1}

 \\ \cline{2-11} &minic  \cite{karimireddy2020byzantine} &90.7$\pm${0.2}  &92.1$\pm${0.2}  &90.9$\pm${0.2}  &91.7$\pm${0.1}  &89.6$\pm${0.1}  &90.4$\pm${0.2}  &92.5$\pm${0.1}  &92.2$\pm${0.1}  &92.4$\pm${0.1}

 \\ \cline{2-11} &Collusion (Ours)  &90.5$\pm${0.2}  &91.0$\pm${0.4}  &88.9$\pm${0.8}  &88.6$\pm${0.8}  &88.7$\pm${0.6}  &89.0$\pm${0.4}  &91.7$\pm${0.4}  &76.5$\pm${0.7}  &92.2$\pm${0.2} 

 \\ \cline{2-11} &Averaged 
 &85.8$\pm${0.3}  &91.7$\pm${0.2}  &88.6$\pm${0.3}  &88.5$\pm${0.2}  &86.9$\pm${0.3}  &86.0$\pm${0.9}  &91.4$\pm${0.2}  &86.4$\pm${0.3}  &92.3$\pm${0.1} 

 \\ \cline{2-11} &Worst-case 
 &43.8$\pm${2.0}  &89.5$\pm${0.6}  &83.6$\pm${0.2}  &82.8$\pm${0.4}  &71.2$\pm${0.6}  &75.6$\pm${3.8}  &83.6$\pm${0.5}  &50.7$\pm${0.9}  &92.2$\pm${0.1}

\\ \hline

\end{tabular}
\caption{Model performances of different Byzantine-resilient methods under different attacks (with IID and 10 attackers for classification of MNIST).}

\end{table*}

\begin{table*}[htbp]

\centering 
\scriptsize
\setlength{\tabcolsep}{1.0mm}
\renewcommand\arraystretch{1.5}
\begin{tabular}{|c||c|c|c|c|c|c|c|c|c|c|c|}
\hline

&&\begin{tabular}[c]{@{}c@{}} 
Krum \\ \cite{blanchard2017machine} \end{tabular} &\begin{tabular}[c]{@{}c@{}}GeoMeidan\\ \cite{chen2017distributed}\end{tabular}&\begin{tabular}[c]{@{}c@{}}Median\\ \cite{yin2018byzantine}\end{tabular}&\begin{tabular}[c]{@{}c@{}}Trimmed\\ \cite{yin2018byzantine}\end{tabular}&\begin{tabular}[c]{@{}c@{}} Bulyan \\ \cite{guerraoui2018hidden}\end{tabular}&\begin{tabular}[c]{@{}c@{}} FLtrust \\ \cite{cao2020fltrust}\end{tabular}&\begin{tabular}[c]{@{}c@{}} DnC\\ \cite{shejwalkar2021manipulating}\end{tabular}&\begin{tabular}[c]{@{}c@{}}Kmeans\\ \cite{shen2016auror}\end{tabular}&\begin{tabular}[c]{@{}c@{}}FedCut\\ (Ours)\end{tabular} \\ \hline \hline

\multirow{9}{*}{\begin{tabular}[c]{@{}c@{}} 
M\\N
\\ I\\S\\T \end{tabular} }

 &No attack  &63.2$\pm${8.0}  &92.1$\pm${0.1}  &63.3$\pm${3.4}  &66.0$\pm${0.7}  &75.6$\pm${2.6}  &89.9$\pm${0.2}  &91.7$\pm${0.3}  &91.8$\pm${0.2}  &92.3$\pm${0.1}

 \\ \cline{2-11} &Lie \cite{baruch2019little}  &66.6$\pm${2.7}  &86.4$\pm${0.2}  &81.6$\pm${0.2}  &83.5$\pm${0.4}  &79.2$\pm${0.5}  &88.7$\pm${0.5}  &92.0$\pm${0.3}  &92.0$\pm${0.2}  &92.1$\pm${0.1}

 \\ \cline{2-11} &Fang-v1 \cite{fang2020local}  &68.4$\pm${2.5}  &79.0$\pm${0.8}  &32.3$\pm${4.2}  &30.7$\pm${5.5}  &63.2$\pm${6.2}  &78.8$\pm${2.2}  &92.0$\pm${0.3}  &92.0$\pm${0.4}  &92.2$\pm${0.1}

 \\ \cline{2-11} &Fang-v2 \cite{fang2020local}  &52.1$\pm${3.9}  &62.5$\pm${4.0}  &26.4$\pm${1.9}  &28.2$\pm${3.1}  &39.5$\pm${9.1}  &81.0$\pm${3.3}  &89.0$\pm${1.3}  &91.5$\pm${0.1}  &92.2$\pm${0.1}

 \\ \cline{2-11} &Same value \cite{li2019rsa}  &66.3$\pm${6.7}  &86.8$\pm${0.5}  &52.7$\pm${5.1}  &47.9$\pm${2.2}  &71.1$\pm${4.4}  &89.8$\pm${0.4}  &91.9$\pm${0.3}  &92.0$\pm${0.2}  &92.1$\pm${0.1}

 \\ \cline{2-11} &Gaussian \cite{blanchard2017machine}  &68.6$\pm${1.9}  &91.7$\pm${0.1}  &68.0$\pm${3.6}  &70.0$\pm${3.1}  &71.4$\pm${3.5}  &89.4$\pm${0.4}  &78.9$\pm${0.2}  &27.1$\pm${1.8}  &92.1$\pm${0.2}

 \\ \cline{2-11} &sign flipping \cite{data2021byzantine} &70.4$\pm${4.6}  &90.5$\pm${0.1}  &66.5$\pm${5.2}  &66.9$\pm${4.0}  &71.9$\pm${3.5}  &78.6$\pm${0.8}  &82.7$\pm${1.7}  &66.7$\pm${9.8}  &89.6$\pm${0.7}

 \\ \cline{2-11} &label flipping \cite{data2021byzantine} &65.3$\pm${1.5}  &90.0$\pm${0.1}  &51.4$\pm${2.3}  &47.9$\pm${7.2}  &72.2$\pm${2.9}  &89.7$\pm${0.5}  &88.2$\pm${0.7}  &87.1$\pm${0.2}  &85.4$\pm${0.4}

 \\ \cline{2-11} &minic  \cite{karimireddy2020byzantine} &68.0$\pm${2.6}  &87.4$\pm${1.8}  &81.4$\pm${1.0}  &82.3$\pm${3.1}  &65.3$\pm${7.9}  &90.3$\pm${0.2}  &90.5$\pm${0.2}  &91.2$\pm${0.1}  &92.0$\pm${0.3}

 \\ \cline{2-11} &Collusion (Ours)  &65.9$\pm${4.8}  &84.6$\pm${0.3}  &69.3$\pm${5.2}  &69.2$\pm${7.3}  &76.5$\pm${6.3}  &89.3$\pm${0.6}  &91.3$\pm${0.2}  &41.2$\pm${2.0}  &92.1$\pm${0.1} 

 \\ \cline{2-11} &Averaged 
 &65.5$\pm${3.9}  &85.1$\pm${0.8}  &59.3$\pm${3.2}  &59.3$\pm${3.7}  &68.6$\pm${4.7}  &86.6$\pm${0.9}  &88.8$\pm${0.6}  &77.3$\pm${1.5}  &91.2$\pm${0.2} 

 \\ \cline{2-11} &Worst-case 
 &52.1$\pm${3.9}  &62.5$\pm${4.0}  &26.4$\pm${1.9}  &28.2$\pm${3.1}  &39.5$\pm${9.1}  &78.6$\pm${0.8}  &78.9$\pm${0.2}  &27.1$\pm${1.8}  &85.4$\pm${0.4}

\\ \hline

\end{tabular}
\caption{Model performances of different Byzantine-resilient methods under different attacks (with Non-IID setting $\beta = 0.1$ and 20 attackers for classification of MNIST).}

\end{table*}

\begin{table*}[htbp]

\centering 
\scriptsize
\setlength{\tabcolsep}{1.0mm}
\renewcommand\arraystretch{1.5}
\begin{tabular}{|c||c|c|c|c|c|c|c|c|c|c|c|}
\hline

&&\begin{tabular}[c]{@{}c@{}} 
Krum \\ \cite{blanchard2017machine} \end{tabular} &\begin{tabular}[c]{@{}c@{}}GeoMeidan\\ \cite{chen2017distributed}\end{tabular}&\begin{tabular}[c]{@{}c@{}}Median\\ \cite{yin2018byzantine}\end{tabular}&\begin{tabular}[c]{@{}c@{}}Trimmed\\ \cite{yin2018byzantine}\end{tabular}&\begin{tabular}[c]{@{}c@{}} Bulyan \\ \cite{guerraoui2018hidden}\end{tabular}&\begin{tabular}[c]{@{}c@{}} FLtrust \\ \cite{cao2020fltrust}\end{tabular}&\begin{tabular}[c]{@{}c@{}} DnC\\ \cite{shejwalkar2021manipulating}\end{tabular}&\begin{tabular}[c]{@{}c@{}}Kmeans\\ \cite{shen2016auror}\end{tabular}&\begin{tabular}[c]{@{}c@{}}FedCut\\ (Ours)\end{tabular} \\ \hline \hline

\multirow{9}{*}{\begin{tabular}[c]{@{}c@{}} 
M\\N
\\ I\\S\\T \end{tabular} }

 &No attack  &89.2$\pm${0.6}  &92.4$\pm${0.1}  &84.9$\pm${0.6}  &84.6$\pm${0.4}  &86.7$\pm${0.2}  &89.0$\pm${0.8}  &92.2$\pm${0.2}  &92.0$\pm${0.3}  &92.4$\pm${0.1}

 \\ \cline{2-11} &Lie \cite{baruch2019little}  &89.3$\pm${0.3}  &88.3$\pm${0.3}  &86.8$\pm${0.2}  &86.8$\pm${0.2}  &86.7$\pm${0.3}  &89.0$\pm${0.3}  &92.3$\pm${0.2}  &92.3$\pm${0.1}  &92.4$\pm${0.1}

 \\ \cline{2-11} &Fang-v1 \cite{fang2020local}  &89.3$\pm${0.4}  &80.4$\pm${1.6}  &53.6$\pm${4.4}  &52.8$\pm${2.4}  &82.9$\pm${1.3}  &78.9$\pm${1.4}  &92.3$\pm${0.3}  &92.3$\pm${0.2}  &92.1$\pm${0.1}

 \\ \cline{2-11} &Fang-v2 \cite{fang2020local}  &47.8$\pm${5.7}  &70.2$\pm${2.4}  &50.1$\pm${4.5}  &47.3$\pm${4.2}  &34.0$\pm${2.3}  &80.3$\pm${0.7}  &90.3$\pm${0.9}  &91.3$\pm${0.4}  &92.3$\pm${0.0}

 \\ \cline{2-11} &Same value \cite{li2019rsa}  &89.3$\pm${0.2}  &88.5$\pm${0.1}  &69.7$\pm${1.9}  &70.8$\pm${1.6}  &85.7$\pm${0.5}  &88.9$\pm${0.5}  &92.3$\pm${0.2}  &92.4$\pm${0.1}  &92.3$\pm${0.1}

 \\ \cline{2-11} &Gaussian \cite{blanchard2017machine}  &89.4$\pm${0.4}  &92.1$\pm${0.2}  &86.1$\pm${0.4}  &86.2$\pm${0.6}  &86.7$\pm${0.2}  &87.9$\pm${1.1}  &80.3$\pm${0.6}  &31.1$\pm${2.2}  &92.3$\pm${0.0}

 \\ \cline{2-11} &sign flipping \cite{data2021byzantine} &89.2$\pm${0.2}  &91.9$\pm${0.2}  &86.3$\pm${0.4}  &86.4$\pm${0.6}  &86.3$\pm${0.5}  &76.9$\pm${2.7}  &91.3$\pm${0.2}  &76.9$\pm${1.5}  &91.6$\pm${0.4}

 \\ \cline{2-11} &label flipping \cite{data2021byzantine} &89.4$\pm${0.2}  &90.8$\pm${0.1}  &78.1$\pm${1.6}  &78.3$\pm${1.5}  &86.5$\pm${0.4}  &88.9$\pm${1.0}  &90.9$\pm${0.1}  &89.4$\pm${0.5}  &91.3$\pm${0.9}

 \\ \cline{2-11} &minic  \cite{karimireddy2020byzantine} &86.6$\pm${4.2}  &89.9$\pm${0.4}  &88.5$\pm${1.1}  &90.4$\pm${0.3}  &75.6$\pm${1.2}  &90.7$\pm${0.2}  &92.2$\pm${0.2}  &92.1$\pm${0.1}  &92.3$\pm${0.1}

 \\ \cline{2-11} &Collusion (Ours)  &88.9$\pm${0.3}  &86.3$\pm${1.2}  &80.0$\pm${1.4}  &79.4$\pm${3.4}  &87.3$\pm${1.4}  &88.9$\pm${0.6}  &91.5$\pm${0.2}  &48.0$\pm${11.0}  &92.3$\pm${0.1} 

 \\ \cline{2-11} &Averaged 
 &84.8$\pm${1.3}  &87.1$\pm${0.7}  &76.4$\pm${1.7}  &76.3$\pm${1.5}  &79.8$\pm${0.8}  &85.9$\pm${0.9}  &90.6$\pm${0.3}  &79.8$\pm${1.6}  &92.1$\pm${0.2} 

 \\ \cline{2-11} &Worst-case 
 &47.8$\pm${5.7}  &70.2$\pm${2.4}  &50.1$\pm${4.5}  &47.3$\pm${4.2}  &34.0$\pm${2.3}  &76.9$\pm${2.7}  &80.3$\pm${0.6}  &31.1$\pm${2.2}  &91.3$\pm${0.9}

\\ \hline

\end{tabular}
\caption{Model performances of different Byzantine-resilient methods under different attacks (with Non-IID setting $\beta = 0.5$ and 20 attackers for classification of MNIST).}

\end{table*}

\begin{table*}[htbp]

\centering 
\scriptsize
\setlength{\tabcolsep}{1.0mm}
\renewcommand\arraystretch{1.5}
\begin{tabular}{|c||c|c|c|c|c|c|c|c|c|c|c|}
\hline

&&\begin{tabular}[c]{@{}c@{}} 
Krum \\ \cite{blanchard2017machine} \end{tabular} &\begin{tabular}[c]{@{}c@{}}GeoMeidan\\ \cite{chen2017distributed}\end{tabular}&\begin{tabular}[c]{@{}c@{}}Median\\ \cite{yin2018byzantine}\end{tabular}&\begin{tabular}[c]{@{}c@{}}Trimmed\\ \cite{yin2018byzantine}\end{tabular}&\begin{tabular}[c]{@{}c@{}} Bulyan \\ \cite{guerraoui2018hidden}\end{tabular}&\begin{tabular}[c]{@{}c@{}} FLtrust \\ \cite{cao2020fltrust}\end{tabular}&\begin{tabular}[c]{@{}c@{}} DnC\\ \cite{shejwalkar2021manipulating}\end{tabular}&\begin{tabular}[c]{@{}c@{}}Kmeans\\ \cite{shen2016auror}\end{tabular}&\begin{tabular}[c]{@{}c@{}}FedCut\\ (Ours)\end{tabular} \\ \hline \hline

\multirow{9}{*}{\begin{tabular}[c]{@{}c@{}} 
M\\N
\\ I\\S\\T \end{tabular} }

 &No attack  &90.4$\pm${0.2}  &92.4$\pm${0.1}  &90.4$\pm${0.1}  &90.4$\pm${0.2}  &89.9$\pm${0.4}  &89.4$\pm${0.1}  &92.3$\pm${0.3}  &92.3$\pm${0.2}  &92.3$\pm${0.1}

 \\ \cline{2-11} &Lie \cite{baruch2019little}  &90.3$\pm${0.3}  &89.6$\pm${0.1}  &87.8$\pm${0.5}  &87.5$\pm${0.1}  &88.6$\pm${0.1}  &89.7$\pm${0.2}  &92.3$\pm${0.2}  &92.3$\pm${0.3}  &92.4$\pm${0.1}

 \\ \cline{2-11} &Fang-v1 \cite{fang2020local}  &90.4$\pm${0.2}  &78.5$\pm${1.4}  &67.4$\pm${1.9}  &66.8$\pm${1.0}  &88.3$\pm${0.2}  &74.8$\pm${3.0}  &92.2$\pm${0.3}  &92.2$\pm${0.3}  &92.4$\pm${0.0}

 \\ \cline{2-11} &Fang-v2 \cite{fang2020local}  &47.1$\pm${2.4}  &72.8$\pm${1.0}  &69.0$\pm${1.1}  &69.1$\pm${0.5}  &31.6$\pm${2.8}  &79.3$\pm${1.7}  &90.7$\pm${0.5}  &91.4$\pm${0.3}  &92.3$\pm${0.1}

 \\ \cline{2-11} &Same value \cite{li2019rsa}  &90.5$\pm${0.1}  &89.8$\pm${0.4}  &83.4$\pm${0.7}  &82.4$\pm${1.2}  &89.6$\pm${0.1}  &89.6$\pm${0.2}  &92.3$\pm${0.3}  &92.2$\pm${0.2}  &92.4$\pm${0.1}

 \\ \cline{2-11} &Gaussian \cite{blanchard2017machine}  &90.5$\pm${0.1}  &92.4$\pm${0.1}  &90.7$\pm${0.1}  &90.4$\pm${0.1}  &90.0$\pm${0.0}  &87.8$\pm${0.9}  &79.4$\pm${0.6}  &31.9$\pm${4.8}  &92.4$\pm${0.1}

 \\ \cline{2-11} &sign flipping \cite{data2021byzantine} &90.5$\pm${0.4}  &92.1$\pm${0.1}  &90.5$\pm${0.2}  &90.4$\pm${0.1}  &90.1$\pm${0.1}  &73.9$\pm${4.7}  &91.8$\pm${0.1}  &65.8$\pm${5.1}  &92.2$\pm${0.1}

 \\ \cline{2-11} &label flipping \cite{data2021byzantine} &90.5$\pm${0.3}  &90.6$\pm${0.1}  &88.0$\pm${0.1}  &87.9$\pm${0.2}  &90.1$\pm${0.2}  &89.4$\pm${0.5}  &92.3$\pm${0.3}  &92.0$\pm${0.1}  &92.0$\pm${0.3}

 \\ \cline{2-11} &minic  \cite{karimireddy2020byzantine} &87.3$\pm${1.1}  &86.5$\pm${0.3}  &90.5$\pm${0.1}  &91.1$\pm${0.2}  &88.2$\pm${0.6}  &90.2$\pm${0.5}  &92.3$\pm${0.1}  &92.2$\pm${0.1}  &92.2$\pm${0.1}

 \\ \cline{2-11} &Collusion (Ours)  &90.4$\pm${0.2}  &88.0$\pm${2.0}  &84.7$\pm${1.7}  &83.5$\pm${0.4}  &90.0$\pm${0.2}  &89.1$\pm${0.7}  &91.5$\pm${0.2}  &46.9$\pm${7.1}  &92.1$\pm${0.1} 

 \\ \cline{2-11} &Averaged 
 &85.8$\pm${0.5}  &87.3$\pm${0.6}  &84.2$\pm${0.6}  &83.9$\pm${0.4}  &83.6$\pm${0.5}  &85.3$\pm${1.3}  &90.7$\pm${0.3}  &78.9$\pm${1.9}  &92.3$\pm${0.1} 

 \\ \cline{2-11} &Worst-case 
 &47.1$\pm${2.4}  &72.8$\pm${1.0}  &67.4$\pm${1.9}  &66.8$\pm${1.0}  &31.6$\pm${2.8}  &73.9$\pm${4.7}  &79.4$\pm${0.6}  &31.9$\pm${4.8}  &92.0$\pm${0.3}

\\ \hline

\end{tabular}
\caption{Model performances of different Byzantine-resilient methods under different attacks (with IID and 20 attackers for classification of MNIST).}

\end{table*}

\begin{table*}[htbp]

\centering 
\scriptsize
\setlength{\tabcolsep}{1.0mm}
\renewcommand\arraystretch{1.5}
\begin{tabular}{|c||c|c|c|c|c|c|c|c|c|c|c|}
\hline

&&\begin{tabular}[c]{@{}c@{}} 
Krum \\ \cite{blanchard2017machine} \end{tabular} &\begin{tabular}[c]{@{}c@{}}GeoMeidan\\ \cite{chen2017distributed}\end{tabular}&\begin{tabular}[c]{@{}c@{}}Median\\ \cite{yin2018byzantine}\end{tabular}&\begin{tabular}[c]{@{}c@{}}Trimmed\\ \cite{yin2018byzantine}\end{tabular}&\begin{tabular}[c]{@{}c@{}} Bulyan \\ \cite{guerraoui2018hidden}\end{tabular}&\begin{tabular}[c]{@{}c@{}} FLtrust \\ \cite{cao2020fltrust}\end{tabular}&\begin{tabular}[c]{@{}c@{}} DnC\\ \cite{shejwalkar2021manipulating}\end{tabular}&\begin{tabular}[c]{@{}c@{}}Kmeans\\ \cite{shen2016auror}\end{tabular}&\begin{tabular}[c]{@{}c@{}}FedCut\\ (Ours)\end{tabular} \\ \hline \hline

\multirow{9}{*}{\begin{tabular}[c]{@{}c@{}} 
M\\N
\\ I\\S\\T \end{tabular} }

 &No attack  &62.6$\pm${4.8}  &92.1$\pm${0.0}  &66.1$\pm${5.3}  &64.1$\pm${3.2}  &73.7$\pm${1.5}  &90.1$\pm${0.5}  &91.5$\pm${0.5}  &91.6$\pm${0.2}  &92.3$\pm${0.1}

 \\ \cline{2-11} &Lie \cite{baruch2019little}  &68.8$\pm${0.9}  &78.7$\pm${0.4}  &79.4$\pm${1.3}  &79.1$\pm${0.8}  &71.5$\pm${3.1}  &86.6$\pm${2.0}  &92.0$\pm${0.2}  &92.1$\pm${0.3}  &92.1$\pm${0.1}

 \\ \cline{2-11} &Fang-v1 \cite{fang2020local}  &68.7$\pm${3.5}  &49.7$\pm${8.3}  &20.8$\pm${3.5}  &20.9$\pm${2.1}  &54.0$\pm${3.3}  &80.2$\pm${2.5}  &91.9$\pm${0.3}  &92.0$\pm${0.3}  &92.0$\pm${0.1}

 \\ \cline{2-11} &Fang-v2 \cite{fang2020local}  &30.7$\pm${17.1}  &34.2$\pm${4.8}  &26.9$\pm${6.6}  &25.2$\pm${2.6}  &12.2$\pm${6.2}  &80.8$\pm${2.0}  &82.3$\pm${4.6}  &88.4$\pm${1.3}  &92.0$\pm${0.0}

 \\ \cline{2-11} &Same value \cite{li2019rsa}  &66.8$\pm${2.9}  &77.7$\pm${0.2}  &46.9$\pm${5.4}  &47.9$\pm${2.5}  &70.1$\pm${3.3}  &89.9$\pm${0.3}  &92.0$\pm${0.2}  &92.1$\pm${0.3}  &92.0$\pm${0.1}

 \\ \cline{2-11} &Gaussian \cite{blanchard2017machine}  &67.9$\pm${0.5}  &91.7$\pm${0.0}  &72.1$\pm${4.0}  &71.4$\pm${2.7}  &70.6$\pm${5.1}  &89.0$\pm${0.5}  &74.9$\pm${1.3}  &20.2$\pm${2.2}  &92.1$\pm${0.1}

 \\ \cline{2-11} &sign flipping \cite{data2021byzantine} &66.3$\pm${3.9}  &88.6$\pm${0.5}  &73.6$\pm${4.1}  &73.0$\pm${4.2}  &72.0$\pm${4.9}  &76.9$\pm${1.6}  &52.5$\pm${1.2}  &37.5$\pm${6.4}  &89.5$\pm${0.4}

 \\ \cline{2-11} &label flipping \cite{data2021byzantine} &64.1$\pm${2.0}  &88.6$\pm${0.3}  &47.2$\pm${1.4}  &42.7$\pm${2.8}  &74.2$\pm${3.5}  &89.9$\pm${0.6}  &86.7$\pm${0.8}  &82.9$\pm${3.3}  &89.6$\pm${0.4}

 \\ \cline{2-11} &minic  \cite{karimireddy2020byzantine} &69.9$\pm${2.8}  &80.1$\pm${5.1}  &74.9$\pm${2.6}  &82.0$\pm${2.4}  &47.7$\pm${3.5}  &89.8$\pm${0.3}  &90.5$\pm${0.2}  &90.7$\pm${0.5}  &91.7$\pm${0.3}

 \\ \cline{2-11} &Collusion (Ours)  &69.6$\pm${6.8}  &79.7$\pm${1.5}  &66.3$\pm${2.9}  &66.2$\pm${3.4}  &75.6$\pm${7.7}  &89.5$\pm${0.3}  &88.9$\pm${1.3}  &39.0$\pm${16.9}  &91.9$\pm${0.2} 

 \\ \cline{2-11} &Averaged 
 &63.5$\pm${4.5}  &76.1$\pm${2.1}  &57.4$\pm${3.7}  &57.3$\pm${2.7}  &62.2$\pm${4.2}  &86.3$\pm${1.1}  &84.3$\pm${1.1}  &72.7$\pm${3.2}  &91.5$\pm${0.2} 

 \\ \cline{2-11} &Worst-case 
 &30.7$\pm${17.1}  &34.2$\pm${4.8}  &20.8$\pm${3.5}  &20.9$\pm${2.1}  &12.2$\pm${6.2}  &76.9$\pm${1.6}  &52.5$\pm${1.2}  &20.2$\pm${2.2}  &89.5$\pm${0.4}

\\ \hline

\end{tabular}
\caption{Model performances of different Byzantine-resilient methods under different attacks (with Non-IID setting $\beta = 0.1$ and 30 attackers for classification of MNIST).}

\end{table*}

\begin{table*}[htbp]

\centering 
\scriptsize
\setlength{\tabcolsep}{1.0mm}
\renewcommand\arraystretch{1.5}
\begin{tabular}{|c||c|c|c|c|c|c|c|c|c|c|c|}
\hline

&&\begin{tabular}[c]{@{}c@{}} 
Krum \\ \cite{blanchard2017machine} \end{tabular} &\begin{tabular}[c]{@{}c@{}}GeoMeidan\\ \cite{chen2017distributed}\end{tabular}&\begin{tabular}[c]{@{}c@{}}Median\\ \cite{yin2018byzantine}\end{tabular}&\begin{tabular}[c]{@{}c@{}}Trimmed\\ \cite{yin2018byzantine}\end{tabular}&\begin{tabular}[c]{@{}c@{}} Bulyan \\ \cite{guerraoui2018hidden}\end{tabular}&\begin{tabular}[c]{@{}c@{}} FLtrust \\ \cite{cao2020fltrust}\end{tabular}&\begin{tabular}[c]{@{}c@{}} DnC\\ \cite{shejwalkar2021manipulating}\end{tabular}&\begin{tabular}[c]{@{}c@{}}Kmeans\\ \cite{shen2016auror}\end{tabular}&\begin{tabular}[c]{@{}c@{}}FedCut\\ (Ours)\end{tabular} \\ \hline \hline

\multirow{9}{*}{\begin{tabular}[c]{@{}c@{}} 
M\\N
\\ I\\S\\T \end{tabular} }

 &No attack  &88.9$\pm${0.3}  &92.2$\pm${0.2}  &85.2$\pm${0.2}  &84.8$\pm${0.7}  &86.0$\pm${0.2}  &88.8$\pm${0.4}  &92.2$\pm${0.3}  &92.1$\pm${0.4}  &92.4$\pm${0.1}

 \\ \cline{2-11} &Lie \cite{baruch2019little}  &89.1$\pm${0.1}  &82.2$\pm${0.6}  &82.3$\pm${1.5}  &81.6$\pm${1.0}  &83.8$\pm${0.8}  &88.9$\pm${0.5}  &92.3$\pm${0.3}  &92.3$\pm${0.2}  &92.2$\pm${0.0}

 \\ \cline{2-11} &Fang-v1 \cite{fang2020local}  &89.3$\pm${0.2}  &51.0$\pm${3.1}  &35.3$\pm${7.0}  &34.3$\pm${2.9}  &78.1$\pm${0.4}  &78.3$\pm${3.3}  &92.2$\pm${0.2}  &92.3$\pm${0.1}  &92.2$\pm${0.0}

 \\ \cline{2-11} &Fang-v2 \cite{fang2020local}  &33.5$\pm${6.9}  &40.6$\pm${3.3}  &40.6$\pm${4.4}  &34.9$\pm${3.3}  &14.9$\pm${3.3}  &81.7$\pm${1.6}  &84.4$\pm${1.8}  &87.4$\pm${0.5}  &92.3$\pm${0.1}

 \\ \cline{2-11} &Same value \cite{li2019rsa}  &88.9$\pm${0.6}  &82.4$\pm${0.7}  &55.7$\pm${0.3}  &53.7$\pm${2.8}  &85.5$\pm${1.2}  &88.8$\pm${0.8}  &92.3$\pm${0.2}  &92.2$\pm${0.1}  &92.2$\pm${0.1}

 \\ \cline{2-11} &Gaussian \cite{blanchard2017machine}  &89.2$\pm${0.3}  &92.3$\pm${0.1}  &87.1$\pm${0.7}  &87.4$\pm${0.4}  &85.7$\pm${0.4}  &87.7$\pm${0.6}  &74.0$\pm${1.7}  &20.3$\pm${4.4}  &92.3$\pm${0.0}

 \\ \cline{2-11} &sign flipping \cite{data2021byzantine} &89.0$\pm${0.6}  &91.2$\pm${0.3}  &87.2$\pm${0.3}  &87.5$\pm${0.5}  &86.0$\pm${0.7}  &75.4$\pm${0.6}  &83.1$\pm${1.5}  &64.7$\pm${9.1}  &91.5$\pm${0.9}

 \\ \cline{2-11} &label flipping \cite{data2021byzantine} &89.0$\pm${0.5}  &89.3$\pm${0.1}  &73.3$\pm${2.3}  &71.3$\pm${1.6}  &86.7$\pm${0.3}  &88.8$\pm${0.9}  &88.9$\pm${0.2}  &88.3$\pm${0.7}  &90.7$\pm${0.7}

 \\ \cline{2-11} &minic  \cite{karimireddy2020byzantine} &73.4$\pm${5.4}  &82.9$\pm${1.4}  &83.1$\pm${3.1}  &87.7$\pm${0.3}  &72.5$\pm${8.6}  &90.8$\pm${0.5}  &91.8$\pm${0.3}  &91.9$\pm${0.2}  &92.2$\pm${0.1}

 \\ \cline{2-11} &Collusion (Ours)  &89.0$\pm${0.5}  &82.2$\pm${1.2}  &73.6$\pm${1.2}  &72.3$\pm${1.2}  &86.7$\pm${1.2}  &89.2$\pm${0.5}  &89.4$\pm${0.5}  &31.2$\pm${2.2}  &92.2$\pm${0.0} 

 \\ \cline{2-11} &Averaged 
 &81.9$\pm${1.5}  &78.6$\pm${1.1}  &70.3$\pm${2.1}  &69.6$\pm${1.5}  &76.6$\pm${1.7}  &85.8$\pm${1.0}  &88.1$\pm${0.7}  &75.3$\pm${1.8}  &92.0$\pm${0.2} 

 \\ \cline{2-11} &Worst-case 
 &33.5$\pm${6.9}  &40.6$\pm${3.3}  &35.3$\pm${7.0}  &34.3$\pm${2.9}  &14.9$\pm${3.3}  &75.4$\pm${0.6}  &74.0$\pm${1.7}  &20.3$\pm${4.4}  &90.7$\pm${0.7}

\\ \hline

\end{tabular}
\caption{Model performances of different Byzantine-resilient methods under different attacks (with Non-IID setting $\beta = 0.5$ and 30 attackers for classification of MNIST).}

\end{table*}

\begin{table*}[htbp]

\centering 
\scriptsize
\setlength{\tabcolsep}{1.0mm}
\renewcommand\arraystretch{1.5}
\begin{tabular}{|c||c|c|c|c|c|c|c|c|c|c|c|}
\hline

&&\begin{tabular}[c]{@{}c@{}} 
Krum \\ \cite{blanchard2017machine} \end{tabular} &\begin{tabular}[c]{@{}c@{}}GeoMeidan\\ \cite{chen2017distributed}\end{tabular}&\begin{tabular}[c]{@{}c@{}}Median\\ \cite{yin2018byzantine}\end{tabular}&\begin{tabular}[c]{@{}c@{}}Trimmed\\ \cite{yin2018byzantine}\end{tabular}&\begin{tabular}[c]{@{}c@{}} Bulyan \\ \cite{guerraoui2018hidden}\end{tabular}&\begin{tabular}[c]{@{}c@{}} FLtrust \\ \cite{cao2020fltrust}\end{tabular}&\begin{tabular}[c]{@{}c@{}} DnC\\ \cite{shejwalkar2021manipulating}\end{tabular}&\begin{tabular}[c]{@{}c@{}}Kmeans\\ \cite{shen2016auror}\end{tabular}&\begin{tabular}[c]{@{}c@{}}FedCut\\ (Ours)\end{tabular} \\ \hline \hline

\multirow{9}{*}{\begin{tabular}[c]{@{}c@{}} 
C\\I
\\ F\\A\\R\\10 \end{tabular} }

 &No attack  &29.1$\pm${7.0}  &53.0$\pm${8.0}  &16.7$\pm${5.1}  &19.1$\pm${9.4}  &17.3$\pm${3.4}  &56.3$\pm${2.8}  &56.5$\pm${12.8}  &54.3$\pm${12.9}  &66.7$\pm${0.7}

 \\ \cline{2-11} &Lie \cite{baruch2019little}  &26.6$\pm${6.5}  &11.2$\pm${2.5}  &9.8$\pm${0.2}  &11.8$\pm${2.4}  &10.0$\pm${0.2}  &11.2$\pm${1.8}  &10.4$\pm${0.7}  &47.0$\pm${18.2}  &63.8$\pm${0.2}

 \\ \cline{2-11} &Fang-v1 \cite{fang2020local}  &25.5$\pm${6.7}  &32.5$\pm${7.3}  &11.3$\pm${2.5}  &16.2$\pm${3.8}  &13.4$\pm${2.0}  &24.5$\pm${19.5}  &36.5$\pm${23.4}  &55.5$\pm${15.7}  &63.9$\pm${1.3}

 \\ \cline{2-11} &Fang-v2 \cite{fang2020local}  &10.0$\pm${0.0}  &10.0$\pm${0.0}  &9.9$\pm${0.3}  &10.0$\pm${0.2}  &10.0$\pm${0.0}  &10.0$\pm${0.0}  &55.9$\pm${10.0}  &10.0$\pm${0.0}  &64.4$\pm${1.5}

 \\ \cline{2-11} &Same value \cite{li2019rsa}  &26.8$\pm${6.5}  &10.0$\pm${0.0}  &13.4$\pm${3.7}  &10.0$\pm${3.6}  &25.9$\pm${11.6}  &49.3$\pm${18.6}  &45.7$\pm${13.0}  &44.9$\pm${21.9}  &62.3$\pm${1.4}

 \\ \cline{2-11} &Gaussian \cite{blanchard2017machine}  &20.7$\pm${9.9}  &27.6$\pm${20.3}  &16.6$\pm${2.5}  &27.9$\pm${11.9}  &11.7$\pm${2.0}  &55.5$\pm${2.0}  &26.3$\pm${7.9}  &13.3$\pm${1.3}  &61.6$\pm${0.8}

 \\ \cline{2-11} &sign flipping \cite{data2021byzantine} &27.6$\pm${8.3}  &28.1$\pm${17.9}  &16.5$\pm${3.8}  &25.6$\pm${6.9}  &14.6$\pm${4.7}  &37.4$\pm${23.2}  &14.3$\pm${3.4}  &11.2$\pm${1.5}  &60.5$\pm${1.3}

 \\ \cline{2-11} &label flipping \cite{data2021byzantine} &20.3$\pm${10.1}  &41.6$\pm${8.5}  &14.9$\pm${3.4}  &22.1$\pm${6.2}  &12.7$\pm${2.9}  &52.3$\pm${2.5}  &43.1$\pm${10.6}  &24.9$\pm${24.8}  &54.8$\pm${1.4}

 \\ \cline{2-11} &minic  \cite{karimireddy2020byzantine} &22.3$\pm${10.7}  &52.6$\pm${1.7}  &15.2$\pm${2.4}  &21.3$\pm${7.2}  &13.3$\pm${1.5}  &10.0$\pm${0.0}  &50.6$\pm${12.0}  &42.6$\pm${28.2}  &63.4$\pm${0.9}

 \\ \cline{2-11} &Collusion (Ours)  &28.6$\pm${6.7}  &10.0$\pm${0.1}  &9.8$\pm${0.3}  &9.9$\pm${0.4}  &13.2$\pm${3.4}  &55.4$\pm${3.1}  &29.5$\pm${7.8}  &10.1$\pm${0.2}  &62.0$\pm${1.1} 

 \\ \cline{2-11} &Averaged 
 &23.8$\pm${7.2}  &27.7$\pm${6.6}  &13.4$\pm${2.4}  &17.4$\pm${5.2}  &14.2$\pm${3.2}  &36.2$\pm${7.4}  &36.9$\pm${10.2}  &31.4$\pm${12.5}  &62.3$\pm${1.1} 

 \\ \cline{2-11} &Worst-case 
 &10.0$\pm${0.0}  &10.0$\pm${0.0}  &9.8$\pm${0.2}  &9.9$\pm${0.4}  &10.0$\pm${0.2}  &10.0$\pm${0.0}  &10.4$\pm${0.7}  &10.0$\pm${0.0}  &54.8$\pm${1.4}

\\ \hline

\end{tabular}
\caption{Model performances of different Byzantine-resilient methods under different attacks (with Non-IID setting $\beta = 0.1$ and and 6 attackers for classification of CIFAR10).}
\label{tab:MP18}
\end{table*}

\begin{table*}[htbp]

\centering 
\scriptsize
\setlength{\tabcolsep}{1.0mm}
\renewcommand\arraystretch{1.5}
\begin{tabular}{|c||c|c|c|c|c|c|c|c|c|c|c|}
\hline

&&\begin{tabular}[c]{@{}c@{}} 
Krum \\ \cite{blanchard2017machine} \end{tabular} &\begin{tabular}[c]{@{}c@{}}GeoMeidan\\ \cite{chen2017distributed}\end{tabular}&\begin{tabular}[c]{@{}c@{}}Median\\ \cite{yin2018byzantine}\end{tabular}&\begin{tabular}[c]{@{}c@{}}Trimmed\\ \cite{yin2018byzantine}\end{tabular}&\begin{tabular}[c]{@{}c@{}} Bulyan \\ \cite{guerraoui2018hidden}\end{tabular}&\begin{tabular}[c]{@{}c@{}} FLtrust \\ \cite{cao2020fltrust}\end{tabular}&\begin{tabular}[c]{@{}c@{}} DnC\\ \cite{shejwalkar2021manipulating}\end{tabular}&\begin{tabular}[c]{@{}c@{}}Kmeans\\ \cite{shen2016auror}\end{tabular}&\begin{tabular}[c]{@{}c@{}}FedCut\\ (Ours)\end{tabular} \\ \hline \hline

\multirow{9}{*}{\begin{tabular}[c]{@{}c@{}} 
C\\I
\\ F\\A\\R\\10 \end{tabular} }

 &No attack  &25.2$\pm${16.8}  &55.9$\pm${13.0}  &18.5$\pm${10.2}  &59.0$\pm${13.7}  &17.7$\pm${7.4}  &62.9$\pm${0.2}  &59.9$\pm${15.2}  &56.9$\pm${16.3}  &66.7$\pm${1.6}

 \\ \cline{2-11} &Lie \cite{baruch2019little}  &37.5$\pm${9.1}  &10.0$\pm${0.0}  &10.0$\pm${0.1}  &11.7$\pm${2.9}  &10.7$\pm${1.2}  &16.7$\pm${7.6}  &10.0$\pm${0.0}  &31.5$\pm${30.8}  &67.3$\pm${0.5}

 \\ \cline{2-11} &Fang-v1 \cite{fang2020local}  &35.2$\pm${10.5}  &49.9$\pm${16.7}  &15.1$\pm${3.1}  &23.2$\pm${5.0}  &15.8$\pm${3.6}  &43.5$\pm${29.0}  &56.2$\pm${17.1}  &56.4$\pm${16.0}  &64.1$\pm${0.7}

 \\ \cline{2-11} &Fang-v2 \cite{fang2020local}  &10.0$\pm${0.0}  &10.0$\pm${0.0}  &13.8$\pm${2.4}  &13.1$\pm${6.4}  &10.0$\pm${0.1}  &10.0$\pm${0.0}  &63.3$\pm${4.0}  &15.9$\pm${4.8}  &65.4$\pm${0.6}

 \\ \cline{2-11} &Same value \cite{li2019rsa}  &37.8$\pm${8.9}  &15.1$\pm${8.8}  &10.2$\pm${0.3}  &10.7$\pm${1.2}  &37.1$\pm${16.4}  &59.5$\pm${0.7}  &58.3$\pm${15.3}  &56.2$\pm${17.1}  &62.9$\pm${0.9}

 \\ \cline{2-11} &Gaussian \cite{blanchard2017machine}  &37.5$\pm${11.5}  &56.1$\pm${14.2}  &13.7$\pm${3.0}  &47.9$\pm${27.7}  &15.2$\pm${4.6}  &62.3$\pm${1.6}  &26.6$\pm${8.4}  &12.6$\pm${0.5}  &64.9$\pm${1.2}

 \\ \cline{2-11} &sign flipping \cite{data2021byzantine} &36.3$\pm${11.1}  &51.1$\pm${20.8}  &17.0$\pm${2.2}  &41.7$\pm${12.1}  &16.2$\pm${3.1}  &60.8$\pm${0.6}  &17.4$\pm${1.4}  &20.6$\pm${9.5}  &63.8$\pm${0.6}

 \\ \cline{2-11} &label flipping \cite{data2021byzantine} &20.1$\pm${18.1}  &47.1$\pm${23.7}  &14.5$\pm${5.6}  &42.3$\pm${9.4}  &12.8$\pm${4.5}  &58.5$\pm${2.3}  &48.3$\pm${13.6}  &50.6$\pm${16.3}  &55.7$\pm${2.7}

 \\ \cline{2-11} &minic  \cite{karimireddy2020byzantine} &40.0$\pm${3.4}  &60.7$\pm${5.4}  &22.7$\pm${2.0}  &62.1$\pm${1.8}  &17.7$\pm${1.7}  &10.0$\pm${0.0}  &63.2$\pm${1.6}  &45.5$\pm${30.8}  &67.1$\pm${1.4}

 \\ \cline{2-11} &Collusion (Ours)  &39.2$\pm${10.8}  &10.0$\pm${0.0}  &10.0$\pm${0.1}  &10.2$\pm${0.5}  &14.1$\pm${4.8}  &62.3$\pm${0.2}  &14.1$\pm${4.0}  &9.9$\pm${0.2}  &66.2$\pm${0.6} 

 \\ \cline{2-11} &Averaged 
 &31.9$\pm${10.0}  &36.6$\pm${10.3}  &14.6$\pm${2.9}  &32.2$\pm${8.1}  &16.7$\pm${4.7}  &44.7$\pm${4.2}  &41.7$\pm${8.1}  &35.6$\pm${14.2}  &64.4$\pm${1.1} 

 \\ \cline{2-11} &Worst-case 
 &10.0$\pm${0.0}  &10.0$\pm${0.0}  &10.0$\pm${0.1}  &10.2$\pm${0.5}  &10.0$\pm${0.1}  &10.0$\pm${0.0}  &10.0$\pm${0.0}  &9.9$\pm${0.2}  &55.7$\pm${2.7}

\\ \hline

\end{tabular}
\caption{Model performances of different Byzantine-resilient methods under different attacks (with Non-IID setting $\beta = 0.5$ and and 6 attackers for classification of CIFAR10).}
\label{tab:MP19}
\end{table*}

\subsection{FedCut v.s. FedCut without Temporal Consistency}
In this part, we compare FedCut and NCut, which remove the temporal consistency. Fig. \ref{fig:fedcutvsncut} shows the averaged MP of FedCut and NCut, which demonstrates the stability of FedCut while the averaged MP of NCut is highly influenced by the Non-IID extent of clients' data (e.g., the MP of Ncut drops to 80\% with Non-IID parameter $\beta =0.1$). The reason temporal consistency would help to distinguish benign and byzantine clients during the training process.

\subsection{Robustness for Spectral Heuristics in Heterogeneous Dataset}
In this part, we demonstrate the effectiveness of estimating the number of communities, Gaussian scaling factors via the largest eigengap even for the heterogeneous dataset among clients.

Fig. \ref{fig:eigengap-clusternum1-app} and \ref{fig:eigengap-clusternum2-app} display the position of the largest eigengap for Non-IID dataset ($q=0.1, 0.5$), which illustrates the largest eigengap is also a good estimation for the number of communities. Moreover, Fig. \ref{fig:eigengap-gamma-app} also shows the eigengap and clustering accuracy could reach the optimal at the same time even for Non-IID dataset ($q=0.1, 0.5$), which demonstrates we could select the Gaussian scaling factor according to the largest eigengap.

\clearpage

\section{Proof} \label{sec:appE}

We consider a \textit{horizontal federated learning} \cite{yang2019federated,mcmahan2017communication} setting consisting of one server and $K$ clients. We assume $K$ clients\footnote{In this article we use terms "client", "node", "participant" and "party" interchangeably. } have their local dataset $\calD_i = \{(\bx_{i,j},y_{i,j} \}_{j=1}^{n_i}, i=1\cdots K$, where $\bx_{i,j}$ is the input data, $y_{i,j}$ is the label and $n_i$ is the total number of data points for $i_{th}$ client. The training in federated learning is divided into three steps which iteratively run until the learning converges:
\begin{itemize}
    \item The $i_{th}$ client takes empirical risk minimization as:
    \begin{equation}
        \min_{\bw_i} F_{i}(\bw_i, \calD_i) =\min_{\bw_i} \frac{1}{n_i} \sum_{j=1}^{n_i} \ell(\bw_i, \bx_{i,j}, y_{i,j}),
    \end{equation}
    where $\bw_i \in \RR^d$ is the $i_{th}$ client's local model weight and $\ell(\cdot)$ is a loss function that measures the accuracy of the prediction made by the model on each data point.
    \item Each client sends respective local model updates $\nabla F_i$ to the server and the server updates the \textit{global model} $\bw$ as $\bw = \bw - \eta \frac{1}{K}\sum_{i=1}^K\nabla F_i$, where $\eta$ is learning rate.
    \item The server distributes the updated global model $\bw$ to all clients.
\end{itemize}

We assume a malicious threat mode where an unknown number of participants out of K clients are Byzantine, i.e., they may upload arbitrarily corrupt updates $\bg_b$ to degrade the global model performance (MP). Under this assumption, behaviours of Byzantine clients and the rest of benign clients can be summarized as follows:
\begin{equation} \label{eq:upload-gradients-app}
\bg_i = \left\{
\begin{aligned}
\nabla F_i \qquad &\text{Benign clients}\\
\bg_b  \qquad & \text{Byzantine clients}\\
\end{aligned}
\right.
\end{equation}

Moreover, we regard model updates contributed by $K$ clients as an undirected graph $G = (V,E)$, where $V={v_1,\cdots, v_K}$ represent $K$ model updates, $E$ is a set of weighted edge representing similarities between uploaded model updates corresponding to clients in $V$.
We assume that the graph $G = (V,E)$ is weighted, and each edge between two nodes $v_i$ and $v_j$ carries a non-negative weight, e.g., $A_{ij} = \text{exp}(-||\bg_i-\bg_j||^2/2\sigma^2) \geq0$, where $\bg_i$ is uploaded gradient for $i_{th}$ client and $\sigma$ is the Gaussian scaling factor. Let $G_R = (V_R, E_R)$ and $G_B =(V_B, E_B)$ respectively denote \textit{two subgraphs} of $G$ representing benign and Byzantine clients.

\subsection{Proof of Proposition 1}

\begin{lem}\label{lemma:clusternum-eigengap-app} \cite{von2007tutorial}
Let G be an undirected graph with non-negative weights. Then the multiplicity c of the eigenvalue 1 of $L$ equals the number of connected components $B_1,\cdots, B_c$ in the graph, 
\end{lem}
\begin{rmk}
The analysis given in \cite{ng2001spectral,von2007tutorial} shows that one could estimate $c$ by counting the number of eigenvalues equaling 1 as Lemma \ref{lemma:clusternum-eigengap-app}
\end{rmk}

\begin{lem}[\cite{stewart1990matrix}] \label{lem:matrix-perturb-app}
Let $\lambda$, $Y$ and $\delta$ be eigenvalue, principle eigenvectors and the eigengap of $L$ separately. Assume a matrix small perturbation for $L$ as $\tilde{L} = L +E$ so that $||E||$\footnote{\text{$||\cdot||$ in the paper represents the $\ell_2$ norm}} is small enough, let $\tilde{\lambda}$, $\tilde{Y}$ be eigenvalue and principle eigenvectors of $\tilde{L}$, then 
\begin{equation}
    ||\lambda -\tilde{\lambda}|| \leq ||E|| 
\end{equation}

\begin{equation}
    ||Y -\tilde{Y}|| \leq \frac{4||E||}{\delta - \sqrt{2}||E||}
\end{equation}
\end{lem}

\begin{rmk}
From the perturbation theory, Lemma \ref{lem:matrix-perturb-app} demonstrates the stability of eigenvectors when there is a small perturbation of block diagonal adjacency matrix. \\
\end{rmk}

\begin{assumption} \label{assum: Bounded gradient dissimilarity-app}
Assume the difference of local gradients $\nabla F_i$ and the mean of benign model update ${\nabla F} = \frac{1}{|V_R|}\sum_{i \in V_R} \nabla F_i$ is bounded ($V_R$ is the set of benign clients), i.e., there exists a finite $\ka$, such that
\begin{equation*}
    || \nabla F_i-{\nabla F}|| \leq \ka.
\end{equation*}
\end{assumption}

\begin{rmk}
$\ka$ in Assumption \ref{assum: Bounded gradient dissimilarity-app} has also been used earlier to bound heterogeneity in
datasets \cite{li2019communication}; Specifically, when the data is homogeneous, we have $\ka = 0$ in Assumption \ref{assum: Bounded gradient dissimilarity-app}.
\end{rmk}

\begin{prop} \label{prop:attack-type-app}
Suppose $K$ clients consist of $m$ benign clients and $q$ attackers ($q<m-1$). If Assumption \ref{assum: Bounded gradient dissimilarity-app} holds for Non-collusion and Collusion-diff attacks, then only the first $c$ eigenvalues are close to 1 and
\begin{itemize}
    \item  $c=1+q<\frac{K}{2}$ for Non-collusion attacks provided that $||\bg_b-\nabla F || \gg \ka$ and malicious updates ($\bg_b$) are far away from each other;
    \item $c=1+B<\frac{K}{2}$ for Collusion-diff attacks provided that malicious updates form $B$ groups and $||\bg_b-\nabla F || \gg \ka$;
    \item $c >\frac{K}{2}$ for Collusion-mimic attacks that $||\bg_b-\nabla F || < \ka$ and malicious updates are almost identical.
\end{itemize}
\end{prop}

\begin{proof}
For Non-Collusion attack,  $\|\bg_b - \nabla F \| > \ka $ and malicious updates ($\bg_b$) are far away from each other. Therefore, we define normalized adjacency matrix $A$ in the ideal case is block-wised where benign clients form a block with no relations to other attackers. Consequently, we obtain the largest $c$ eigenvalues of $A$ are 1 according to Lemma \ref{lemma:clusternum-eigengap-app}, where $c = 1+q < \frac{K}{2}$. In general, normalized adjacency matrix $\tilde{A} = A+E$. Based on conditions that malicious updates ($\bg_b$) are far away from each other and $||\bg_b - \nabla F|| \gg \ka$, we have $\tilde{A}_{ij} = \text{exp}(-||\bg_b - \nabla F_i||/\sigma^2) $ is small, so $||E||$ is small. Consequently, according to Lemma \ref{lem:matrix-perturb-app},
\begin{equation}
    \tilde{\lambda} - \lambda \leq ||E||
\end{equation}
As a a result, we obtain the first $q$ eigenvalues of $\tilde{A}$ are close to 1 due to the small $||E||$. \\
Similarly, we have $c =1+B< \frac{K}{2}$ for Collusion-diff attack; the first $c$ eigenvalues of $\tilde{A}$ are close to 1 for Collusion-mimic attack, and $c \geq m+1 > \frac{K}{2}$.
\end{proof}

\begin{rmk}
Proposition \ref{prop:attack-type-app} illustrates that the eigenvalue 1 has multiplicity $c$, and then there is a gap to the $(c+1)_{th}$ eigenvalue $\lambda_{c+1} <1$. Therefore, we can use the position of the obvious eigengap to elucidate the number of community clusters\footnote{We use the largest eigengap to determine the community clusters}. Moreover, Proposition \ref{prop:attack-type-app} demonstrates the different positions of obvious eigengap for Non-collusion, Collusion-diff and Collusion-mimic attacks.
\end{rmk}

\subsection{Proof of Proposition 2}

\begin{lem} \label{lem:lem1-app}
 Minimizing $\sum_{i=1}^cW(B_i, \overline{B_i})/vol(B_i)$ over all c-partition $V = B_1 \cup \cdots \cup B_c$ is equivalent with maximizing $\text{tr}(H^TD^{-1/2}A D^{-1/2}H)$ over $\{H| H\in \mathbb{R}^{K \times c}, H^T H =I \}$, where $W(B_i,B_j): = \sum_{i\in B_i, j\in B_j}A_{ij}$ and $\text{tr}$ represents the trace of matrix.
\end{lem}
\begin{proof}
Define $\bbf_i = (f_{i1}, \cdots, f_{iK})^T $
\begin{equation}
f_{ij} =\left\{
             \begin{array}{cc}
             \sqrt{\frac{1}{vol(B_i)}}, &v_j \in B_i \\
              0,& v_j \notin B_i  
             \end{array}
\right.
\end{equation}
Then 
\begin{equation} \label{}
\begin{split}
 \bbf_i^T (D-A) \bbf_i & = \frac{1}{2}\sum_{m=1}^K \sum_{n=1}^K A_{mn}(f_{im}-f_{in})^2 \\
 & = \frac{1}{2} [\sum_{v_m\in B_i,v_n \notin B_i}A_{mn}(f_{im}-f_{in})^2 \\
 & \quad + \sum_{v_n\in B_i, v_m \notin B_i}A_{mn}(f_{im}-f_{in})^2] \\
 & = \frac{1}{2} [\sum_{v_m\in B_i, v_n \notin B_i}A_{mn}(\sqrt{\frac{1}{vol(B_i)}}- 0)^2 \\
 & \quad + \sum_{v_n\in B_i, v_m \notin B_i}A_{mn}(0 -\sqrt{\frac{1}{vol(B_i)}})^2] \\
 & = W(B_i, \overline{B_i})/Vol(B_i),
\end{split}
\end{equation}
where $W(B_i,B_j): = \sum_{i\in B_i, j\in B_j}A_{ij}$. Moreover, we have $\bbf_i^T \bbf_i = 1/vol(B_i)$. Thus,
\begin{equation}
\begin{split}
    \sum_{i=1}^c  W(B_i, \overline{B_i})/Vol(B_i) &=  \sum_{i=1}^c \bbf_i^T (D-A) \bbf_i\\
    &= tr(F^T(D-A)F),
\end{split}
\end{equation}
where F is matrix combining all $\bbf_i$. Let $F = D^{-1/2} H$, then $H^TH = I$ and
\begin{equation}
\begin{split}
    \min_{( B_1 \cup \cdots \cup B_c)=V} &\sum_i^c  W(B_i, \overline{B_i})/Vol(B_i)\\
    & = \min_F tr(F^T(D-A)F) \\
    & = \min_H \text{tr}(H^TD^{-1/2}(D-A)D^{-1/2}H) \\
    & = c - \max_H \text{tr}(H^TD^{-1/2}A D^{-1/2}H)
\end{split}
\end{equation}
%Moreover, it is noted that the eigenvector corresponding to the $c$ smallest eigenvalue of $D^{-1/2}(D-A)D^{-1/2}$ equals to the eigenvector corresponding to the $c$ largest eigenvalue of $L = D^{-1/2}AD^{-1/2}$ (normalized adjacency matrix). 
% Thus lemma is proved.
\end{proof}
In federated learning, clients would send the updates to the server in each iteration, thus we consider the graph for all iterations as follows:
\begin{defi} \label{def:Spatial-Temporal graph-app}
(Spatial-Temporal Graph) Define a Spatial-Temporal graph $G = (V,E)$ as a sequence of snapshots $<G^1, \cdots, G^T>$, where $G^t=(V, E^t)$ is an undirected graph at iteration $t$. $V$ denotes a fixed set of vertexes representing model updates belonging to $K$ clients. $E^t$ is a set of weighted edge representing similarities between model updates corresponding to clients in $V$ at iteration $t$, where related adjacency matrix of the graph is  $A^t$. 
\end{defi}

The \textit{c-partition Ncut for Spatial-Temporal Graph} is thus defined as follows:
% First, the ncut is averaged over all iterations to make the detection of Byzantine colluders more consistent over multiple iterations. Second, since there may exist multiple groups of colluders against benign clients, one has to resort to the k-partition instead of bi-partition of vertexes. To this end, we propose the optimization objective as follows. 
%\HL
{
\begin{defi} (c-partition Ncut for Spatial-Temporal Graph) \label{def:partition for Spatial-Temporal graph-app}
Let $G = (V,E)$ be a Spatial-Temporal graph as Def. \ref{def:Spatial-Temporal graph-app}
Denote the c-partition for graph $G$ as $V = B_1 \cup \cdots \cup B_c$ and $B_i \cap B_j = \varnothing$ for any $i,j$, c-partition Ncut for Spatial-Temporal Graph aims to optimize:
\begin{equation} \label{eq:loss-Spatial-Temporal-app}
    \min_{( B_1 \cup \cdots \cup B_c)=V}\sum_{t=1}^T \sum_{i=1}^c\frac{W^t(B_i, \overline{B_i})}{Vol^t(B_i)}, 
\end{equation}
where $\overline{B_i}$ is the complement of $B_i$, $W^t(B_i,B_j): = \sum_{v_i\in B_i, v_j\in B_j}A^t_{ij}$, and $A_{ij}^t$ is edge weight of $B_i$ and $B_j$, and $Vol^t(B_i): = \sum_{v_i \in B}\sum_{v_j\in V} A_{ij}^t$
\end{defi}}

\begin{prop} \label{prop:prop1-app}
%Solving FedCut (Algo. \ref{algo:fedcut}) $\iff$ doing \textit{C-partition Ncut for Spatial-Temporal Graph}
FedCut (Algo. 2) solves the \textit{c-partition Ncut} in Eq. (\ref{eq:loss-Spatial-Temporal-app}).
\end{prop}

\begin{proof}
Firstly, according to the line 2 in Algo. 2,
\begin{equation}
    \begin{split}
        \tilde{L}^t &= \frac{t-1}{t}\tilde{L}^{t-1} + \frac{1}{t}L^t \\
        & = \frac{t-1}{t}(\frac{t-2}{t-1}\tilde{L}^{t-2} +\frac{1}{t-1}L^{t-1}) + \frac{1}{t}L^t \\
        & = \frac{t-2}{t}\tilde{L}^{t-2} + \frac{1}{t}(L^{t-1} + L^t) \\
        & = \cdots \\
        & = \frac{1}{t} \sum_{i=1}^t L^i,
    \end{split}
\end{equation}
where $L^i = D^{-1/2}A^tD^{-1/2}$. Furthermore, according to Lemma \ref{lem:lem1-app}, 
\begin{equation}
    \begin{split}
        & \min_{( B_1 \cup \cdots \cup B_c)=V}\quad \sum_{t=1}^t \sum_{i=1}^c\frac{W^t(B_i, \overline{B_i})}{Vol^t(B_i)} \\
        & = \min_H \sum_{i=1}^t\text{tr}(H^T{D^t}^{-1/2}(D^t-A^t){D^t}^{-1/2}H) \\
        & =  \min_H \text{tr}( H^T(\sum_{i=1}^t{D^t}^{-1/2}(D^t-A^t){D^t}^{-1/2})H) \\ 
        & = \min_H \text{tr}(H^T \sum_{i=1}^t L^t H) \\
        & = \min_H \text{tr}(H^T \tilde{L}^t H)t,
    \end{split}
\end{equation}
where $H^TH =I$. By the Rayleigh-Ritz theorem \cite{lutkepohl1997handbook} it can be seen immediately that the solution of $\min_H \text{tr}(H^T \tilde{L}^tH)t$ is given by the $H$ which is the eigenvector corresponding to the $c$ largest eigenvalue of $\tilde{L}^t$. Thus we prove that FedCut aims to solve the \textit{c-partition Ncut} in Eq. \eqref{eq:loss-Spatial-Temporal-app} for Spatial-Temporal graph among $t$ iterations. 
\end{proof}

\subsection{Proof of Theorem 1}
% \begin{assumption} \label{assum:connected-app}
% Connectivity. $G_R$ is connected i.e., averaged shortest path length (ASPL) $< \infty$.
% \end{assumption}
% \begin{rmk}
% We observe this connected phenomenon in Sect. \ref{sec:appendixc} and this assumption is widely used in \cite{data2021byzantine, li2019communication}.
% \end{rmk}

% \begin{assumption} \label{assum: Bounded gradient dissimilarity-app}
% We define benign mean model update across
% clients to be $\bar{\bg} = \frac{1}{|V_R|}\sum_{i \in V_R} \bg_{i}$
% where $V_R$ is the set of benign clients, hence the variance across client updates as $\EE||\bg_i-\bar{\bg}|| \leq \ka$
% across all rounds of training.
% \end{assumption}

\begin{assumption} \label{assum: diff-cluster-app}
For malicious updates $\bg_b$ provided that $||\bg_b- \nabla F||> \ka$, the difference between the mean of benign updates and colluders' updates has at least $C\ka$ distance,  where $C$ is a large constant, i.e.,
\begin{equation*}
    \| \bg_b -  \nabla F \| > C\ka.
\end{equation*}

%the updates of clients in same parties has at most $C\ka$ distance (C>2).
\end{assumption}

% \begin{rmk}
% 1) We also observe the phenomenon of Assumption \ref{assum: diff-cluster-app} in Sect. \ref{sec:appendixc}. Moreover, for some parties which distance from benign clients is smaller than $C\ka$, even if they are misclustered, the $||\hat{\bg} - \bg||$ \footnote{$||\cdot||$ represent $\ell_2$ norm.} is bounded by $\calO(C^2\alpha^2\ka)$ by induction. \\
% 2) If we do NCut for each iteration separately, here $C$ is selected by smallest values over all iterations so that $C$ is small. FedCut averages the normalized adjacency matrix for multiple iterations so that $C$ is average case instead of minimum case. Consequently, $C$ in FedCut is larger than $C$ in NCut, which result in being more stable than NCut.
% \end{rmk}

\begin{lem}[\cite{ng2001spectral}]  \label{lem:spectral-clustering-app}
Let adjacency matrix A's off-diagonal blocks $A^{ij}$, $i \neq j$ be zero (block diagonal matrix). Also assume that each cluster is connected. Then there exists $c$ orthogonal vectors $r_1,\cdots, r_c$
($r_i^T r_j = 1$ if i = j, 0 otherwise) so that top $c$ eigenvectors ($Y \in \mathbb{R}^{K \times c}$) of normalized adjacency matrix satisfy 
\begin{equation}
    y_j^{(i)} = r_i,
\end{equation}
for all $i = 1,\cdots, c$, $j = 1, \cdots, n_i$, where $y_j^{(i)}$ represents $j_{th}$ rows in $i_{th}$ clusters of $Y$. In other words, there are $c$ mutually orthogonal points on the surface of the unit c-sphere around which Y's rows will cluster. Moreover, these clusters correspond exactly to the true clustering of the original data. 
\end{lem}

\begin{lem}[Theorem 4.14 in \cite{ostrovsky2013effectiveness}]  \label{lem:kmeans-app}
Defined $Y$ same as Lemma \ref{lem:spectral-clustering-app}, the Kmeans algorithm that clusters $\tilde{Y} = Y + E_1$ into c groups returns an optimal solution of cost at most $\frac{1-\epsilon^2}{1-32\epsilon^2}\epsilon^2$ with probability $1- \calO(\sqrt{\epsilon})$, where $\epsilon = ||E_1||$
\end{lem}

\begin{rmk}
1) Lemma \ref{lem:spectral-clustering-app} illustrates the clustering accuracy of NCut is 100\% in clustering top c eigenvectors when adjacency matrix is the block diagonal matrix.\\
2) Lemma \ref{lem:kmeans-app} shows Kmeans clustering for eigenvectors $Y$ still reach optimal clustering with a large probability when eigenvectors $Y$ has a small perturbation. 

\end{rmk}

\begin{thm} \label{thm:thm1-app}
Suppose an $0< \alpha < \frac{1}{2}$ fraction of clients are Byzantine attackers. If Assumption \ref{assum: Bounded gradient dissimilarity-app} and \ref{assum: diff-cluster-app} holds, we can find the estimate of $\hat{\bg}$ of $\bar{\bg}$ according to line 12 in Algo. 1 and  with the probability $1-\tilde{O}(\sqrt{Z_1})$, such that $||\hat{\bg} - \bar{\bg}|| \leq \calO(C\alpha \ka)$, where $Z_1 =  \frac{4K \sqrt{Z^{\frac{C^2}{4}}}}{\delta - K \sqrt{2Z^{\frac{C^2}{4}}}}, Z = \text{exp}(-2 \ka^2/\sigma^2)$.
\end{thm}

\begin{proof}
Suppose the server receives first (K - q) correct gradients $\{\bg_1,\cdots, \bg_{K-q}\}$ with mean $\bar{\bg}$.
Therefore, we have
\begin{equation}
   ||\bg_i - \bg_j|| \leq ||\bg_i-\bar{\bg}|| + ||\bg_j - \bar{\bg}|| \leq 2\ka,
\end{equation}
for any $1 \leq i, j \leq K-q$. Thus we have $A_{ij} \geq \text{exp}(-2 \ka^2/\sigma^2) \triangleq Z (0<Z<1)$, for any $1 \leq i, j \leq K-q$. 

Moreover, for any updated gradients of clients in different party (such as $\bg_i, \bg_j$), according to Assumption \ref{assum: diff-cluster-app}, we have 
\begin{equation}
   ||\bg_i - \bg_j|| \geq C\ka.
\end{equation} 
Consequently, $A_{ij} < \text{exp}(-C^2 k^2/(2\sigma^2)) = Z^{\frac{C^2}{4}}$, for any i, j in different parties.

On one hand, if $A_{ij} = 0$, for any i, j in different parties, then the adjacency matrix $A$ is block diagonal matrix so we could apply Lemma \ref{lem:spectral-clustering-app} to get 100\% clustering accuracy for solving FedCut (c-partition NCut for Spatial-temporal Graph) in $t_{th}$ iteration, i.e., server would select clusters ($\calI$) with largest number of clusters. Since $\alpha<0.5$, the clusters with largest number must include all benign updates and malicious updates that $||\bg_b-\bar{\bg}||<\ka$. As a result,  

% and some Byzantine attackers $\calU_1 = \{\bg_i| ||\bg_i - \bar{\bg}||\leq(C+1)\ka$. As a result, 
\begin{equation}
    ||\hat{\bg} - \bar{\bg}|| \leq  \frac{q||\bg_b-\bar{\bg}||}{K}  \leq \calO(\frac{q\ka}{K})
\end{equation}.

On the other hand, in general case that adjacency matrix $\tilde{A}_{ij}>0$, for any i, j in different parties, set $\tilde{A} = A + E$, where $A$ is block diagonal matrix and $E$ is perturbation matrix.  For the large $C$,  $ Z^{\frac{C^2}{4}}$ is small enough. Consequently, we have
\begin{equation}
\begin{split}
     ||Y -\tilde{Y}|| &\leq \frac{4||E||_2}{\delta - \sqrt{2}||E||_2} \\
     & \leq \frac{4K \sqrt{\alpha(1-\alpha)Z^{\frac{C^2}{4}}}}{\delta - K \sqrt{2\alpha(1-\alpha)Z^{\frac{C^2}{4}}}} \triangleq Z_1.
\end{split}
\end{equation}
The first inequality is according to Lemma \ref{lem:matrix-perturb-app} and second is due to $||E||_2 \leq K \sqrt{\alpha(1-\alpha)} \sqrt{Z^{\frac{C^2}{4}}}$.
according to Lemma \ref{lem:kmeans-app}, we could obtain the optimal clustering  with probability $1-\tilde{O}(\sqrt{Z_1})$. Similarly, we have
\begin{equation}
    ||\hat{\bg} - \bar{\bg}|| \leq  \frac{q||\bg_b-\bar{\bg}||}{K}  \leq \calO(\frac{q\ka}{K})
\end{equation}
for line 12 in Algo. 1 with probability $1-\tilde{O}(\sqrt{Z_1})$, where $Z_1 =  \frac{4K \sqrt{\alpha(1-\alpha)Z^{\frac{C^2}{4}}}}{\delta - K \sqrt{2\alpha(1-\alpha)Z^{\frac{C^2}{4}}}}, Z = \text{exp}(-2 ka^2/\sigma^2)$.

\end{proof}

\subsection{Proof for Theorem 2}
Following definitions of federated learning introduced in \cite{yang2019federated,mcmahan2017communication}, we consider a \textit{horizontal federated learning} setting consisting of one server and $K$ clients. We assume $K$ clients have their local dataset $\calD_i, i=1\cdots K$. In each step, $i_{th}$ client minimizes the local risk function as $\min_{\bw_i} F_{i}(\bw_i) =\min_{\bw_i} \frac{1}{|\calD_i|} \sum_{j=1}^{|\calD_i|} F_i(\bw, \bx_{i,j})$, where $\bx_{i,j} \in \calD_i, j \in [|\calD_i|]$. Then clients send the updates $\bg_i = \nabla F_i(\bw_i)$ to the server and server aggregate the gradients to update the global model $\bw$, finally the server distribute the $\bw$ to all clients (see details in Algo. 1). Specifically,  for $t+1$ iteration, server update the global model $\bw_{t+1}$ as:
\begin{equation}
    \bw^{t+1}  = \bw^{t} - \eta \hat{\bg}^t
\end{equation}

\begin{assumption} \label{assum:bound local variance-app}
The stochastic gradients sampled from any local dataset have uniformly bounded variance over $\calD_i$ for all clients, i.e., there exists a
finite $\sigma_0$, such that for all $x \in \calD_i , i\in [K]$, we have
 \begin{equation}
     \EE_{j}||[\nabla F_i(\bw, x_{i,j})) - \nabla F_i(\bw)||^2 \leq \sigma_0^2.
 \end{equation}
 %, where $F_i(\bw) = \sum_{j=1}^{|D_i|} f(\bw, x_{i,j})$.
 \end{assumption}

\begin{rmk}
% 1) It is not hard to find $||\EE_{j\in \calD_{Mini}}[\nabla f_i(\bw, x_{i,j})) - \nabla F_i(\bw)|| \leq \sigma_0^2$, where $\calD_{Mini}$ is one mini- batch with  batch size $b$. \\
The difference between Assumption \ref{assum: Bounded gradient dissimilarity-app} and \ref{assum:bound local variance-app} is that the former bounds the variance across gradient estimates within the same client while the latter bounds the variance between model updates across clients. 

\end{rmk}

\begin{assumption} \label{assu:l-smooth-app}
We assume that F(x) is L-smooth and has $\mu$-strong convex.
\end{assumption}

\begin{lem} [\cite{boyd2004convex}] \label{lem:l-smmoth}
If $F$ is L-smooth, then for all $\bw^*$ (optimal solution w.r.t $F$), then:
\begin{equation}
    \frac{1}{2L}||\nabla F(x)||^2 \leq F(\bw) -F(\bw^*) \leq \frac{L}{2}||\bw - \bw^*||^2
\end{equation}
\end{lem}

\begin{thm}\label{thm:thm2-app}
Suppose an $0< \alpha<\frac{1}{2}$ fraction of clients are corrupted.
For a global objective function $F:R^d\to R$, the server distributes a sequence of iterates $\{\bw^t : t\in[0:T]\}$ (see Algo. 1) when run with a fixed step-size $\eta< \min\{\frac{1}{4L},\frac{1}{\mu} \}$. If Assumption \ref{assum: Bounded gradient dissimilarity-app},  \ref{assum: diff-cluster-app}, \ref{assum:bound local variance-app} and \ref{assu:l-smooth-app} holds, the sequence of average iterates $\{\bw^t: t\in[0:T]\}$
satisfy the following convergence guarantees\footnote{There are some modifications for theorem 2 compared to main text, please refer to Theorem 2 here instead of main text}: 
\begin{equation}
    \left\|\bw^{T} - \bw^*\right\|^2 \leq (1-\frac{C_1\mu}{L})^T\left\|\bw^{0} - \bw^*\right\|^2 + \frac{\Gamma}{\mu^2},
\end{equation}
where $\Gamma = \calO(\sigma_0^2 + \ka^2 + C^2 \alpha^2 \ka^2)$, $b$ is batch size and $\bw^*$ is the global optimal weights in federated learning.
\end{thm}

\begin{proof}
Note that $\bw^t = \frac{1}{K}\sum_{i=1}^K \bw_i^t$, so we have:
\begin{equation}
    \begin{split}
        \bw^{t+1}  &= \bw^{t} - \eta \hat{\bg^t} \\
                   &= \bw^{t} - \eta \nabla F(\bw^t) + \eta \nabla F(\bw^t) - \eta  \frac{1}{K}\sum_{i=1}^K\nabla F_i(\bw^t)   \\
                   & \quad +\eta  \frac{1}{K}\sum_{i=1}^K\nabla F_i(\bw^t) - \eta \hat{\bg^t} \\
                   &=\underbrace{[\bw^{t} - \eta \nabla F(\bw^t)]}_{=: U_1} + \eta \underbrace{[ \frac{1}{K}\sum_{i=1}^K(\nabla F(\bw^t) -\nabla F_i(\bw_i^t))]}_{=: U_2} \\
                   & \quad +  \eta  \underbrace{[\frac{1}{K}\sum_{i=1}^K\nabla F_i(\bw_i^t) - \hat{\bg^t}]}_{=: U_3}\\
    \end{split}
\end{equation}
Therefore, 
\begin{equation}
    \bw^{t+1} - \bw^{*} = [U_1 - \bw^*] +\eta U_2 + \eta U_3
\end{equation}
Taking norm on both sides and then squaring:
\begin{equation} \label{eq:norm-wt}
\begin{split}
    ||\bw^{t+1} - \bw^{*}||^2& = ||U_1 - \bw^*||^2 + \eta^2||U_2 + U_3||^2 \\
    & \quad + 2\eta<U_1 - \bw^*, U_2 + U_3>
\end{split}
\end{equation}
Taking use of $2<\mathbf{a}, \mathbf{b}> \leq ||\mathbf{a}||^2 + ||\mathbf{b}^2||$, we have 
\begin{equation} \label{eq:square-inequality}
\begin{split}
    &2\eta<U_1 - \bw^*, U_2 + U_3> \\
    & = 2<\sqrt{\frac{\eta \mu}{2}}(U_1 - \bw^*), \sqrt{\frac{\eta}{\mu}}(U_2 + U_3)> \\
    & \leq \frac{\eta \mu}{2}||U_1 - \bw^*||^2 +\frac{2\eta}{\mu}||U_2 +U_3||^2 
\end{split}
\end{equation}
Combining Eq. \eqref{eq:norm-wt} and \eqref{eq:square-inequality}, we have:
\begin{equation}
    \begin{split}
        ||\bw^{t+1} - \bw^{*}||^2 &\leq (1+\frac{\eta\mu}{2})||U_1 - \bw^*||^2 \\
        & + (\eta^2 + \frac{2\eta}{\mu})||U_2+U_3||^2 \\
        &\leq (1+\frac{\eta\mu}{2})||U_1 - \bw^*||^2  \\
        &  + (2\eta^2 + \frac{4\eta}{\mu})||U_2||^2 + (2\eta^2 + \frac{4\eta}{\mu})||U_3||^2
    \end{split}
\end{equation}
Substituting the values of $U_1$, $U_2$ and $U_3$ into Eq. \eqref{eq:norm-wt}, taking expectation w.r.t. the stochastic sampling of gradients by clients, we could get:
\begin{equation} \label{eq: main-result}
    \begin{split}
        ||\bw^{t+1}& - \bw^{*}||^2 \leq (1+\frac{\eta\mu}{2}) ||\bw^{t} - \eta \nabla F(\bw^t)-w^*||^2 \\
       & + (2\eta^2 + \frac{4\eta}{\mu}) ||\frac{1}{K}\sum_{i=1}^K(\nabla F(\bw^t) -\nabla F_i(\bw_i^t))||^2 \\
       &+ (2\eta^2 + \frac{4\eta}{\mu})||\frac{1}{K}\sum_{i=1}^K\nabla F_i(\bw_i^t) - \eta \hat{\bg^t}||^2
    \end{split}
\end{equation}
Now we bound each of the three terms on the RHS of Eq. \eqref{eq: main-result} separately as follows: \\
\textbf{Firstly,}  since $F$ is $\mu$ strong convex, $<\bw^*-\bw^t,\nabla F(\bw^t)> \leq F(\bw^*)-F(\bw^t) -\frac{\mu}{2}||\bw^*-\bw^t||^2$. Also, according to Lemma \ref{lem:l-smmoth}, we have:
\begin{equation}
    \begin{split}
         &||\bw^{t} - \eta \nabla F(\bw^t)-\bw^*||^2 \\
         &=  ||\bw^t-\bw^*||^2 + \eta^2||\nabla F(\bw^t)||^2  \\
         &+ 2\eta<\bw^*-\bw^t,\nabla F(\bw^t)> \\
         & \leq (1-\mu \eta)||\bw^t-\bw^*||^2 + (\eta^2-\frac{\eta}{L})||\nabla F(\bw^t)||^2  \\
         & \leq  (1-\mu \eta)||\bw^t-\bw^*||^2,
         \end{split}
\end{equation}
where the last inequality is due to $\eta<\frac{1}{4L}$.
\textbf{Secondly,}
\begin{equation} \label{eq:eq32}
\begin{split}
    ||\frac{1}{K}&\sum_{i=1}^K(\nabla F(\bw^t) -\nabla F_i(\bw_i^t))||^2 \\
   &\leq \frac{1}{K^2}\sum_{i=1}^K||\nabla F(\bw^t) -\nabla F_i(\bw_i^t))||^2 \\
   &\leq \frac{2}{K^2}\sum_{i=1}^K[||\nabla F(\bw^t) -\nabla F(\bw_i^t))||^2 \\
   &+ ||\nabla F(\bw_i^t) -\nabla F_i(\bw_i^t))||^2] \\
   &\leq  \frac{2L^2}{K^2}\sum_{i=1}^K||\bw^t-\bw_i^t||^2 +  2\ka^2 
\end{split}
\end{equation}
Note that 
\begin{equation}  \label{eq:eq33}
    ||\bw^t-\bw_i^t||^2 =|| \frac{1}{K} \sum_{k=1}^K(\bw_k^t-\bw_i^t)||^2 
     \leq \frac{1}{K}\sum_{k=1}^K||\bw_k^t - \bw_i^t||^2 
\end{equation}
Moreover, for the batch data $\mathcal{S}$, we have
\begin{equation} \label{eq:eq34}
\begin{split}
    & ||\bw_k^t - \bw_i^t||^2 = \frac{\eta^2}{|\mathcal{S}|} ||\sum_j(\nabla F_i(\bw_i^t, x_{ij}) - \nabla F_k(\bw_k^t, x_{kj}))||^2\\
    & \leq  \frac{\eta^2}{|\mathcal{S}|} \sum_j ||\nabla F_i(\bw_i^t, x_{ij}) - \nabla F_k(\bw_k^t, x_{kj})||^2 \\
    & \leq \frac{5\eta^2}{|\mathcal{S}|} \sum_j [(||\nabla F_i(\bw_i^t, x_{ij}) - \nabla F_i(\bw_i^t)||^2) \\
    &\quad + (||\nabla F_k(\bw_k^t) -\nabla F_k(\bw_k^t, x_{kj}) ||^2) \\
    & \quad + (|| \nabla F_i(\bw_i^t)) - \nabla F(\bw_i^t) ||^2)  + (||\nabla F(\bw_i^t) - \nabla F(\bw_k^t)) ||^2)  \\
    &\quad + (||\nabla F(\bw_k^t)) - \nabla F_k(\bw_k^t)) ||^2)] \\
    &\leq 5\eta^2[\sigma_0^2 + \sigma_0^2 + 4\ka^2 + L^2||\bw_i^t-\bw_k^t||^2 + 4\ka^2] \\
\end{split}
\end{equation}
The second inequality is due to the Cauchy–Schwarz inequality. The Third inequality is because of Lemma \ref{assum:bound local variance-app}, Assumption \ref{assum: Bounded gradient dissimilarity-app} and \ref{assu:l-smooth-app}.
Then we transfer $||\bw_i^t-\bw_k^t||^2$ to RHS of Eq. \eqref{eq:eq34}, we derive:
\begin{equation}
    (1-5\eta^2L^2) ||\bw_i^t-\bw_k^t||^2\leq 5( 2\sigma_0^2 + 8\ka^2)
\end{equation}
Thus 
\begin{equation} \label{eq:eq36}
     ||\bw_i^t-\bw_k^t||^2 \leq \frac{5( 2\sigma_0^2 + 8\ka^2)}{(1-5\eta^2L^2) }
\end{equation}
Combining Eq. \eqref{eq:eq32}, \eqref{eq:eq33} and \eqref{eq:eq36}, we obtain
\begin{equation}
\begin{split}
    ||\frac{1}{K}&\sum_{i=1}^K(\nabla F(\bw^t) -\nabla F_i(\bw_i^t))||^2  \\
    & \leq \frac{2L^2}{K^2} \frac{5K^2( 2\sigma_0^2 + 8\ka^2)}{(1-5\eta^2)} + 2\ka^2 \\
    & = \frac{10L^2(2\sigma_0^2 + 8\ka^2)}{(1-5\eta^2L^2)} + 2\ka^2
    \end{split}
\end{equation}\\
\textbf{Thirdly,} according to Theorem \ref{thm:thm1-app}, we obtain
\begin{equation}
    ||\frac{1}{K}\sum_{i=1}^K\nabla F_i(\bw_i^t) - \hat{\bg^t}|| \leq  \calO(\alpha^2 \ka^2)
\end{equation}\\
\textbf{Finally, } Noted that $1-5\eta^2L^2 > \frac{1}{2}$ and $(1+\frac{\eta\mu}{2})(1-\mu\eta)< 1-\frac{\mu\eta}{2}$ when $\eta < \frac{1}{4L}$. And if $\eta + \frac{2}{\mu} < \frac{3}{\mu} $, i.e., $\eta< \frac{1}{\mu}$, we bounded Eq.\eqref{eq: main-result} as:
\begin{equation} \label{eq:eq39}
    \begin{split}
           &  ||\bw^{t+1} - \bw^{*}||^2  \leq (1+\frac{\eta\mu}{2})(1-\mu\eta) ||\bw^t-\bw^*||^2 \\
       & + (2\eta^2 + \frac{4\eta}{\mu}) (\frac{10L^2(2\sigma_0^2 + 8\ka^2)}{(1-5\eta^2L^2)} + 2\ka^2) 
       + (2\eta^2 + \frac{4\eta}{\mu}) \calO( \alpha^2 \ka^2)  \\
       & \leq (1-\frac{\mu\eta}{2}) ||\bw^t-\bw^*||^2 
        + (2\eta^2 + \frac{4\eta}{\mu}) (\frac{10L^2(2\sigma_0^2 + 8\ka^2)}{(1-5\eta^2L^2)} + 2\ka^2)  \\
        & + (2\eta^2 + \frac{4\eta}{\mu}) \calO( \alpha^2 \ka^2) \\
       &\leq (1-\frac{\mu\eta}{2}) ||\bw^t-\bw^*||^2 
        +  \frac{240\eta}{ \mu}(\sigma_0^2 + 4\ka^2) + \frac{12\ka^2 \eta}{\mu} \\
       & + (\frac{6\eta}{ \mu}) \calO( \alpha^2 \ka^2) 
    \end{split}
\end{equation}
Therefore, we could use Eq. \eqref{eq:eq39} by induction for $t=1,\cdots, T-1$ to obtain
\begin{equation}
\begin{split}
        &\left\|\bw^{T} - \bw^*\right\|^2  \\ &\leq (1-\frac{\eta \mu}{2})^T\left\|\bw^{0} - \bw^*\right\|^2 + [\frac{240\eta}{ \mu}(\sigma_0^2 + 4\ka^2) + \frac{12\ka^2 \eta}{\mu}
       \\
       & + (\frac{6\eta}{ \mu}) \calO(C^2 \alpha^2 \ka^2)]/(\frac{\mu \eta}{2}) \\
       & = (1-\frac{\eta \mu}{2})^T\left\|\bw^{0} - \bw^*\right\|^2 + \calO(\sigma_0^2 + \ka^2 + \alpha^2 \ka^2)/\mu^2 \\
       & \leq  (1-\frac{C_1 \mu}{L})^T\left\|\bw^{0} - \bw^*\right\|^2 + \frac{1}{\mu^2}\Gamma
\end{split}
\end{equation}
where $\Gamma = \calO(\sigma_0^2 + \ka^2 + \alpha^2 \ka^2)$ and $C_1$ is a constant.
\end{proof}
\end{appendices}

% \bibliographystyle{plain}
% \bibliography{reference/attack, reference/defense, reference/background, reference/spectral_clustering, reference/byzantine-other-field}

\end{document}